\documentclass[sn-mathphys]{sn-jnl}
\jyear{2021}

\theoremstyle{thmstyleone}
\newtheorem{theorem}{Theorem}[section]
\theoremstyle{thmstyletwo}%

\theoremstyle{thmstylethree}%

\usepackage{amssymb,amsmath,epsfig,graphics,psfrag,graphicx,color,subfigure,float,mathtools}

\numberwithin{equation}{section}
\numberwithin{figure}{section}
\numberwithin{table}{section}
\newcommand\cero{\boldsymbol{0}}

\newcommand\bH{\mathbf{H}}

\newcommand\bI{\mathbf{I}}

\newcommand\bV{\mathbf{V}}
\newcommand\bW{\mathbf{\Lambda}}

\newcommand\beps{\boldsymbol{\varepsilon}}

\newcommand\nn{\boldsymbol{n}}

\newcommand\bsigma{\boldsymbol{\sigma}}

\newcommand{\bsym}[1]{\boldsymbol{#1}}

\newcommand\bu{\boldsymbol{u}}
\newcommand\bv{\boldsymbol{v}}

\newcommand\bx{\boldsymbol{x}}

\newcommand\bnu{\boldsymbol{\nu}}

\newcommand\RR{\mathbb{R}}
\newcommand\cT{\mathcal{T}}

\newcommand{\norm}[1]{\left\|#1\right\|}

\DeclarePairedDelimiter\snorm{\lvert}{\rvert}
\DeclarePairedDelimiter\abs{\lvert}{\rvert}
\renewcommand{\vec}[1]{\boldsymbol{#1}} 
\newcommand{\pt}{\partial_t\,}
\newcommand{\Dt}{\Delta t}

\renewcommand\bdiv{\mathop{\mathbf{div}}\nolimits}
\newcommand\bcurl{\mathop{\mathbf{curl}}\nolimits}
\newcommand\vdiv{\mathop{\mathrm{div}}\nolimits}

\newcommand\bomega{\boldsymbol{\omega}}
\newcommand\btheta{\boldsymbol{\theta}}

\renewenvironment{proof}{\noindent{\it Proof.}}{\hfill$\square$}

\allowdisplaybreaks

\begin{document}

\title[Mechanochemical models for embryonic epithelial tissue]{Mechanochemical models for 
calcium waves in embryonic epithelia}
\author[1]{\fnm{Katerina} \sur{Kaouri}}\email{kaourik@cardiff.ac.uk}
\author[2]{\fnm{Paul E.} \sur{M\'endez}}\email{paul.mendez01@epn.edu.ec}
\author*[3]{\fnm{Ricardo} \sur{Ruiz-Baier}}\email{ricardo.ruizbaier@monash.edu}

\affil[1]{\orgdiv{School of Mathematics}, \orgname{Cardiff University}, \orgaddress{\street{Senghennydd Road}, \city{Cardiff}, \postcode{CF24 4AG}, \country{UK}}}

\affil[2]{\orgdiv{Research Centre on Mathematical Modelling (MODEMAT)}, \orgname{Escuela Polit\'ecnica Nacional}, \orgaddress{\city{Quito},  \country{Ecuador}}}

\affil[3]{\orgdiv{School of Mathematics}, \orgname{Monash University}, \orgaddress{\street{9 Rainforest Walk}, \city{Melbourne}, \postcode{3800}, \state{Victoria}, \country{Australia}}}


\abstract{
In embryogenesis, epithelial cells, acting as individual entities or as coordinated aggregates in a tissue, exhibit strong coupling between chemical signalling and mechanical responses to internally or externally applied stresses. Intercellular communication in combination with such coordination of morphogenetic movements can lead to drastic modifications in the calcium distribution in the cells.  In this paper we extend the recent mechanochemical model in [K. Kaouri, P.K. Maini, P.A. Skourides, N. Christodoulou, S.J. Chapman. J. Math. Biol., 78 (2019) 2059--2092], for an epithelial continuum in one dimension, to a more realistic multi-dimensional case. The resulting parametrised governing equations consist of an advection-diffusion-reaction system for calcium signalling coupled with active-stress linear viscoelasticity and equipped with pure Neumann boundary conditions. We implement a mixed finite element method for the simulation of this complex multiphysics problem. Special care is taken in the treatment of the stress-free boundary conditions for the viscoelasticity in order to eliminate rigid motions from the space of admissible displacements. The stability and solvability of the continuous weak formulation is shown using fixed-point theory. We investigate numerically the solutions of this system and show that solitary waves and periodic wavetrains of calcium propagate through the embryonic epithelial sheet. We analyse the bifurcations of the system guided by the bifurcation analysis of the one-dimensional model. We also demonstrate the nucleation of calcium sparks into synchronous calcium waves coupled with contraction. This coupled model can be employed to gain insights into recent experimental observations in the context of embryogenesis, but also in  other biological systems such as cancer cells, wound healing,  keratinocytes, or white blood cells. 
}



\date{Updated: \today}
\maketitle 

\section{Introduction}
Embryogenesis is a remarkable example of a complex process where different sub-mechanisms involving mechanical and chemical effects closely interact in a self-organised manner, forming complex spatio-temporal patterns. The coupling between rearrangement of tissue, cell migration, active cell contraction to diffusion of morphogens or signalling molecules has been proposed and studied in earlier works \cite{murray03,murray84}. It is well known that a variety of responses in the cell are driven by the transduction of mechanical stimulation into chemical signals such as calcium oscillations and waves \cite{sanderson90}. In turn, localisation of stresses or strains within the cells can generate patterns of motion distribution in tissue by changing their displacement magnitude, direction and velocity \cite{guiu15,lecuit07}. Recent experimental evidence \cite{christo15, suzuki17,narciso} shows that tissue mechanics are actively coupled to chemical patterns during development. 

Hence, more modelling and analysis has recently appeared in order to elucidate the many open questions that directly impact the healthy evolution of an embryo \cite{suzuki17, brinkmann18,kaouri19, narciso}. A variety of models have described general interactions between chemical species concentration and mechanics; these include descriptions where stress is triggered by chemical signalling \cite{murray84}, or mainly by migration \cite{moreo}, but we are also interested more generally in the mechanochemical feedback due to pressure and velocity of the underlying embryonic tissue. In all cases,  including the coupling with continuum mechanics significantly affects the propagation of the chemical signals.

Here, we develop a multi-dimensional extension of the recent mechanochemical model proposed in \cite{kaouri19} which describes the interplay of calcium signalling with the mechanics of embryonic epithelial tissue during apical constriction, an active deformation process. Inspired by the recent, interesting experimental observations in \cite{christo15, suzuki17} where increasing tension in the contracting cells yields further calcium release and this, in turn, elicits contractions which are sensed as mechanical stimuli by the neighbouring cells. The model we propose here takes a step further in elucidating this important mechanochemical coupling.

Mechanical properties of different cell types indicate diverse behaviour, including elastic 
\cite{deoliveira19,dillon99,ohayon05}, poroelastic \cite{deoliveira20,moee13,rads14,recho19}, or nonlinear and nonlocal characteristics \cite{merker16,wy12} but more predominantly, viscoelastic effects \cite{allena13,beysens00,bausch98,ghosh07,kim14,preziozi10,yamada00}.  The specific constitutive rheological model to adopt in a tissue depends on the characteristics of each constituent cell, on the properties inherent to distinct biological states, on the nature and intensity of the stresses and strains that are to be applied, and on the spatio-temporal scales involved. Here, we restrict the description to the regime of small strains and model the cell and tissue as a modification of the simple Kelvin-Voigt viscoelastic solid (with one elastic spring and two viscous dashpots), where only after the initial stress has vanished, the material goes back to its original state. These linear viscoelastic materials are completely defined by the stiffness and viscosity, which can be determined using diverse measuring approaches such as pipette suction, optical laser tweezers, microrheology tools, particle tracking, or even contact-free techniques \cite{nguyen20}. In the present mechanochemical model, we assume that the viscoelastic stress has an active tension component which depends on the concentration of calcium, following the formulation in \cite{murray84, bane11, kaouri19}.  

We adopt the following fundamental assumptions: a) the equilibrium of forces in the system is established by a quasi-static balance of linear momentum using displacements and hydrostatic pressure (the so-called Herrmann formulation \cite{herrmann} where the introduction of solid pressure ensures that the system is robust with respect to the modulus of dilation of the solid); b) the spatio-temporal dynamics of calcium concentration  and the percentage of IP$_3$ receptors (IPR) on the Endoplasmic Reticulum (ER) that have not been inactivated are governed by an advection-reaction-diffusion system; and c) mechanochemical coupling is modelled directly in the viscoelastic stress through an active stress Hill function that depends on calcium and through the modification of the reaction kinetics by volume change. The two-way coupling mechanism we adopt follows the model structure used in \cite{murray,murray03,neville06,ruiz14, kaouri19}.  

Finding closed-form solutions to this inherently highly nonlinear and multidimensional problem is only possible in very restricted scenarios and simplified settings. We, hence, resort to solving the governing equations numerically. The numerical framework undertaken here uses the method of lines, adopting a backward Euler scheme for the discretisation in time and a primal-mixed formulation  consisting of mixed approximations for the viscoelasticity in terms of displacement and pressure using the classical Taylor-Hood and MINI-elements \cite{arnold84}, and piecewise quadratic or piecewise linear approximations for the calcium concentration and the fraction of non-inactivated IPR. Methods of this type are known to perform well in a variety of scenarios. As we assume that the cell or tissue is not attached to any supporting structure, we consider pure traction boundary conditions. In this case the viscoelasticity problem is not well-posed unless we incorporate a constraint to eliminate the rigid motions from the set of admissible solutions. To achieve this we employ additional vector Lagrange multipliers for the coupled problem following \cite{kutcha18} (see also \cite{barnafi21}). The overall problem is treated as a monolithic system, so at each time iteration we solve a set of nonlinear equations with the Newton-Raphson method and the tangent system at each step is inverted with a direct solver. 

We have organised the contents of the paper as follows. In Section~\ref{sec:model} we lay out the details of the mechanochemical 1D model in \cite{kaouri19} and construct its multidimensional extension. We also explain the coupling mechanisms and then perform an appropriate nondimensionalisation. In the same section we also state the weak form of the governing equations and introduce a suitable finite element discretisation. In Section~\ref{sec:wellp} we address the continuous dependence on data of the weak formulation, as well as the existence of weak solutions using Brouwer's fixed-point theory.   A number of illustrative numerical computations are then presented in Section~\ref{sec:results}, where we also perform a simple verification of convergence. We specifically explore different regimes of wave propagation of practical interest, such as solitary waves and periodic wavetrains of calcium and the the nucleation of calcium sparks into calcium waves. We conclude in  Section~\ref{sec:concl} with a summary of our findings and a   discussion on the limitations of the model, addressing also possible extensions and future directions. 

\section{Model description and weak formulations}\label{sec:model}
\subsection{Coupling calcium signalling with mechanics}
We assume that the cell/tissue can be macroscopically regarded as a linear, viscoelastic material of Kelvin-Voigt type, occupying the 
spatial domain $\Omega\subset\RR^d$ with $d=2$ or $d=3$. 

Following, e.g., \cite{moreo,kaouri19}, and assuming that gravitational forces and inertial effects are negligible, 
one seeks for each time $t\in (0,t_{\mathrm{final}}]$, 
 the displacements of the tissue, $\bu(t):\Omega\to \RR^d$, and the dilation $\theta(t):\Omega\to\RR^d$, 
 such that 
 \begin{subequations}
\begin{align}
\theta & =  \vdiv \bu, \label{eq:p}\\ 
\bsigma & =\frac{E}{1+\nu} \bigl(\beps(\bu) + \frac{\nu}{1-2\nu}\theta \bI\bigr) + \tilde\alpha_1 \partial_t \beps(\bu) + \tilde\alpha_2 \partial_t \theta \bI - T(c)\bI, \label{eq:sigma}\\ 
-\bdiv\bsigma & = \cero, \label{eq:momentum}\\
& \qquad \qquad \qquad \qquad \text{in $\Omega\times(0,t_{\mathrm{final}}]$}, \nonumber 
\end{align}\end{subequations}
where $\bsigma$ is the Cauchy stress tensor, $\beps(\bu)=\frac{1}{2}(\nabla \bu+\nabla \bu^T)$ is the tensor of infinitesimal strains, and $E$, and $\nu$ are the Young modulus and Poisson ratio associated with the 
constitutive law of the material, respectively. Equation \eqref{eq:momentum} represents the force equilibrium for the system in the absence of inertia, whereas both \eqref{eq:p} and \eqref{eq:sigma} are constitutive equations describing properties of the viscoelastic material. 
The parameters $\tilde\alpha_i$ in the constitutive relation \eqref{eq:sigma} are the shear viscosity and the bulk viscosity related to the total Cauchy stress exerted in the cell/tissue, that characterises the viscoelastic response to deformations. 
In addition, the last term in \eqref{eq:sigma}, $T(c)$ describes the active contraction stress which is dependent on the calcium concentration, $c$.
There are various ways to model this term. As in \cite{kaouri19} we assume a Hill function which saturates to a constant value $T_0$ for high values of $c$, in line with experimental observations from \cite{christo15}. Therefore,
\begin{equation}\label{eq:defT}
T(c)=T_0\frac{\kappa c^n}{1+\kappa c^n},
\end{equation}
where $n$ is a positive integer and $\kappa>0$. We, thus, assume that the calcium concentration controls entirely the material's motion 
by regulating the active tension and therefore the generation of stress. Such control is indeed an activation and not an inhibition of active tension, so $T_0$ is assumed positive.

As in \cite{moreo}, we do not incorporate external body loads or restoring displacement-dependent body forces on the right-hand side of \eqref{eq:momentum}, since such terms are only relevant to substrate-on-substrate or in tissue-on-substrate configurations.

Focusing on the spatiotemporal behaviour of calcium, we denote its  concentration 
by $c$ and $h$ represents the percentage of non-inactivated IPR on the ER. The  model is a generalisation of the recent model in \cite{kaouri19}, which employs the calcium dynamics of the well-known model in \cite{atri93}. In dimensional form of the model is written as follows
\begin{subequations}
 \begin{align}
	\label{eq:cDIM}
  \partial_t c + \partial_t \bu \cdot \nabla c -  
	D \nabla^2 c 
&= J_{\rm ER}-J_{\rm pump}
+J_{\rm SSCC} 
 & \quad \text{in } \Omega\times(0,t_{\mathrm{final}}],\\
	\label{eq:nDIM}
\tau_j\partial_t h + \partial_t \bu \cdot \nabla h &= \frac{k^2_2}{k^2_2+c^2} -h & \quad \text{in } \Omega\times(0,t_{\mathrm{final}}],
\end{align}\end{subequations}
where 
$$J_{\rm ER}=k_f\mu(p_3)h\frac{bk_1+c}{k_1+c},\,\,\,J_{\rm pump}=\frac{\gamma c}{k_{\gamma}+c},
\,\,\,J_{\rm SSCC}=S \vdiv\bu.$$ 
Here, $\vdiv\bu = \theta$ 
represents the dilation/compression of the apical surface area of the cell and $D$ is the calcium diffusion coefficient inherent to the tissue (or to the cytosol in the single-cell case), which is assumed 
positive and constant.  The fluxes in \eqref{eq:cDIM} are as follows: the term 
$J_{\rm ER}$ models the flux of calcium from the ER into the cytosol through the IPR, {$\mu(p_3)=p_3/(k_{\mu}+p_3)$} is the fraction of IPR that have bound IP$_3$ and is an increasing function of $p_3$, the  IP$_3$ concentration. The constant $k_f$ denotes the calcium flux when all IP$_3$ receptors are open and activated, and $b$ represents a basal current through the IPR; $J_{\rm pump}$ represents the calcium flux pumped out of the cytosol where $\gamma$ is the maximum rate of pumping of calcium from the cytosol and $k_{\gamma}$ is the calcium concentration at which the rate of pumping from the cytosol is at half-maximum. One could 
also include calcium fluxes leaking into the cytosol from outside the cell, but we leave those terms out, as they are assumed small.  
The term $J_{\rm SSCC}$ encodes the calcium flux due to the activated stretch-sensitive calcium channels (SSCC). These channels have been identified experimentally in recent years \cite{hamil}- they are on the cell membrane. 
SSC are activated when exposed to mechanical stimulation and they close either by relaxation of the mechanical force or by adaptation to the mechanical force \cite{hamil}. The constant $S$ represents the strength of stretch activation. (In \cite{kaouri19} the authors derive a relationship for $S$ as a function of the characteristics of a SSCC, under suitable assumptions.)

The inactivation of the  IPR by calcium acts on a slower timescale compared to activation \cite{atri93} and so the time constant for the dynamics of $h$, $\tau_j>1$ in the ODE  \eqref{eq:nDIM} governing the dynamics of  $h$. As in \cite{kaouri19} and \cite{atri93} we take $\tau_j=2\,$s. The values of all other calcium signalling parameters are also taken as in \cite{atri93}.

We nondimensionalise the set of governing equations \eqref{eq:p}-\eqref{eq:momentum} and \eqref{eq:cDIM}--\eqref{eq:nDIM} 
using 
\begin{gather*}
c= k_1\bar c, \quad t=\tau_j \bar t, \quad \bu=L\bar\bu, \quad \theta = \bar \theta, \quad \bx=L\bar \bx,\quad \alpha_1 = \frac{\tilde\alpha_1(1+\nu)}{E\tau_j}, \\
\alpha_2 = \frac{(1+\nu)(1-2\nu)\tilde\alpha_2}{E\nu\tau_j}, \quad \beta_1 = \frac{T_0 (1+\nu)}{E}, \quad \beta_2 = \frac{1}{\kappa k_1^n},\end{gather*}
where $L$ is the length of an embryonic epithelial tissue (or the maximal length of the cell if considering the single-cell case, see \cite{luu11}). In addition, instead of the dilation $\theta$, 
we use the rescaled pressure $p$, which leads to the following system, 
where we have dropped the bars in the dimensionless variables for notational convenience
 \begin{subequations}\label{eq:system-adim}\begin{align}
p + \frac{\nu}{(1-2\nu)} \vdiv \bu & = 0, \label{eq:p2}\\ 
-\bdiv\bsigma  = \cero ,  \  \text{with} \ 
\bsigma & =\beps(\bu) - p\bI  + \alpha_1\partial_t\beps(\bu) - \alpha_2 \partial_t p\bI - \beta_1\frac{c^n}{\beta_2+c^n}\mathbf{I},  
\label{eq:mom2} \\
	\label{eq:c}
  \partial_t c + \partial_t \bu \cdot \nabla c -  
	D^{\star} \nabla^2 c
&= \mu hK_1\frac{b+c}{1+c}-\frac{Gc}{K+c}+ \lambda \vdiv \bu,\\
	\label{eq:n}
\partial_t h + \partial_t \bu \cdot \nabla h &=  \frac{1}{1+c^2} - h,
\end{align}\end{subequations}
where $D^{\star} =D\tau_j/L^2$, $K_1=k_f\tau_j/k_1$, $G=\gamma \tau_j/k_1$, $K=k_{\gamma}/k_1$, and $\lambda=\tau_j S/k_1$. 
For the parameter values we have chosen from \cite{allbritton,atri93} and also taking $D=20 \mu m^2/s$ and $L=100\,\mu$m we obtain the following values for the nondimensional parameters 
$D^{\star} =4 \times 10^{-3}$, $K_1=46.29$, $G=40/7$, $K=1/7$. In this context, a large value of $K_1$ 
captures the fact that calcium is a fast messenger \cite{berridge}. 
We have assumed here that mechanics modify the behaviour of calcium only through the advection term and an additional calcium flux which is linearly dependent on the local dilation. 
The latter flux is modulated by $\lambda>0$, thus representing 
a source for $c$ if the solid volume increases, and a calcium sink otherwise (see e.g. \cite{neville06}).  This parameter will be treated as a bifurcation constant, as in \cite{kaouri19}. In \cite{kaouri19} it has been shown that for a critical value of $\lambda$, oscillations were suppressed, and we expect a similar behaviour in the higher dimensional case here. We will vary $\lambda$ from 0 (calcium dynamics not affected by mechanics) to $\lambda=5$. However it is not clear what the maximum $\lambda$ value should be - inspired by the results in \cite{kaouri19} we expect the oscillations/waves to vanish for large enough $\lambda$. 
 
Note that $\alpha_1,\alpha_2$ determine the magnitude of the viscous effects. Also, the smaller $\beta_2$ is, the faster $T(c)$  in \eqref{eq:defT} (actually its nondimensional form) will tend to its saturation value $\beta_1$. In \cite{kaouri19} the authors have explored this variation of the Hill function $T(c)$; here we take the values $\alpha_1 = 1$, $\alpha_2 = 0.5$, $\beta_2 =0.1$. 
 
The system composed by equations \eqref{eq:p2}-\eqref{eq:mom2} and \eqref{eq:c}-\eqref{eq:n} 
is complemented with appropriate initial data at rest
\begin{equation}\label{eq:initial}
c(0) = c_{0}, \quad h(0)=h_{0}, \quad \bu(0)= \cero, \quad p(0) = 0, \quad \text{in $\Omega\times\{0\}$.}
\end{equation}
Homogeneous boundary conditions can be incorporated for normal displacements, calcium fluxes, and traction, in the following manner  
\begin{subequations}
\begin{align}\label{bc:Gamma}
\bu\cdot \nn = 0 \quad \text{and} \quad D^\star \nabla c \cdot\nn &= 0 &\text{on $\Gamma\times(0,t_{\text{final}}]$},\\
\label{bc:Sigma}
\bsigma\nn = \cero \quad\text{and}\quad \quad D^\star \nabla c \cdot\nn&=0  &\text{on $\Sigma\times(0,t_{\text{final}}]$},
\end{align}\end{subequations}
where the boundary $\partial\Omega = \Gamma\cup\Sigma$  is disjointly split into $\Gamma$ and $\Sigma$ 
where we prescribe slip boundaries and zero traction, 
respectively. This case assumes that the tissue is allowed to slip along the substrate on $\Gamma$, while it is free to deform on $\Sigma$. 
However, in most of our numerical tests we will consider instead the pure traction boundary conditions
\begin{equation}
\label{bc:pureTraction}
\bsigma\nn = \cero \quad\quad\text{and}\quad \quad D^\star \nabla c \cdot\nn=0  \qquad \text{on $\partial\Omega\times(0,t_{\text{final}}]$}.
\end{equation}
In this case an additional condition is required to make the system well-defined. 
For instance, we can  impose orthogonality to the space of rigid motions  defined as (see \cite[Eq. (11.1.7)]{brenner-2008})
 \begin{equation}\label{RM}
\mathbb{RM}(\Omega):=\left\{\bv\in \bH^1(\Omega):\ \beps(\bv)=\cero\right\},
 \end{equation}
and  unique solvability (for a given calcium concentration $c$) will follow since $\mathbb{RM}(\Omega)$ is a null space of the relevant functional space.

\subsection{Mixed-primal weak formulation}
 Multiplying the nondimensional governing equations \eqref{eq:system-adim} 
 by adequate test functions and integrating by parts (in space) whenever 
appropriate, we end up with the following variational problem, here restricted to the 
case of pure traction boundary conditions: For a given $t>0$, find $\bu(t)\in\bV$, 
$p(t)\in L^2(\Omega)$, $c(t),n(t)\in  H^1(\Omega)$, 
such that
\begin{subequations}\label{eq:weak}
\begin{align}
 \int_{\Omega} \beps(\bu):\beps(\bv) -\!\int_{\Omega}\!p\vdiv \bv 
 +\alpha_1 \!\int_{\Omega}\! \partial_t \beps(\bu):\beps(\bv) \qquad \qquad & \nonumber \\
 -\alpha_2 \!\int_{\Omega}\! \partial_t p \vdiv\bv  
-\!\int_{\Omega}\!  \bsym{\beta}(c) \vdiv \bv  
  = 
 0
  & \quad \forall \bv\in\bV, \\
   -\int_{\Omega}q \vdiv \bu - \frac{(1-2\nu)}{\nu} \int_{\Omega} pq = 0 & \quad\forall q\in L^2(\Omega),\label{eq:weak1}\\
   \int_\Omega  \partial_t c\phi + \int_\Omega  (\partial_t\bu\cdot\nabla c+\frac{1}{2}(c\vdiv\pt\vec{u}))\phi + \int_\Omega D^\star \nabla c\cdot \nabla \phi & \nonumber \\ 
   = 
   \int_{\Omega} \biggl[ \vec{K}(h,c) + \lambda \vdiv \bu \biggr]\phi
   & \quad \forall \phi \in H^1(\Omega),\\
    \int_\Omega  \partial_th\, \psi + \int_\Omega  (\partial_t\bu\cdot\nabla h+\frac{1}{2}(h\vdiv\pt\vec{u})\psi  = 
   \int_{\Omega} \biggl[\vec{J}(c) - h \biggr]\psi & \quad \forall \psi \in H^1(\Omega).\label{eq:weak2}
  \end{align}\end{subequations}
 Here we have defined the additional space 
  $$\bV:=\mathbb{RM}(\Omega)^{\perp}=\left\{\bv\in \bH^1(\Omega):\int_{\Omega}\bv=\boldsymbol{0}\quad\text{and}\quad \int_{\Omega}\bcurl\bv=\cero \right\},$$
meaning that all rigid motions are discarded as 
feasible displacement solutions. Alternatively, one can weakly enforce orthogonality to the space of rigid motions by a Lagrange multiplier, 
which is what we will do at the discrete level. 

Furthermore, we use an equivalent skew-symmetric form for advection terms and define $\bsym{\beta}(c) = \frac{\beta_1 c^n}{\beta_2+c^n}$, $\vec{K}(h,c):= \mu h K_1 \frac{b+c}{1+c} - \frac{G c}{K+c}$ and $\vec{J}(c) = \frac{1}{1+c^2}$.

\subsection{Fully discrete form}\label{sec:FE}
Let us consider a 
 shape-regular partition $ \cT_j$  of  $\bar{\Omega}$ into affine elements (triangles in 2D or 
tetrahedra in 3D) $E$ of 
diameter $j_E$, where $j = \max\{ j_E:\, E\in \cT_j\}$ denotes the meshsize. 
Finite-dimensional subspaces for the approximation of displacement, pressure, and calcium are 
specified as  
\begin{align}\label{eq:FEspaces}
\widetilde{\bV}_j &:= \lbrace \bv_j \in \bH^1(\Omega): \bv_j\vert_{E} \in [\mathbb{P}_2(E)]^d\ \forall E\in \cT_j, \text{ and } \bv_j\vert_{\Gamma}= \cero\rbrace,\nonumber\\
\bV_j &:= \lbrace \bv_j \in \bH^1(\Omega): \bv_j\vert_{E} \in [\mathbb{P}_1(E)\oplus {\rm span}\{b_E\}]^d\ \forall E\in \cT_j, \text{ and } \bv_j\vert_{\Gamma}= \cero\rbrace, \nonumber \\
\widetilde{\Phi}_j &:= \lbrace q_j \in C^0(\Omega): q_j\vert_{E} \in \mathbb{P}_2({E})\ \forall E\in \cT_j\rbrace,\quad \widetilde{\Psi}_j := \widetilde{\Phi}_j,\nonumber \\
Q_j &:= \lbrace q_j \in C^0(\Omega): q_j\vert_{E} \in \mathbb{P}_1({E})\ \forall E\in \cT_j\rbrace, \quad   
\Phi_j := Q_j, \quad \Psi_j := Q_j,
\end{align}
where $\mathbb{P}_k(E) $ denotes the space of polynomials of degree less than or equal than $ k $ defined locally over the 
element $ E \in \cT_j$, and $b_E:= \varphi_1\varphi_2\varphi_3$ is a polynomial bubble function in $E$, and $\varphi_1,\,\varphi_2\,,\varphi_3$ are the barycentric coordinates of $E$. Depending on whether the  displacement is approximated with $\bV_j$ or $\widetilde{\bV}_j$, the pairs 
  $\bV_j\times Q_j$ are known as the MINI element and $\widetilde{\bV}_j\times Q_j$ are known as the Taylor-Hood elements. They are both Stokes inf-sup stable (see, e.g., \cite{boffi13,quarteroni-valli}). On the other hand, 
an $\mathbf{L}^2$-orthonormal basis for the space of rigid motions is  constructed in terms of translations and rotations associated with the 
principal axes of the inertial tensor $\mathbb{I}$ encoding rotational kinetic energy \cite{kutcha18}. 
 Denoting $\boldsymbol{x}_0 = \vert {\Omega} \vert^{-1} \int_{\Omega} \boldsymbol{x} $ 
 the centre of mass of the domain, such basis (having dimension 6 in 3D) assumes the form
\begin{align*}
\bW_j & := \{ \vert\Omega\vert^{-1/2}\bnu_x,\vert\Omega\vert^{-1/2}\bnu_y,\vert\Omega\vert^{-1/2}\bnu_z,\\
& \qquad \lambda_x^{-1/2}(\bx-\bx_0)\times\bnu_x,\lambda_y^{-1/2}(\bx-\bx_0)\times\bnu_y ,\lambda_y^{-1/2}(\bx-\bx_0)\times\bnu_y  \},\end{align*}
where the $(\lambda_i,\bnu_i)$ are the eigenpairs of $\mathbb{I}$ (see also \cite{kutcha15}). As remarked in \cite{kutcha18}, we emphasise that not all Stokes inf-sup stable pairs will lead to $j-$robust bounds when approximating the singular Neumann problem of linear elasticity (for instance, using for displacement and pressure the pair $\mathbb{P}_2-\mathbb{P}_0$ does not allow for constructing a well-posed mixed Poisson auxiliary problem needed in establishing orthogonality with respect to the kernel).

Next we discretise the time interval $(0,t_{\text{final}}]$ into equi-spaced points $t^{k} = k\Delta t$. Applying a backward Euler method and 
an implicit centred difference discretisation of the first- and second-order time derivatives, respectively; we can write a semidiscrete form of \eqref{eq:weak} (but now incorporating the orthogonality with respect to rigid motions with a Lagrange multiplier). 
Now the fully discrete formulation is given by
{\small
\begin{subequations}\label{eq:discr1}
	\begin{align}
 		\frac{\alpha_1}{\Dt}\int_{\Omega}(\beps(\vec{u}^{k+1}_j)-\beps(\vec{u}^{k}_j))\,:\,\beps(\vec{v})+\int_{\Omega} \beps(\vec{u}^{k+1}_j)\,:\,\beps(\vec{v})  - \int_{\Omega}  p^{k+1}_j\vdiv \vec{v} & \nonumber \\ 
   - \frac{\alpha_2}{\Dt} \int_{\Omega}(p^{k+1}_j - p^{k}_j)\vdiv \vec{v} - \int_{\Omega}(\bsym{\beta}(c^{k+1}_j)\vdiv\vec{v}) &  =  0,\\
		-\int_{\Omega} q\vdiv\vec{u}^{k+1}_j - \frac{1-2\nu}{\nu}\int_{\Omega} p^{k+1}_jq  & =0, \\
		\frac{1}{\Dt}\int_{\Omega}(c^{k+1}_j-c^{k}_j)\phi+\frac{1}{\Dt}\int_{\Omega}((\vec{u}^{k+1}_j-\vec{u}^{k}_j)\cdot \nabla c^{k+1}_j\phi) + \frac{1}{2\Dt}\int_{\Omega} c^{k+1}_j \vdiv\vec{u}^{k+1}_j\phi  &  \nonumber \\ 
		 + \int_{\Omega}(D^*\nabla c^{k+1}_j\cdot\nabla\phi)  
 	- \int_{\Omega}(\vec{K}(h^{k+1}_j,c^{k+1}_j)+\lambda\vdiv \vec{u}^{k+1}_j)\phi & = 0, \\
 	\frac{1}{\Dt}\int_{\Omega}(h^{k+1}_j-h^{k}_j)\psi + \frac{1}{\Dt}\int_{\Omega}((\vec{u}^{k+1}_j-\vec{u}^{k}_j)\cdot \nabla h^{k+1}_j\psi) & \nonumber \\
	+ \frac{1}{2\Dt}\int_{\Omega}(h^{k+1}_j\vdiv(\vec{u}^{k+1}_j-\vec{u}^{k}_j)\psi)  - 
 \int_{\Omega}(\vec{J}(c^{k+1}_j)-h^{k+1}_j)\psi & = 0.
	\end{align}
\end{subequations}}

\section{Well-posedness analysis}\label{sec:wellp}

We will consider that the initial data are nonnegative and regular enough, hence the numerical scheme starts from discrete initial data $\bu_j^{0},p_j^{0}, c_j^0, h_j^0$, which are the projections of the exact initial conditions onto the finite element spaces.

Let us introduce the trilinear form $d:  \vec{H}^1(\Omega)\times H^1(\Omega) \times H^1(\Omega) \to \RR$ and bilinear form $e: H^1(\Omega)\times H^1(\Omega) \to \RR$ defined as:
\begin{align*}
	d(\vec{u};\phi,\psi) = \int_{\Omega}(\vec{u}\cdot\phi)\psi+\frac{1}{2}\int_{\Omega}(\phi\vdiv\vec{u})\psi, \quad 
	e(c,\phi) = \int_{\Omega}(D^*\nabla c^{k+1}\cdot\nabla\phi).
\end{align*}

Then, due to the skew-symmetric form of the operator $d(\cdot;\cdot,\cdot)$, we can write
\begin{align}\label{eq:czero}
	d(\vec{u};\phi,\phi) = 0 \quad \text{for all }\phi \in H^1(\Omega).
\end{align}

And recalling the Poincaré-Friedrichs inequality \cite[ Chapter I, Lemma 3.1]{Girault2005} , we have the coercivity for $e(\cdot,\cdot)$:
\begin{align}\label{eq:dcoerv}
	e(\phi,\phi) \geq \alpha_e \norm{\phi}_{1,\Omega}^2 \quad \text{for all }\phi \in H_0^1(\Omega).
\end{align}

We also assume the nonlinearities satisfy the following conditions: 
\begin{itemize}
	\item[(A1)] $\vec{J}(c)$ is uniformly bounded, i.e. $\abs{\vec{J}(c)}\leq C_J$.
	\item[(A2)] $\vec{K}(c,h)$ is uniformly bounded with respect to $c$ and Lipschitz with respect to $h$, in particular we have: $\abs{\vec{K}(h,c)} \leq C_k \abs{h}$.
	\item[(A3)] $\bsym{\beta}(c)$ is uniformly bounded, i.e. $\abs{\bsym{\beta}(c)}<C_{\beta}$.
\end{itemize}

Moreover, we will use the following algebraic relation:  Let $a$ and $b$ be two real numbers, then we have
\begin{align} \label{eq:alg}
	2(a-b,a) &= a^2-b^2+(a-b)^2.
\end{align}

\subsection{Continuous dependence on data}

\begin{theorem} \label{thm:cont_dep}
	Let $(\vec{u}^{k+1}_j,p^{k+1}_j,c^{k+1}_j,h^{k+1}_j) \in (\vec{V}_j,Q_j,\Phi_j,\Phi_j)$ be a solution of problem \eqref{eq:discr1}, with initial data $(\vec{u}^0_j,c^0_j,h^0_j)$ Then under the assumption that $\left(1-\frac{\sqrt{d}}{2}\frac{1-4\nu}{4\nu-3}\right)>0$, the following bounds are satisfied:
	\begin{align*}
&		\norm{\vec{u}^{n+1}_j}_{1,\Omega}^2+\Dt\sum_{k=0}^{n}\norm{\vec{u}^{k+1}_j}_{1,\Omega}^2 + \sum_{k=0}^{n}\norm{\vec{u}^{k+1}_j}_{1,\Omega}^2 \leq C_1 \left( \norm{\vec{u}^0_j}_{1,\Omega}^2 + t_{\mathrm{final}}C_{\beta}^2\right), \\
&		\norm{c^{n+1}_j}^2_{0,\Omega} + \Dt \sum_{k=0}^n \norm{c^{k+1}_j}_{1,\Omega}^2 \\
& \qquad  \leq C_2\left(\norm{c^0_j}_{0,\Omega}^2 + \norm{h^0_j}_{0,\Omega}^2 + \norm{\vec{u}^0_j}_{1,\Omega}^2 + t_{\mathrm{final}}(C_{\beta}^2+C_J^2)\right), \\
&		\norm{h^{n+1}_j}^2_{0,\Omega}  + \Dt  \sum_{k=0}^n\norm{h^{k+1}_j}_{0,\Omega}^2 \leq C_3 \left(\norm{h^0_j}_{0,\Omega}^2 + t_{\mathrm{final}} C_J^2 \right),
	\end{align*}
	where $C_1$, $C_2$ and $C_3$ are positive constants  independent of $j$ and $\Dt$. 
\end{theorem}
\begin{proof}
On the second equation of problem \eqref{eq:discr1}  we choose $q=p^{k+1}_j$, to get
	\begin{align}\label{eq:divp}
		-\frac{1-2\nu}{\nu}\norm{p^{k+1}_j}_{0,\Omega}^2 = (p^{k+1}_j,\vdiv \vec{u}^{k+1}_j)_{\Omega},
	\end{align}
	which after applying Young's inequality, implies
	
	\begin{align*}
		\frac{1-2\nu}{\nu}\norm{p^{k+1}_j}_{0,\Omega}^2 &\leq \frac{1}{2}\norm{p^{k+1}_j}_{0,\Omega}^2 + \frac{1}{2}\norm{\vdiv\vec{u}^{k+1}_j}_{0,\Omega}^2,
	\end{align*}
	hence,
	\begin{align}
		\norm{p^{k+1}_j}_{0,\Omega}^2 &\leq \frac{2\nu\sqrt{d}}{4\nu-3}\snorm{\vec{u}^{k+1}_j}_{1,\Omega}^2\nonumber \\
		&\leq C_p \snorm{\vec{u}^{k+1}_j}_{1,\Omega}^2. \label{eq:bound_p}
	\end{align}
	
	Now,  on the first equation in problem \eqref{eq:discr1}, we choose the test function $\vec{v} = \vec{u}^{k+1}$ multiply by $2\Dt$ and use algebraic relation \eqref{eq:alg} in combination with equations \eqref{eq:divp}, \eqref{eq:bound_p}. We obtain:
	\begin{align*}
		2\Dt \norm{\beps(\vec{u}^{k+1}_j)}_{0,\Omega}^2 + \alpha_1\left(\norm{\beps(\vec{u}^{k+1}_j)}_{0,\Omega}^2 + \norm{\beps(\vec{u}^{k+1}_j) - \beps(\vec{u}^k_j)}_{0,\Omega}^2 -\norm{\beps(\vec{u}^{k}_j)}_{0,\Omega}^2 \right)\\
		+ 2(\Dt + \alpha_2)(\frac{1-2\nu}{\nu})C_p\snorm{\vec{u}^{k+1}_j}_{1,\Omega}^2 \leq  \alpha_2 \norm{p^k_j}_{0,\Omega}^2 + \alpha_2\sqrt{d}\snorm{\vec{u}^{k+1}_j}_{1,\Omega}^2\\
		+\Dt\norm{\bsym{\beta}(c^{k+1}_j)}_{0,\Omega}^2 + \Dt\sqrt{d} \snorm{\vec{u}^{k+1}_j}_{1,\Omega}^2.
	\end{align*}
	
	By Korn's and Young's inequalities, we get
	\begin{align*}
		\Dt\left(1-\frac{\sqrt{d}}{2}\frac{1-4\nu}{4\nu-3}\right)\snorm{\vec{u}^{k+1}_j}_{1,\Omega}^2 + \snorm{\vec{u}^{k+1}_j-\vec{u}^{k}_j}_{1,\Omega}^2+ \frac{\alpha_1}{2} (\snorm{\vec{u}^{k+1}_j}_{1,\Omega}^2 - \snorm{\vec{u}^{k}_j}_{1,\Omega}^2) \\+ 2\alpha_2 C_p (\snorm{\vec{u}^{k+1}_j}_{1,\Omega}^2 - \snorm{\vec{u}^{k}_j}_{1,\Omega}^2) \leq 2\Dt\norm{\bsym{\beta}(c^{k+1}_j)}_{0,\Omega}^2.
	\end{align*}	
	
	Summing over $k$ from $0$ to $n \leq N-1$, applying the Poincaré-Friedrichs inequality and assuming $\left(1-\frac{\sqrt{d}}{2}\frac{1-4\nu}{4\nu-3}\right)>0$ and (A3), we get	
	
	\begin{align} \label{eq:r1}
		\norm{\vec{u}^{n+1}_j}_{1,\Omega}^2+\Dt\sum_{k=0}^{n}\norm{\vec{u}^{k+1}_j}_{1,\Omega}^2 + \sum_{k=0}^{n}\norm{\vec{u}^{k+1}_j}_{1,\Omega}^2 \leq C (\norm{\vec{u}^0_j}_{1,\Omega}^2 + t_{\text{final}}C_{\beta}^2)
	\end{align}
	
	Now we take $\phi = c^{k+1}_j$, use properties \eqref{eq:czero}, \eqref{eq:dcoerv} and multiply by $2\Dt$ the third equation in problem \eqref{eq:discr1}, to deduce:
	\begin{align*}
		&\norm{c^{k+1}_j}^2_{0,\Omega} + \norm{c^{k+1}_j-c^k_j}_{0,\Omega}^2 - \norm{c^k_j}_{0,\Omega}^2+2\alpha_e \Dt \norm{c^{k+1}_j}_{1,\Omega}^2 \leq \\ &\qquad \Dt\norm{\vec{K}(h^{k+1}_j,c^{k+1}_j)}_{0,\Omega}\norm{c^{k+1}_j}_{0,\Omega} + \Dt\lambda \sqrt{d} \snorm{\vec{u}^{k+1}_j}_{1,\Omega}\norm{c^{k+1}_j}_{0,\Omega}
	\end{align*}
	
	As before, summing over $k$ from $0$ to $n \leq N-1$, and applying Young's inequality we deduce
	\begin{align*}
	&	\norm{c^{n+1}_j}^2_{0,\Omega} +\alpha_e \Dt \sum_{k=0}^n \norm{c^{k+1}_j}_{1,\Omega}^2 \\
	&\qquad 	\leq \norm{c^0_j}_{0,\Omega}^2+\frac{\Dt}{2\alpha_e}\sum_{k=0}^n \norm{\vec{K}(h^{k+1}_j,c^{k+1}_j)}_{0,\Omega}^2 + \frac{\lambda^2 d}{2\alpha_e}\sum_{k=0}^n \snorm{\vec{u}^{k+1}_j}_{1,\Omega}^2.
	\end{align*}	
	By assumption (A2) we have,
	\begin{align} \label{eq:r2}
	&	\norm{c^{n+1}_j}^2_{0,\Omega} +\alpha_e \Dt \sum_{k=0}^n \norm{c^{k+1}_j}_{1,\Omega}^2 \nonumber\\
	&\qquad	 \leq \norm{c^0_j}_{0,\Omega}^2+\frac{\Dt C_K^2}{2\alpha_e}\sum_{k=0}^n \norm{h^{k+1}_j}_{0,\Omega}^2 + \frac{\lambda^2 d}{2\alpha_e}\sum_{k=0}^n \snorm{\vec{u}^{k+1}_j}_{1,\Omega}^2. 
	\end{align}

	Similarly in the fourth equation of \eqref{eq:discr1}, we take $\psi = h^{k+1}_j$, multiply by $2\Dt$ and use property \eqref{eq:czero} to get
	
	\begin{align*}
		\|h^{k+1}_j\|^2_{0,\Omega} +\! \|h^{k+1}_j\! - h^k_j\|_{0,\Omega}^2\! -\! \|h^k_j\|_{0,\Omega}^2 + \Dt \|h^{k+1}_j\|_{0,\Omega}^2 \leq \Dt\|\vec{J}(c^{k+1}_j)\|_{0,\Omega}\|h^{k+1}_j\|_{0,\Omega}. 
	\end{align*}
	
	Hence, applying Young's inequality, assumption (A1) and summing over k, we get
	\begin{align} \label{eq:r3}
		\norm{h^{n+1}_j}^2_{0,\Omega}  + \frac{\Dt}{2}  \sum_{k=0}^n\norm{h^{k+1}_j}_{0,\Omega}^2 \leq \norm{h^0_j}_{0,\Omega}^2 +  \frac{t_{\text{final}}}{2}C_J.
	\end{align}
	
	While the first and third results follows directly from \eqref{eq:r1} and \eqref{eq:r3} respectively, we get the second result by substituting bounds for $h_j$ and $u_j$ into \eqref{eq:r2}.
	
\end{proof}

\subsection{Existence result: Fixed-point approach}

Next, we address the unique solvability of problem \eqref{eq:discr1}. To that end, we will make use of Brouwer's fixed-point theorem in the following form 
(given by \cite[Corollary 1.1, Chapter IV]{girault1986}): 

\begin{theorem}[Brouwer's fixed-point theorem] \label{th:brouwer}  Let $H$ be a finite-dimensional Hilbert space with scalar product  $(\cdot,\cdot )_H$ and corresponding norm $\norm{\cdot}_H$. Let $\Phi\colon H \to H$ be a continuous mapping  for which there exists $\mu > 0$ such that $(\Phi(u),u)_H \geq 0$ for all $u \in H$  with $\norm{u}_H  = \mu$. 		Then  there exists $u \in H$ such that $\Phi(u) = 0$ and $\norm{u}_H \leq \mu$. \end{theorem}

\begin{theorem} [Existence of discrete solutions] Under the same assumptions as in Theorem \ref{thm:cont_dep}, problem \eqref{eq:discr1} with initial data $(\vec{u}_j^0, p_j^0, c_j^0,h_j^0)$ admits at least one solution.
\end{theorem}

\begin{proof}
	 To simplify the proof we introduce the constants:
	\begin{gather*}
	C_u  := C_1 \left( \norm{\vec{u}^0_j}_{1,\Omega}^2 + t_{\text{final}}C_{\beta}^2\right), \quad C_h  := C_3 \left(\norm{h^0_j}_{0,\Omega}^2 + t_{\text{final}} C_J^2 \right), \\
	C_c := C_2\left(\norm{c^0_j}_{0,\Omega}^2 + \norm{h^0_j}_{0,\Omega}^2 + \norm{\vec{u}^0_j}_{1,\Omega}^2 + t_{\text{final}}(C_{\beta}^2+C_J^2)\right).
	\end{gather*}
We proceed by induction on $k > 2$. We define the mapping
\begin{align*}
	\Psi\,:\,\vec{V}_j\times Q_j\times \Phi_j\times \Phi_j \to \vec{V}_j\times Q_j\times \Phi_j\times \Phi_j,
\end{align*}
using the relation
\begin{align*}
&\qquad	(\Psi(\vec{u}^{k+1}_j,p^{k+1}_j,c^{k+1}_j,h^{k+1}_j),(\vec{v},q,\phi,\psi))_{\Omega}  = \\
&	\frac{\alpha_1}{\Dt}\int_{\Omega}(\beps(\vec{u}^{k+1}_j)-\beps(\vec{u}^{k}_j))\,:\,\beps(\vec{v})+\int_{\Omega}(\beps(\vec{u}^{k+1}_j)\,:\,\beps(\vec{v})) - \int_{\Omega}p^{k+1}_j\vdiv \vec{v} \\
&  -  \frac{\alpha_2}{\Dt} \int_{\Omega}(p^{k+1}_j - p^{k}_j)\vdiv \vec{v} - \int_{\Omega}\bsym{\beta}(c^{k+1}_j)\vdiv\vec{v}+ 
	\int_{\Omega}q\vdiv\vec{u}^{k+1}_j + \frac{1-2\nu}{\nu}\int_{\Omega}p^{k+1}_jq  \\
&	+ \frac{1}{\Dt}\int_{\Omega}(c^{k+1}_j-c^{k}_j)\phi+\frac{1}{\Dt}\int_{\Omega}((\vec{u}^{k+1}_j-\vec{u}^{k}_j)\cdot \nabla c^{k+1}_j\phi) + \frac{1}{2\Dt}\int_{\Omega}(c^{k+1}_j \vdiv\vec{u}^{k+1}_j\phi)
	\nonumber \\
&+ \int_{\Omega}(D^*\nabla c^{k+1}_j\cdot\nabla\phi) - \int_{\Omega}(\vec{K}(h^{k+1}_j,c^{k+1}_j)+\lambda\vdiv \vec{u}^{k+1}_j)\phi + \frac{1}{\Dt}\int_{\Omega}(h^{k+1}_j-h^{k}_j)\psi \\
&	 + \frac{1}{\Dt}\int_{\Omega}(\vec{u}^{k+1}_j-\vec{u}^{k}_j)\cdot \nabla h^{k+1}_j\psi + \frac{1}{2\Dt}\int_{\Omega}h^{k+1}_j\vdiv(\vec{u}^{k+1}_j-\vec{u}^{k}_j)\psi \\
&  -\int_{\Omega}(\vec{J}(c^{k+1}_j)-h^{k+1}_j)\psi.
\end{align*}
Note that this map is well-defined and continuous on $\vec{V}_j\times Q_j\times \Phi_j\times \Phi_j$. On the other hand, if we take
\begin{align*}
	(\vec{v},q,\phi,\psi) = (\vec{u}^{k+1}_j,p^{k+1}_j,c^{k+1}_j,h^{k+1}_j),
\end{align*}
and employ \eqref{eq:czero}, \eqref{eq:dcoerv}, \eqref{eq:divp}, together with assumptions (A1)-(A3), we obtain
\begin{align*}
	&\qquad (\Psi(\vec{u}^{k+1}_j,p^{k+1}_j,c^{k+1}_j,h^{k+1}_j),(\vec{u}^{k+1}_j,p^{k+1}_j,c^{k+1}_j,h^{k+1}_j))_{\Omega} \geq \\ 
	&\frac{\alpha_1}{\Dt}\left(\norm{\beps(\vec{u}^{k+1}_j)}^2_{0,\Omega} - \norm{\beps(\vec{u}^{k}_j)}_{0,\Omega}\norm{\beps(\vec{u}^{k+1}_j)}_{0,\Omega}\right) + \norm{\beps(\vec{u}^{k+1}_j)}^2_{0,\Omega} \\
	&+ \frac{\alpha_2}{\Dt}\frac{1-2\nu}{\nu}\norm{p^{k+1}_j)}^2_{0,\Omega} - \frac{\alpha_2}{\Dt}\norm{p^k_j}_{0,\Omega}\norm{\vdiv \vec{u}^{k+1}_j}_{0,\Omega}-C_{\beta}\norm{\vdiv\vec{u}^{k+1}_j}_{0,\Omega}  \\
	&+\frac{1}{\Dt} \left(\norm{c^{k+1}_j}_{0,\Omega}^2 - \norm{c^{k+1}_j}_{0,\Omega}\norm{c^{k}_j}_{0,\Omega}\right) + \alpha_e \norm{c^{k+1}_j}_{1,\Omega}^2 - C_{K}\norm{h^{k+1}_j}_{0,\Omega}\norm{c^{k+1}_j}_{0,\Omega}\\ 
	&- \lambda \norm{\vdiv \vec{u}^{k+1}_j}_{0,\Omega}\norm{c^{k+1}_j}_{0,\Omega} + 
	\frac{1}{\Dt}\left(\norm{h^{k+1}_j}_{0,\Omega}^2 - \norm{h^{k+1}_j}_{0,\Omega}\norm{h^{k}_j}_{0,\Omega}\right) \\
	&- C_J\norm{h^{k+1}_j}_{0,\Omega} + \norm{h^{k+1}_j}_{0,\Omega}^2.
\end{align*}
Next, using Korn's Inequality (with constant $\hat{C}_k$) and \eqref{eq:bound_p}, we deduce that
\begin{align*}
	&(\Psi(\vec{u}^{k+1}_j,p^{k+1}_j,c^{k+1}_j,h^{k+1}_j),(\vec{u}^{k+1}_j,p^{k+1}_j,c^{k+1}_j,h^{k+1}_j))_{\Omega} \geq \\ &\quad \frac{1}{\hat{C}_k}\norm{\vec{u}^{k+1}_j}_{1,\Omega}^2 + \frac{\alpha_2\sqrt{d}(2-4\nu)}{\Dt(4\nu-3)}\norm{\vec{u}^{k+1}_j}_{1,\Omega}^2-\lambda\sqrt{d}C_c\norm{\vec{u}^{k+1}_j}_{1,\Omega}\\
	&\quad - \frac{\alpha_2C_pC_u\sqrt{d}}{\Dt}\norm{\vec{u}^{k+1}_j}_{1,\Omega}-C_{\beta}\sqrt{d}\norm{\vec{u}^{k+1}_j}_{1,\Omega}+\alpha_e \norm{c^{k+1}_j}_{1,\Omega}^2 - C_c \norm{c^{k+1}_j}_{1,\Omega}\\
	&\quad  - C_KC_h\norm{c^{k+1}_j}_{1,\Omega} + 
	\norm{h^{k+1}}_{1,\Omega}^2 - 2C_h\norm{h^{k+1}}_{1,\Omega}-C_J\norm{h^{k+1}}_{1,\Omega}.
\end{align*}
Then, setting
\begin{gather*}
	C_R = \min\left\{ \frac{1}{\hat{C}_k},\frac{\alpha_2\sqrt{d}(2-4\nu)}{4\nu-3},\alpha_e,1 \right\}, \\ C_r = \max\left\{ \frac{\alpha_2C_pC_u\sqrt{d}}{\Dt},C_{\beta}\sqrt{d},C_c,C_KC_h, \lambda\sqrt{d}C_c,2C_h,C_J \right\},
\end{gather*}
we may apply the inequality $a+b\leq\sqrt{2}(a^2+b^2)^{1/2}$, valid for all $a,b\in \RR$, to obtain
\begin{align*}
	&(\Psi(\vec{u}^{k+1}_j,p^{k+1}_j,c^{k+1}_j,h^{k+1}_j),(\vec{u}^{k+1}_j,p^{k+1}_j,c^{k+1}_j,h^{k+1}_j))_{\Omega} \geq \\ &\qquad C_R  \left(\norm{\vec{u}^{k+1}_j}_{1,\Omega}^2+\norm{c^{k+1}_j}_{1,\Omega}^2+\norm{h^{k+1}_j}_{1,\Omega}^2 \right) \\ &\qquad - C_r\left(\norm{\vec{u}^{k+1}_j}_{1,\Omega}+\norm{c^{k+1}_j}_{1,\Omega}+\norm{h^{k+1}_j}_{1,\Omega} \right).
\end{align*}
Hence, the right-hand side is nonnegative on a sphere of radius $r := C_r/C_R$. Consequently, by Theorem \ref{th:brouwer}, there exists
a solution to the fixed-point problem $\Psi(\vec{u}^{k+1}_j,p^{k+1}_j,c^{k+1}_j,h^{k+1}_j) = 0$. 
\end{proof}



\subsection{Linearisation}

 At each time 
iteration we are left with a nonlinear system, and proceeding with a Newton linearisation we finally obtain the 
following scheme: 
Starting from discrete initial data $\bu_j^{0},p_j^{0}, c_j^0, h_j^0$ (projections of the exact initial conditions onto the 
finite element spaces), and 
for $\ell=1,\ldots$, we find $\bu_j^{\ell+1}\in\bV_j,p_j^{\ell+1}\in Q_j$ and $c_j^{\ell+1},h_j^{\ell+1}\in \Phi_j$, 
as the converged solutions of the iteration for $k=0,\ldots$ 
\begin{gather*} \bu_j^{k+1} \leftarrow \bu_j^{k} + \delta \bu^{k+1}_j,\quad p_j^{k+1} \leftarrow p_j^{k} + \delta p^{k+1}_j, \\
 c_j^{k+1} \leftarrow c_j^{k} + \delta c^{k+1}_j,\quad h_j^{k+1} \leftarrow h_j^{k} + \delta h^{k+1}_j,\end{gather*}
where the discrete 
Newton increments $\delta(\cdot)_j$ (the iteration superscript is discarded) solve the non-symmetric Jacobian linear problem 
\begin{subequations}\label{eq:jac}
\begin{alignat}{5}
  a^1(\delta\bu_j,\bv_j)   &\;+&\; b^1(\bv_j,\delta p_j)  + g(\bv_j,\delta\bomega_j)  &+ d^1_{c_j^{k}}(\delta c_j,\bv_j) &&= \nonumber\\
 &&& F_{k}(\bv_j)\quad \forall \bv_j\in\bV_j,  \label{jac-u}\\
  b^2(\delta\bu_j,q_j)       &\;-&\;       a^2(\delta p_j,q_j) &&   &=\nonumber\\
  &&& G_k(q_j)  \quad \forall q_j \in Q_j, \label{jac-p}\\
   g(\delta\bu_j,\btheta_j)  && &&& = \nonumber\\
   &&& H_k(\btheta_j) \quad \forall \btheta_j\in \bW_j,\label{jac-w} \\
  d^3_{c_j^{k}}(\delta\bu_j,\phi_j) && & +  a^3_{\bu_j^k,c_j^k,h_j^k}(\delta c_j,\phi_j)  &\;  +\;  e^3_{c_j^k}(\delta h_j,\phi_j)   & = \nonumber\\
   &&& I_{k}(\phi_j) \quad \forall \phi_j \in \Phi_j, \label{weak-c}\\
d^4_{h_j^{k}}(\delta\bu_j,\psi_j) & &   &e^4_{c_j^k}(\delta c_j,\psi_j) &\;  +\;   a^4(h_j,\psi_j)  & =\nonumber \\
&&& J_{k}(\psi_j) \quad \forall \psi_j\in \Psi_j. \label{weak-n}
\end{alignat}
\end{subequations}
We have used the bilinear forms $a_1:\bH^1(\Omega)\times\bH^1(\Omega)\to \RR$, 
 $a^2: L^2(\Omega)\times L^2(\Omega)\to \RR$, 
$a^3,a^4,e^3,e^4: H^1(\Omega)\times H^1(\Omega)\to \RR$, 
$b:\bH^1(\Omega)\times L^2(\Omega)\to \RR$, 
$g: \bH^1(\Omega)\times \mathbb{RM} \to \RR$,  
 $d^1:H^1(\Omega) \times\bV\to \RR$, 
 $d^3:\bH^1(\Omega)\times H^1(\Omega)  \to \RR$, $d^4:\bH^1(\Omega)\times H^1(\Omega)  \to \RR$, 
 where the subscripts explicitly indicate fixed quantities; and 
  linear functionals $F_k:\bH^1(\Omega)\to\RR$,  $G_{k},H_{k},I_k: H^1(\Omega)\to\RR$ where 
  the subscript $k$ denotes that they depend on 
  quantities associated with the state around which one performs linearisation. 
 The forms satisfy the following specifications 
\begin{gather*}
a^1(\bu,\bv) :=  {(1+\frac{\alpha_1}{\Delta t}) \int_{\Omega} \beps(\bu):\beps(\bv)} ,\quad 
b^1(\bv,p):= -   (1+\frac{\alpha_1}{\Delta t}) \int_{\Omega}p \vdiv \bv,\\
b^2(\bu,q):= -   \int_{\Omega}q \vdiv \bu,\quad a_4(h,\psi)  : =  \int_\Omega \bigl(\frac{1}{\Delta t}+1) h\psi,\\
d^1_{\hat{c}}(c,\bv): = -\int_{\Omega} {c \frac{\beta_1\beta_2 \hat{c}^{n-1}}{(\hat{c}^n+ \beta_2)^2} \vdiv \bv}, \quad 
a^2(p,q)  := {\frac{1-2\nu}{\nu}} \int_{\Omega}  pq,\\
d^3_{\hat{c}}(\bu,\phi)  := \frac{1}{\Delta t}\int_\Omega  (\bu \cdot\nabla \hat{c}) \phi  - \int_\Omega \lambda \vdiv \bu \phi,
\quad 
d^4_{\hat{h}}(\bu,\psi)  := \frac{1}{\Delta t}\int_\Omega  (\bu \cdot\nabla \hat{h}) \psi,  
\\
  a^3_{\hat{\bu},\hat{c},\hat{h}}(c,\phi) : =\!\! \int_\Omega \biggl(\frac{1}{\Delta t} +  {\mu K_1 \hat{h}\frac{(b-1-\hat{c} -\hat{c}^2)}{(1+\hat{c})^2}  + \frac{G}{K+\hat{c}}  - \frac{G K}{(K+\hat{c})^2}\biggr)} c\phi \\
  + \int_\Omega (\hat{\bu} \cdot \nabla c) \phi + \int_\Omega D^\star\nabla c \cdot \nabla \phi,\\ 
e^3_{\hat{c}}(h,\phi) :=  -\!\!\int_\Omega { \mu K_1 \frac{b+\hat{c}}{1+\hat{c}} h} \phi,\quad 
e^4_{\hat{c}}(c,\psi) :=  \!\int_\Omega \frac{2\hat{c}c \psi}{(1+\hat{c}^2)^2} ,\quad \\
g(\bv,\btheta) := \int_\Omega \bv\cdot\btheta,\quad F_{k}(\bv) := 
\!\int_{\Omega}\!  \frac{\beta_1(c_j^k)^n}{\beta_2+(c_j^k)^n} \vdiv \bv - a_1(\bu_j^k,\bv) - b^1(\bv,p_j^k), \\ 
G_k(q) := -b^2(\bu_j^k,q) + a_2(p_j^k , q), \qquad   H_{k}(\btheta) : = \int_{\Omega} \bu_j^k \cdot \btheta, \\ 
J_{k}(\psi) : = \int_{\Omega} \biggl(\frac{1}{1+(c_j^k)^2} - h_j^k + \frac{h_j^{\ell} +\bu_j^{\ell+1}\cdot\nabla h_j^{\ell}}{\Delta t} \biggr) \psi, \\
I_{k}(\phi) : = \int_{\Omega} \biggl(\frac{c_j^{\ell} +\bu_j^{\ell+1}\cdot\nabla c_j^{\ell}}{\Delta t} 
- \mu K_1 h_j^k \frac{b+c_j^k}{1+c_j^k} 
+ \frac{G c_j^k}{K+c_j^k} - \lambda \vdiv \bu_j^k\biggr)  \phi \\
-\int_\Omega \frac{1}{\Delta t} c_j^k\phi - \int_\Omega (\bu_j^k \cdot \nabla c_j^k) \phi - 
\int_\Omega D^\star \nabla c_j^k \cdot \nabla \phi. 
\end{gather*}
If one uses \eqref{bc:Gamma}-\eqref{bc:Sigma} then the third column and the third row \eqref{jac-w} in the tangent matrix are not needed.


\section{Numerical results}\label{sec:results}

\begin{figure}[!t]
\begin{center}
    \subfigure[]{\includegraphics[width=0.24\textwidth]{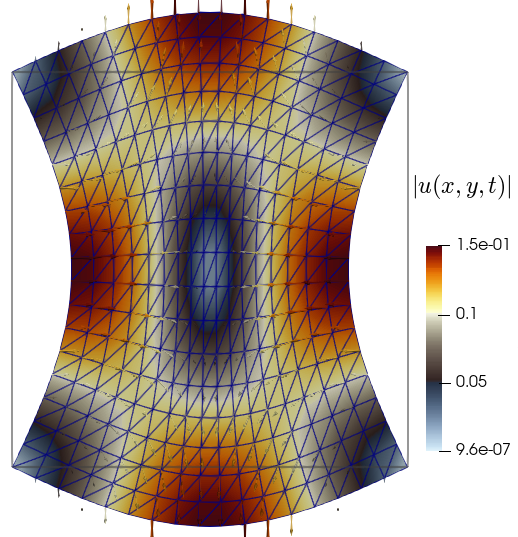}}
    \subfigure[]{\includegraphics[width=0.24\textwidth]{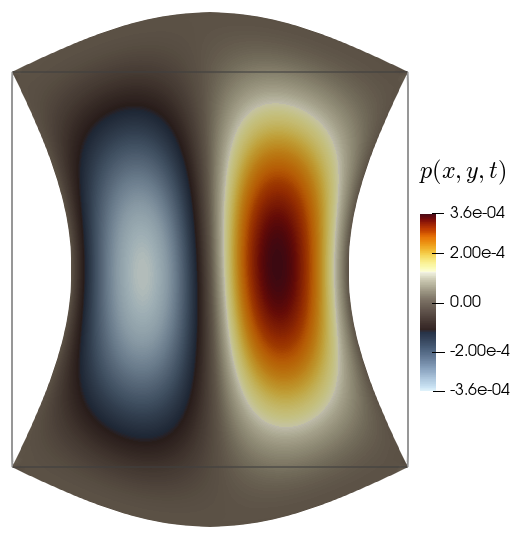}}
    \subfigure[]{\includegraphics[width=0.24\textwidth]{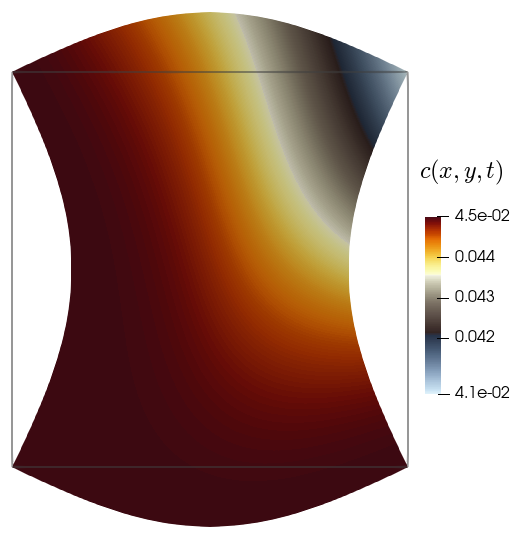}}
    \subfigure[]{\includegraphics[width=0.24\textwidth]{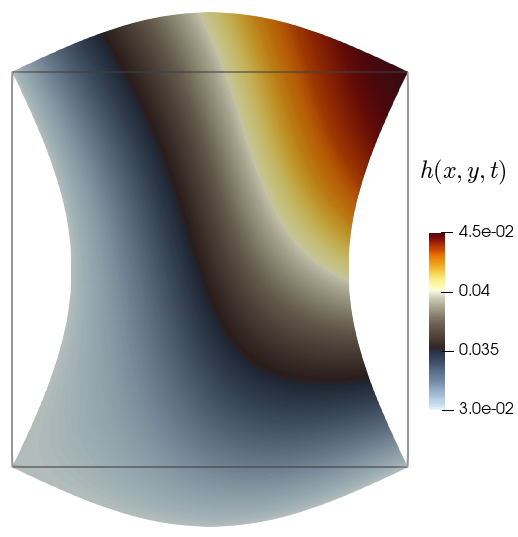}}
\end{center}
\caption{Convergence test. Approximate solutions on the deformed configuration, plotted on a coarse mesh and at $t = 3\Dt$.}\label{fig:test00}
\end{figure}

\subsection{Preliminaries and convergence verification}
All tests in this section have been implemented with the open-source finite element library FEniCS \cite{alnaes}. After matrix assembly, the resulting linear systems at each linearisation step 
are solved with the direct solver MUMPS. The stopping criterion on the nonlinear iterations of the Newton-Raphson algorithm is based on a weighted residual norm dropping below the fixed  tolerance of $1\cdot 10^{-7}$. 

Before computing numerical solutions pertaining to the application of calcium wave trains in 2D and 3D, we briefly address the verification of convergence for the space-time discretisation. For this we simply consider a unit square domain, partitioned uniformly into triangles of successive refinement. On each resolution level we compute approximate solutions and compare (using errors for displacement, pressure, calcium concentration and non-activated IPR in their natural norms),  against the following manufactured exact solutions 
\begin{gather*}
\bu(x,y,t) = f(t)\begin{pmatrix} 
5\cos(2\pi x)\sin(2\pi y)+\frac{x^2(1-x)^2y^2(1-y)^2(1-2\nu)}{\nu}\\[1ex]
-5\sin(2\pi x)\cos(2\pi y)+\frac{x^3(1-x)^3y^3(1-y)^3(1-2\nu)}{\nu}
\end{pmatrix}, \\
c(x,y,t) = f(t)\biggl(\frac{1}{2} + \frac{1}{2} \cos(\frac{\pi}{4}xy)\biggr) , \quad h(x,y,t) = 
f(t)\biggl(\frac{1}{2} + \frac{1}{2} \sin(\frac{\pi}{2}xy)\biggr), 
\end{gather*}
with $p = -\frac{\nu}{1-2\nu}\vdiv \bu$. For the convergence tests we use the Taylor-Hood element for displacement and pressure, and continuous and piecewise quadratic elements for $h$ and $c$. 
The parameters are $\nu = 0.49$, $\alpha_1 = \alpha_2 = 0.001$, $\beta_1 = \beta_2 = \mu = \lambda = b = G = K_1 = 1$, $K = 2$, $D^\star = 0.1$. The exact displacement in the viscoelastic case does not satisfy a homogeneous traction condition, so a synthetic traction is imposed computed from the exact stress. In addition, the exact calcium concentration is imposed as a Dirichlet  datum everywhere on the boundary.  
A time step $\Dt = 0.01$  is chosen 
and we simulate a   short time horizon $t_{\mathrm{final}} = 3\Delta t$. Errors ${\texttt{e}}_s$ between the 
approximate and exact solutions are tabulated against the number of degrees of freedom in Table~\ref{table:ex0j}. For the spatial convergence verification we have used $f(t)=t$. The expected convergence behaviour (quadratic rates for all fields) is observed. 
The Newton-Raphson algorithm takes, in average, four iterations to reach the prescribed residual tolerance. The convergence in time is verified with $f(t) = \sin(t)$ and by partitioning the time interval into successively refined uniform steps and computing accumulated errors $\hat{\texttt{e}}_s$. Here  
\[ {\texttt{e}}_s  = \norm{s(N\Dt)-s_j^{N+1}}_\star, \qquad \hat{\texttt{e}}_s = \biggl( \sum_{n=1}^{N} \Delta t \|s(t_{n+1})- s_j^{n+1}\|^2_{\star}\biggr)^{1/2},\] 
where {$\| \cdot \|_\star$} denotes the appropriate space norm for the generic vector or scalar field $s$ (that is, the $L^2-$norm for pressure and the $H^1-$norm for the remaining variables).  
The results are shown in Table \ref{table:ex0t}, confirming the expected first-order convergence. Samples of approximate solutions are depicted in Figure~\ref{fig:test00}.

\begin{table}[!t]
	\setlength{\tabcolsep}{4pt}
	\renewcommand{\arraystretch}{1.2}
	\centering 
	\caption{Convergence test. Experimental errors associated with the spatial discretisation (using Taylor-Hood and piecewise quadratic elements $\widetilde{\bV}_j\times Q_j \times \widetilde{\Phi}_j \times \widetilde{\Psi}_j$) and convergence rates for the approximate solutions 
		computed at the final time step.} \label{table:ex0j}
	\begin{tabular}{@{}rccccccccc@{}}
		\toprule
		DoF & $j$ & $\texttt{e}_{\bu}$ & \texttt{rate} & $\texttt{e}_{p}$ & \texttt{rate} & $\texttt{e}_{c}$ & \texttt{rate} & $\texttt{e}_{h}$ & \texttt{rate} \\
\midrule
   215 & 0.471 & 6.50e-02 & -- & 1.22e-04 & -- & 1.78e-04 & -- & 5.47e-04 & -- \\
   523 & 0.283 & 1.76e-02 & 2.554 & 3.38e-05 & 2.508 & 7.43e-05 & 1.710 & 2.06e-04 & 1.907 \\
  1547 & 0.157 & 5.52e-03 & 1.976 & 7.57e-06 & 2.548 & 2.14e-05 & 2.117 & 6.56e-05 & 1.951 \\
  5227 & 0.083 & 1.57e-03 & 1.981 & 1.69e-06 & 2.359 & 5.16e-06 & 2.237 & 1.88e-05 & 1.963 \\
 19115 & 0.043 & 4.18e-04 & 1.986 & 4.06e-07 & 2.147 & 1.26e-06 & 2.131 & 5.24e-06 & 1.929 \\
 73003 & 0.022 & 1.09e-04 & 1.991 & 1.02e-07 & 2.044 & 3.14e-07 & 2.043 & 1.36e-06 & 1.990 \\
\botrule
	\end{tabular} 
\end{table}

\begin{table}[t]
	\setlength{\tabcolsep}{4pt}
	\renewcommand{\arraystretch}{1.2}
	\begin{center}
	\centering 
			\caption{Convergence test. Experimental cumulative errors associated with the temporal discretisation and convergence rates for the approximate solutions  using a backward Euler scheme.} \label{table:ex0t}
	\begin{tabular}{@{}rccccccccc@{}}
		\toprule
			$\Delta t$ & $\hat{\texttt{e}}_{\bu}$ & \texttt{rate} & $\hat{\texttt{e}}_{p}$ & \texttt{rate} & $\hat{\texttt{e}}_{c}$ & \texttt{rate} & $\hat{\texttt{e}}_{h}$ & \texttt{rate} \\
\midrule
0.5          & 1.13e-01 & --         & 1.36e-02 & --       & 2.97e-01 & --        & 8.60e-02 & -- \\
 0.25        & 6.89e-02 & 0.927 & 5.66e-03 & 0.932 & 1.66e-01 & 0.783 & 4.14e-02 & 1.112\\ 
 0.125      & 3.36e-02 & 1.071 & 2.90e-03 & 0.945 & 8.20e-02 & 1.081 & 2.52e-02 & 0.844 \\ 
 0.0625    & 1.71e-02 & 0.976 & 1.49e-03 & 0.976 & 4.08e-02 & 1.075 & 1.70e-02 & 0.891 \\ 
 0.03125  & 8.59e-03 & 1.105 & 7.87e-04 & 1.065 & 2.14e-02 & 0.970 & 8.56e-03 & 0.968 \\ 
 0.015625 & 4.28e-03 & 1.177 & 4.02e-04 & 0.982 & 1.06e-02 & 1.025 & 4.83e-03 & 0.962 \\
 \botrule
		\end{tabular} 
		\end{center}
	\end{table}

From now on, and unless otherwise specified, we   take values for all model 
constants as follows:
{$D^\star=0.004$}, 
$\nu = 0.4$, {$\alpha_1 = 1$}, 
{$\alpha_2= 0.5$, $\beta_1 = 1.5$, $\beta_2 = 0.1$, $K_1=46.29$, $K = 0.1429$, $G = 5.7143$}, $b = 0.111$. 
And the approximation of displacement-pressure will be restricted only to the MINI element  and to continuous and piecewise linear elements for calcium and IPR (that is, $\bV_j\times Q_j\times\Phi_j\times\Psi_j$ in  \eqref{eq:FEspaces}).

\subsection{Test 1. Calcium waves on a fixed domain}

The initial conditions are as in \eqref{eq:initial}, in particular, for Tests 1-3 the initial condition for $h$ is the steady state and 
for calcium consists of a Gaussian on the domain 
centre and the steady state elsewhere, as follows  
$${c_0(\bx) = c_s + 6 c_s \exp(-200(x^2+y^2)), \quad h_0(\bx) = \frac{1}{1+c_s^2}},$$ 
with $c_s = 0.14504$, for the case of $\mu = 0.284$, 
$c_s = 0.18572$, for the case of $\mu = 0.288468$,
$c_s = 0.55633$, for the case of $\mu = 0.3$, and 
$c_s = 0.84794$, for the case of $\mu = 0.35$.
The amplitude of this initial disturbance 
is sufficient to induce a calcium wave. 
The boundary conditions 
for the 2D cases are \eqref{bc:pureTraction}, whereas for the 3D cases 
we will use either \eqref{bc:Gamma}-\eqref{bc:Sigma} or \eqref{bc:pureTraction}.  
In all cases we use fixed timesteps dictated by the reaction kinetics. 

\begin{figure}[!t]
\begin{center}
\subfigure[]{\includegraphics[width=0.24\textwidth]{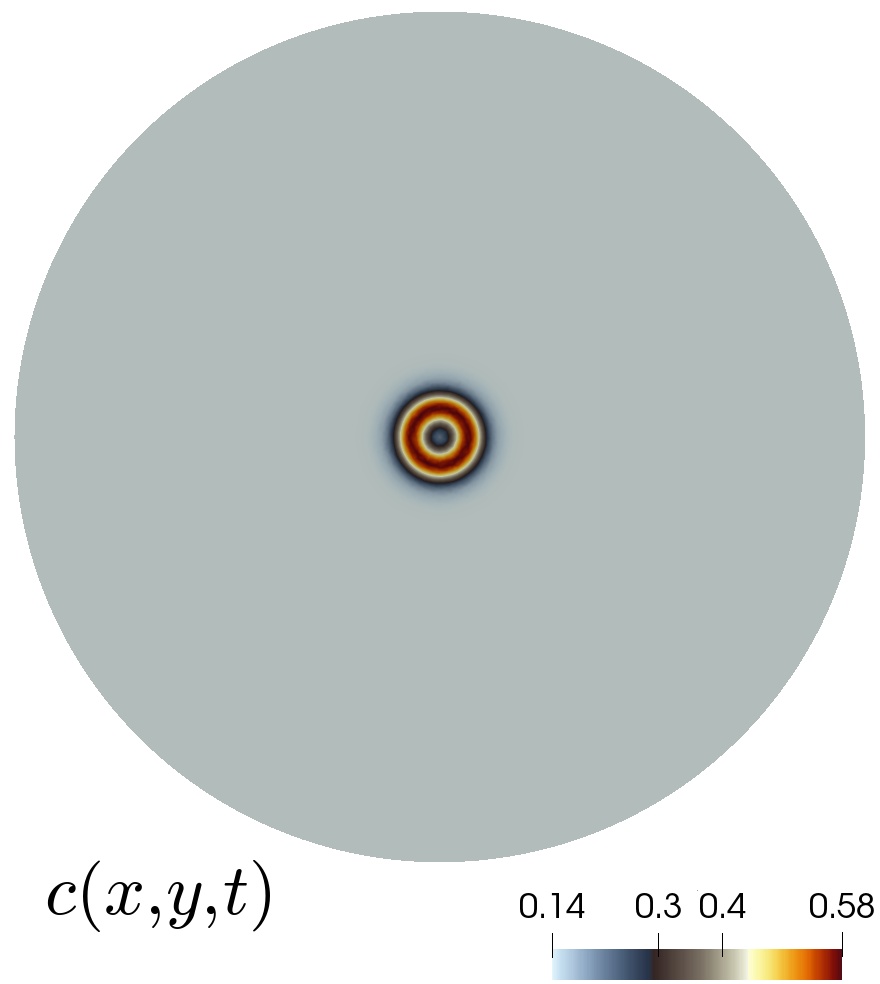}}
\subfigure[]{\includegraphics[width=0.24\textwidth]{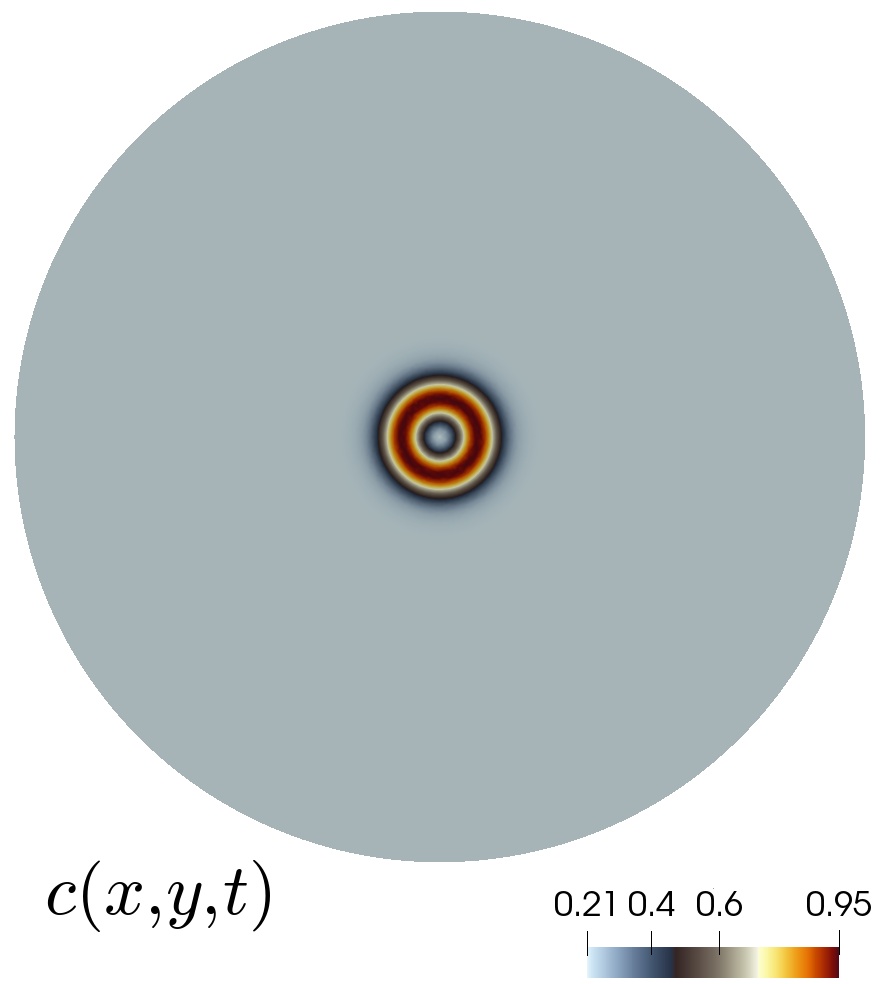}}
\subfigure[]{\includegraphics[width=0.24\textwidth]{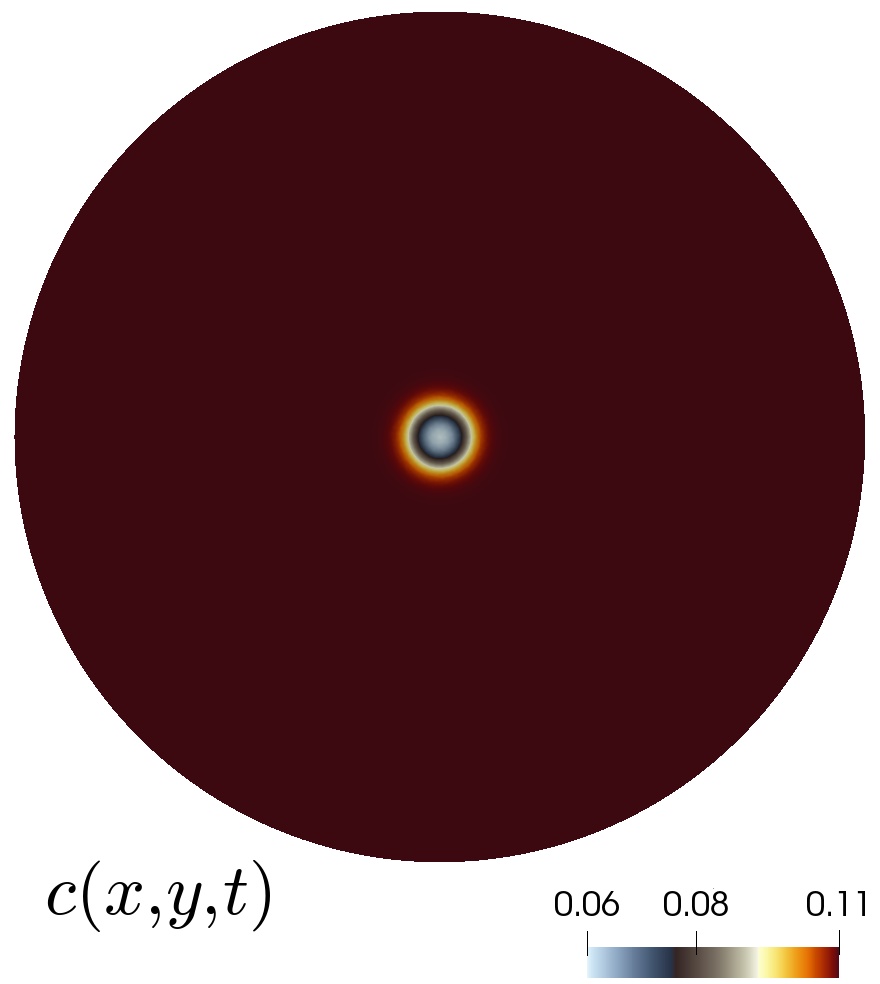}}
\subfigure[]{\includegraphics[width=0.24\textwidth]{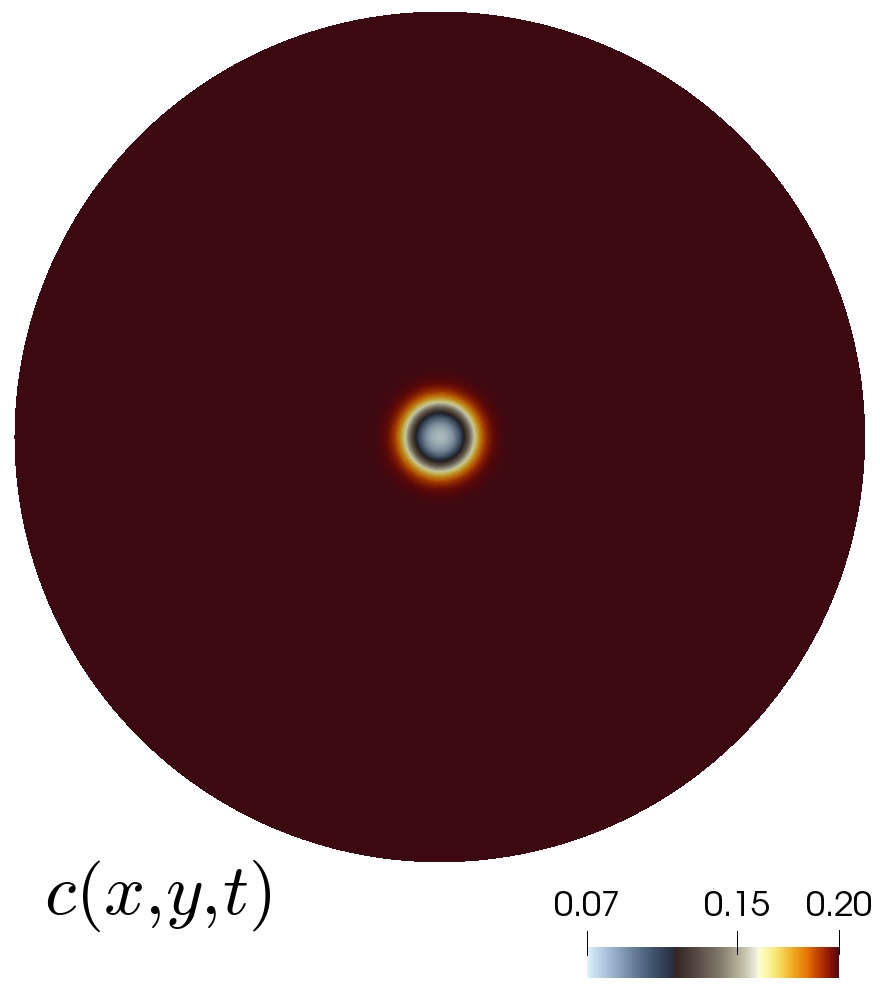}}\\
\subfigure[]{\includegraphics[width=0.24\textwidth]{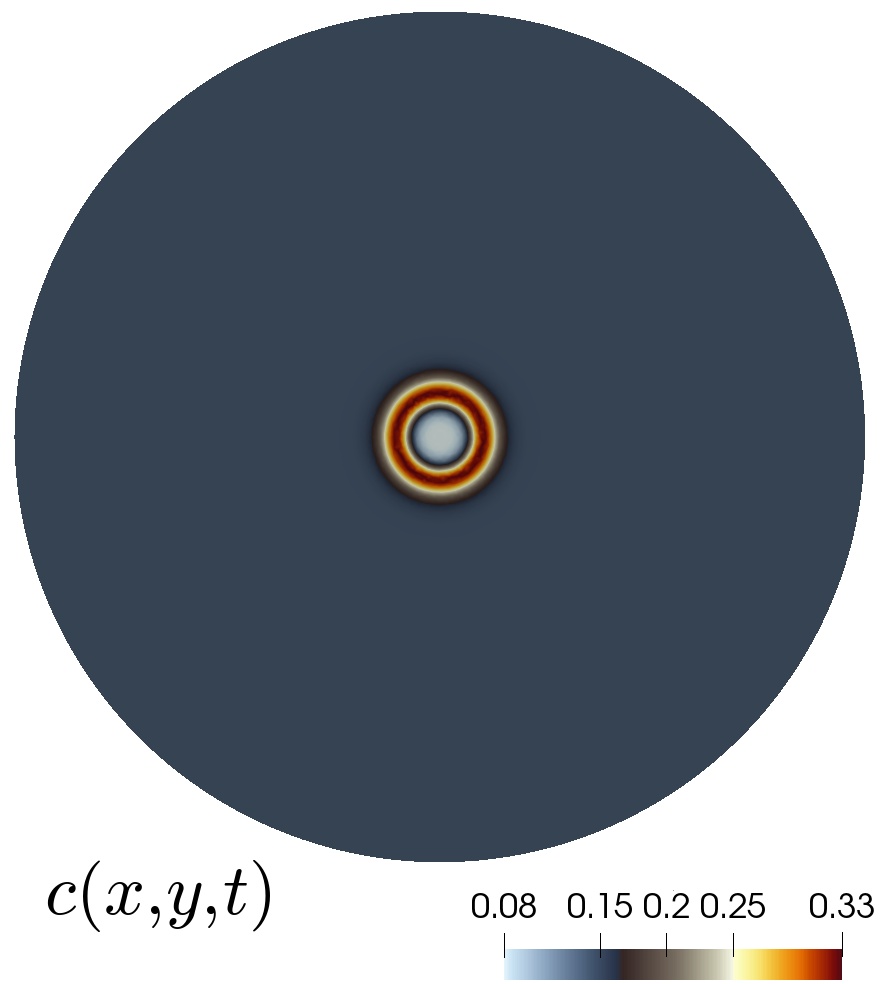}}
\subfigure[]{\includegraphics[width=0.24\textwidth]{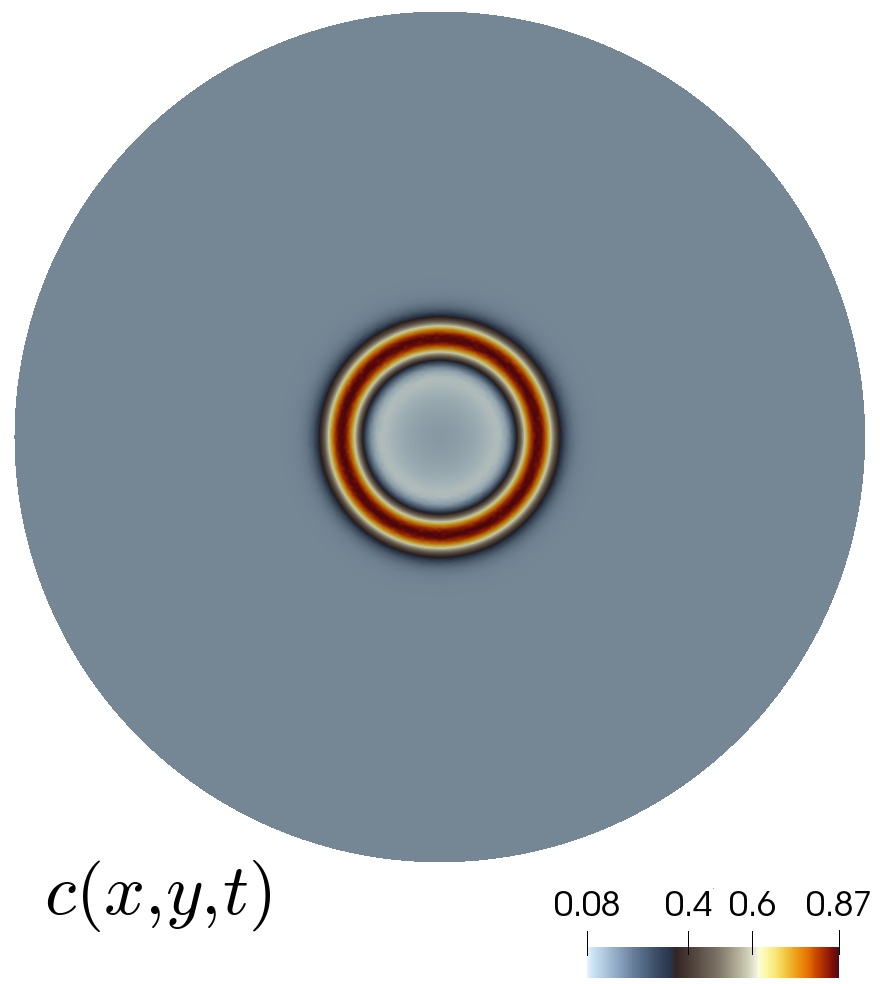}}
\subfigure[]{\includegraphics[width=0.24\textwidth]{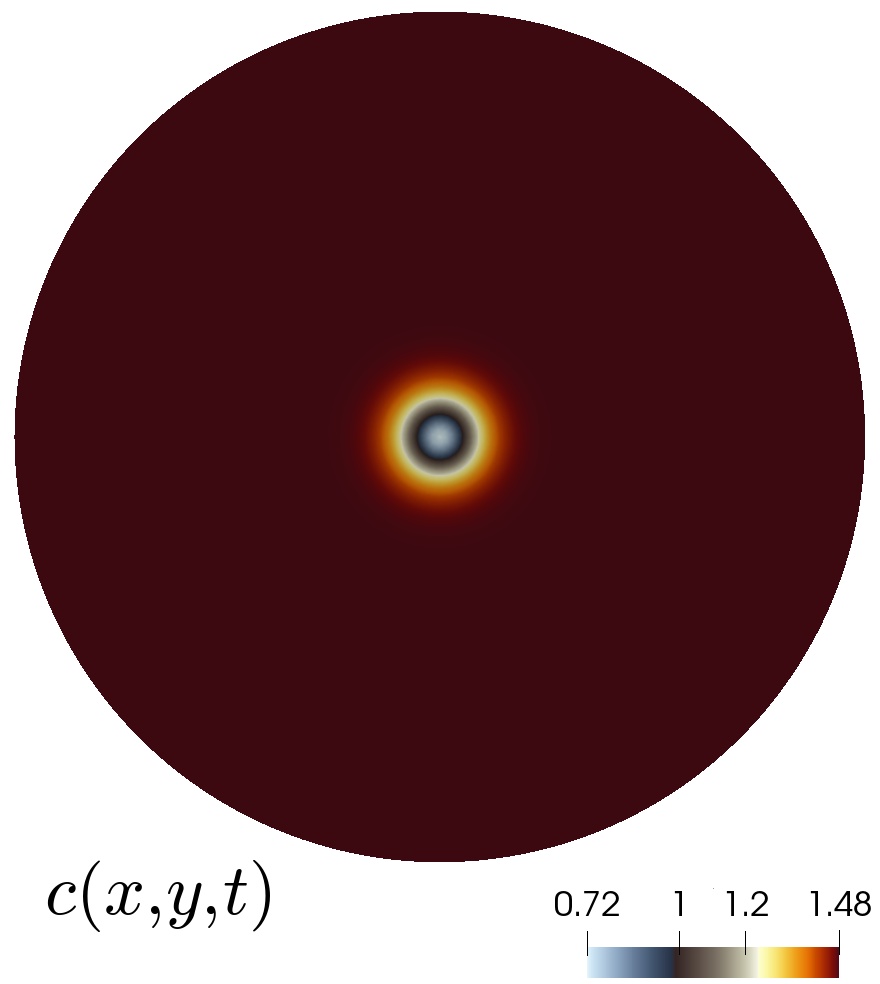}}
\subfigure[]{\includegraphics[width=0.24\textwidth]{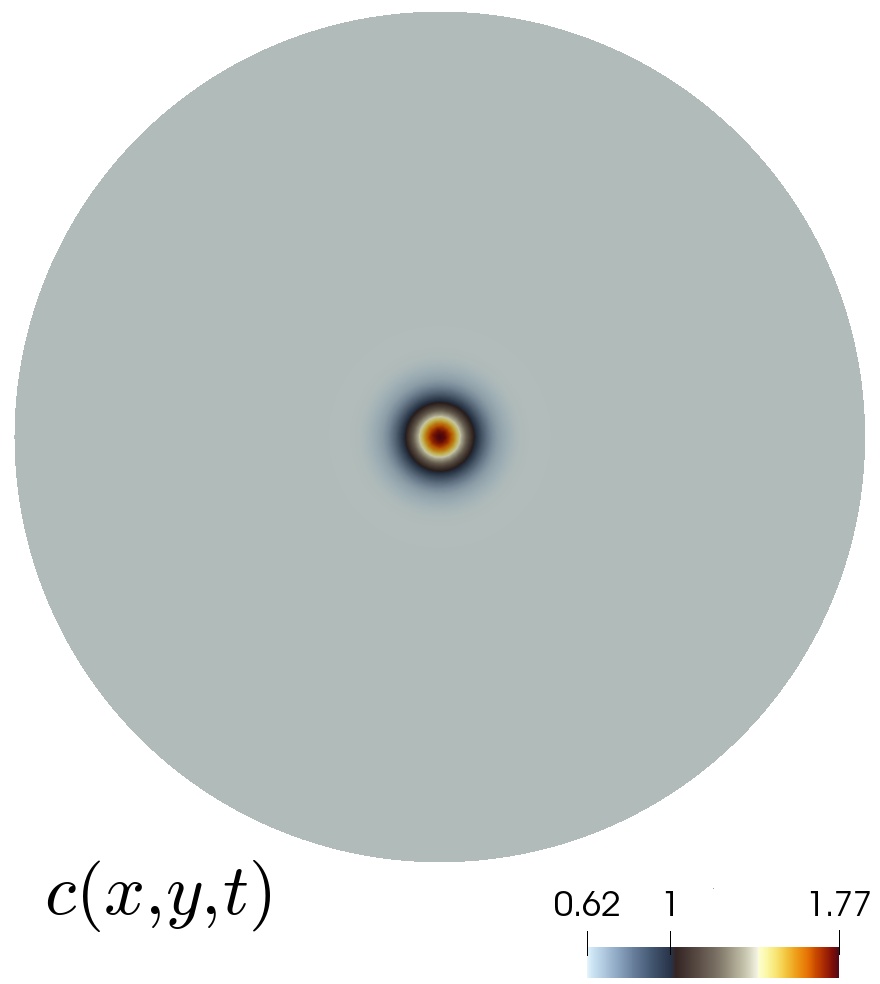}}\\
\subfigure[]{\includegraphics[width=0.24\textwidth]{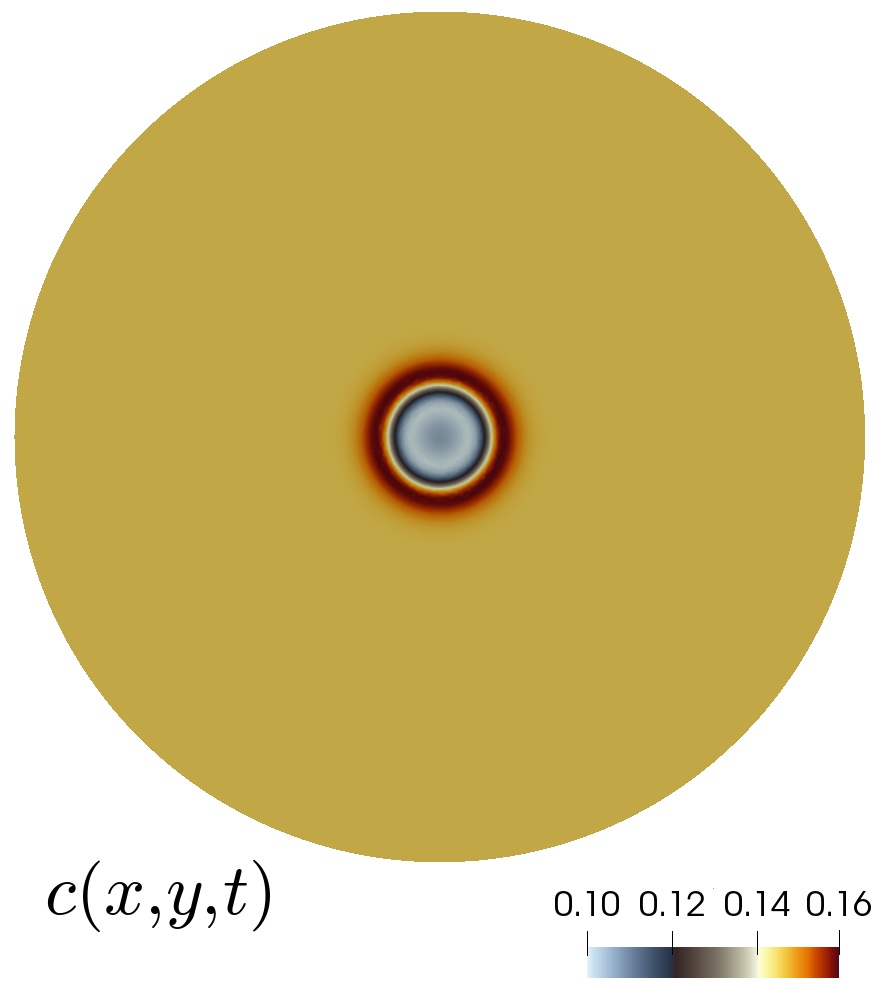}}
\subfigure[]{\includegraphics[width=0.24\textwidth]{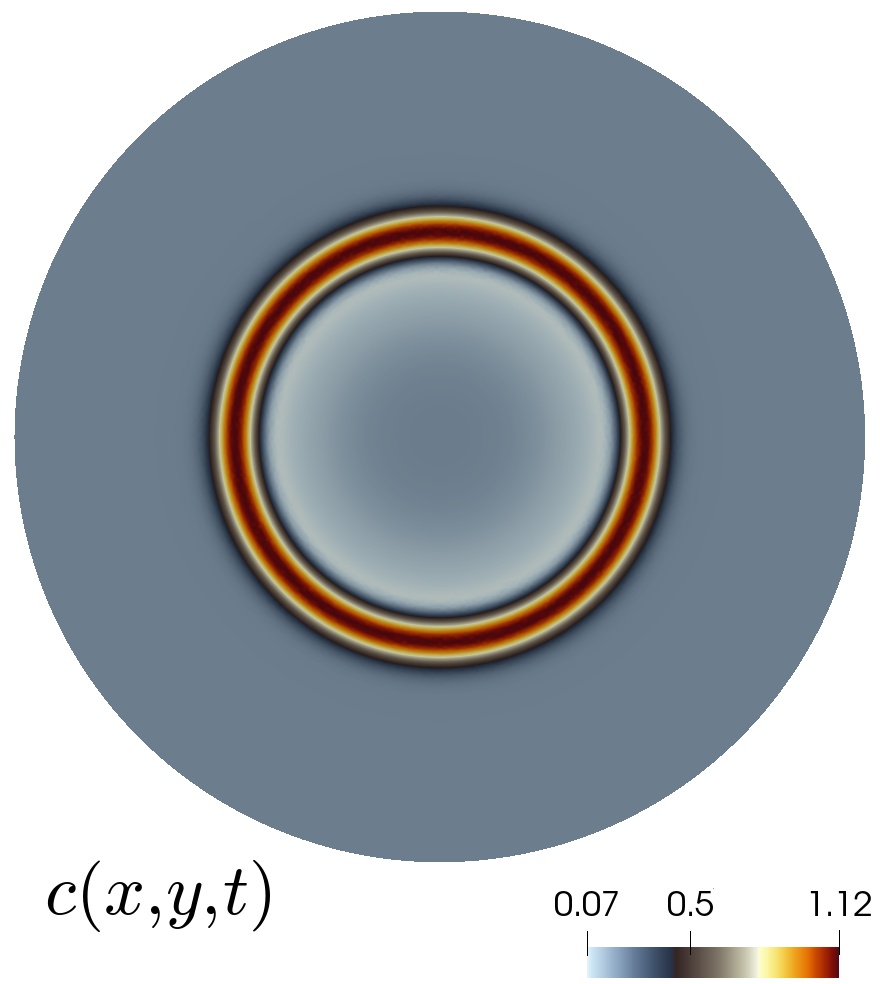}}
\subfigure[]{\includegraphics[width=0.24\textwidth]{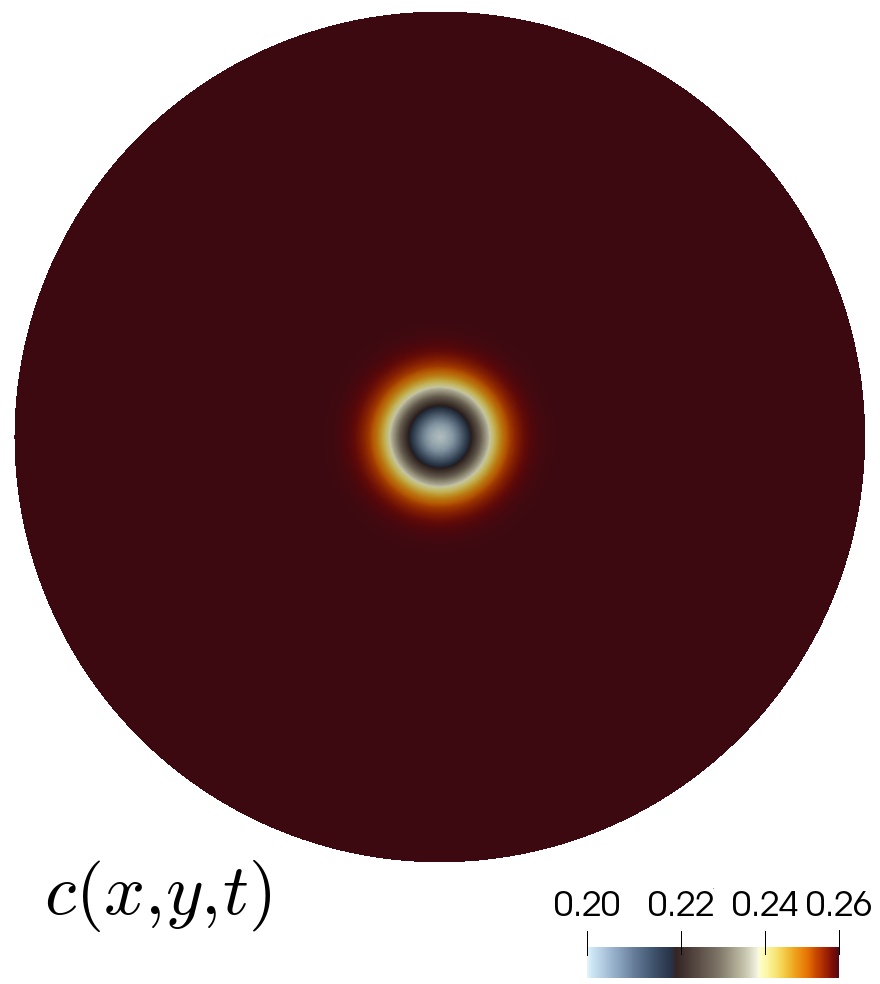}}
\subfigure[]{\includegraphics[width=0.24\textwidth]{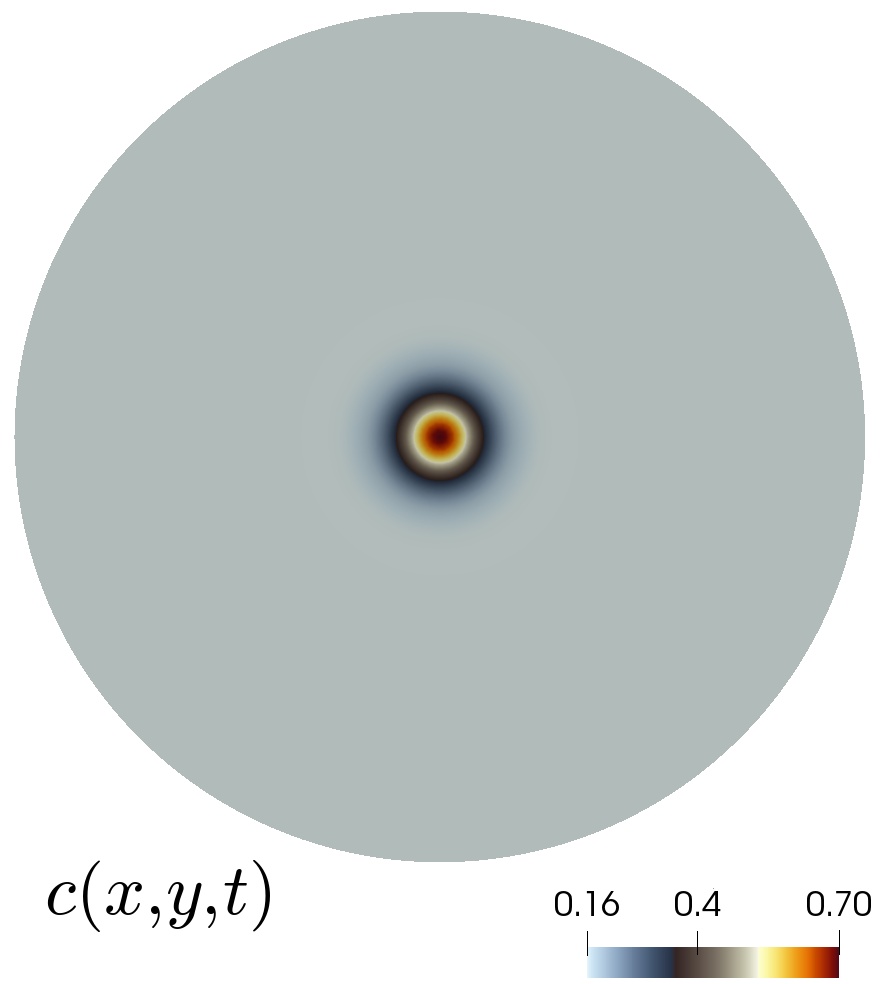}}
\end{center}
 
\caption{Test 1. Transients of calcium starting from an initial spark at the domain centre, using 
$\mu=0.284$ (a,e,i), $\mu=  0.288468$ (b,f,j), $\mu=0.3$ (c,g,k), and $\mu = 0.35$ (d,h,l).  
Snapshots at $t=2,5,10$ (top, middle, and bottom), 
except for panels (e),(i) that correspond to $t=3,4$, after which 
the travelling wave diffuses away. }\label{fig:test01}
\end{figure}

\begin{figure}[!t]
\begin{center}
\subfigure[]{\includegraphics[width=0.32\textwidth]{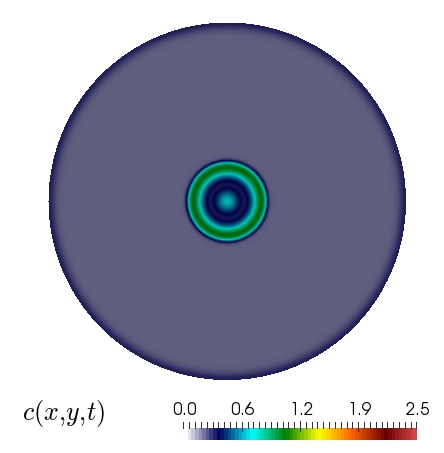}}
\subfigure[]{\includegraphics[width=0.32\textwidth]{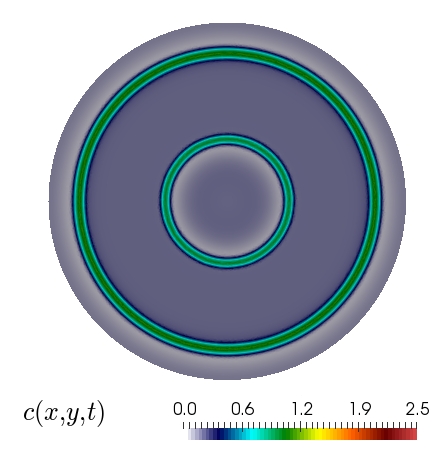}}
\subfigure[]{\includegraphics[width=0.32\textwidth]{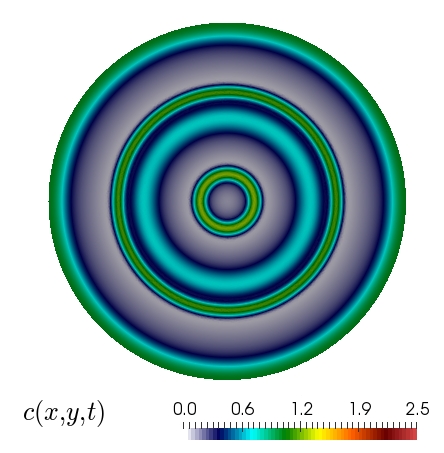}}\\
\subfigure[]{\includegraphics[width=0.32\textwidth]{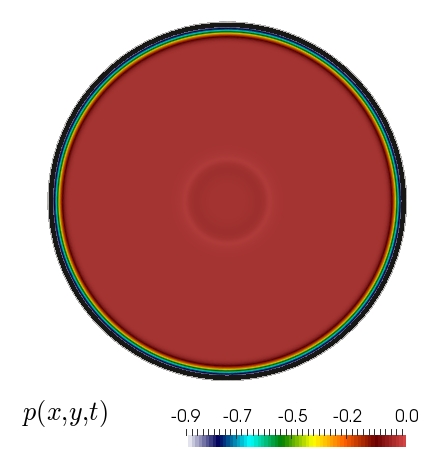}}
\subfigure[]{\includegraphics[width=0.32\textwidth]{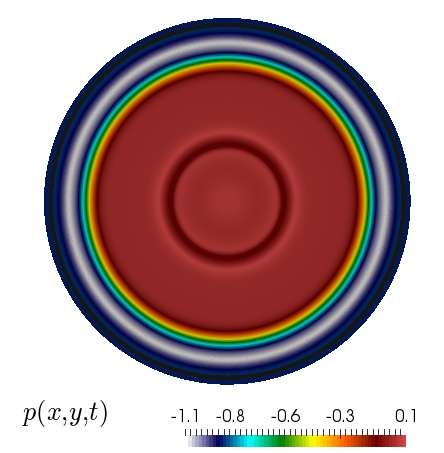}}
\subfigure[]{\includegraphics[width=0.32\textwidth]{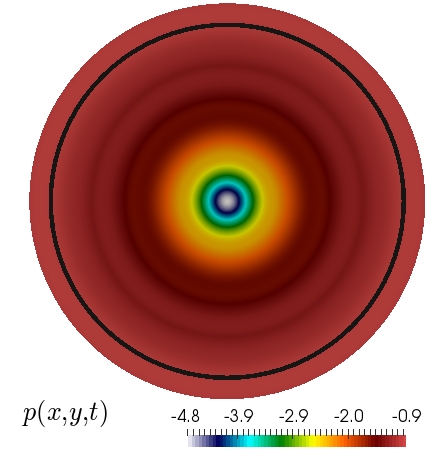}}
\end{center}
 
\caption{Test 2A. Transients of calcium concentration on the undeformed domain and 
solid pressure plotted on the deformed configuration for  
$\mu=  0.288468$ and $\lambda = 0.35$.}\label{fig:test02a}
\end{figure}

Let us consider as  domain the disk centred at $(0,0)$ with radius 2.5. The system  
evolves in time until $t_{\text{final}}=10$. For this case the 
coupling constant $\lambda$ is taken simply as zero. We observe that a single calcium pulse 
propagates from the centre of the domain towards the 
boundary. The amount of IPR in the system also has an influence on the spatio-temporal 
patterns of calcium. Examples of concentrations at three different times are depicted in Figure~\ref{fig:test01}, 
using four different values of IPR concentration, $\mu=0.284,0.288468,0.3,0.35$, respectively. For the first  value one 
expects travelling waves to eventually vanish, before reaching the boundary of the domain.

\subsection{Test 2. Effects of mechanochemical coupling} 
Setting now a different value of the coupling constant $\lambda = 0.35$ results in the behaviour 
illustrated in Figure~\ref{fig:test02a}. No other parameters have been changed and we can 
see significant effects from the mechanochemical coupling. In particular, we observe 
that additional secondary waves are initiated from the boundary (which is the part of the domain 
where most pronounced deformation and dilation occurs), and these are sustained over a longer 
time, suggesting a long-range contraction response.

\begin{figure}[htbp!]
\begin{center}
\subfigure[]{\includegraphics[width=0.24\textwidth]{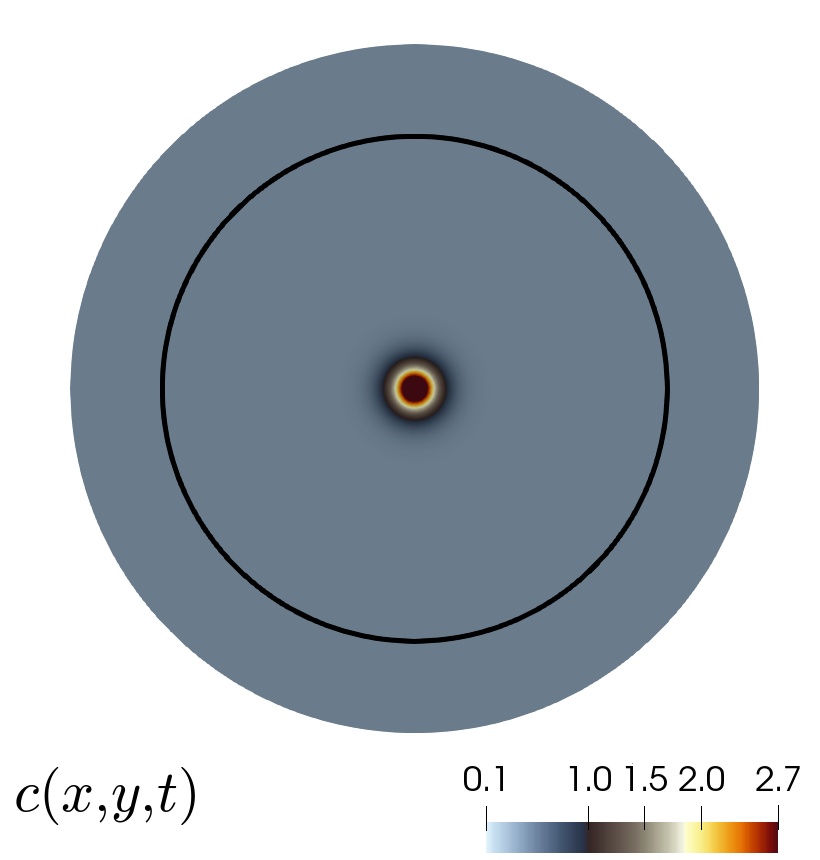}}
\subfigure[]{\includegraphics[width=0.24\textwidth]{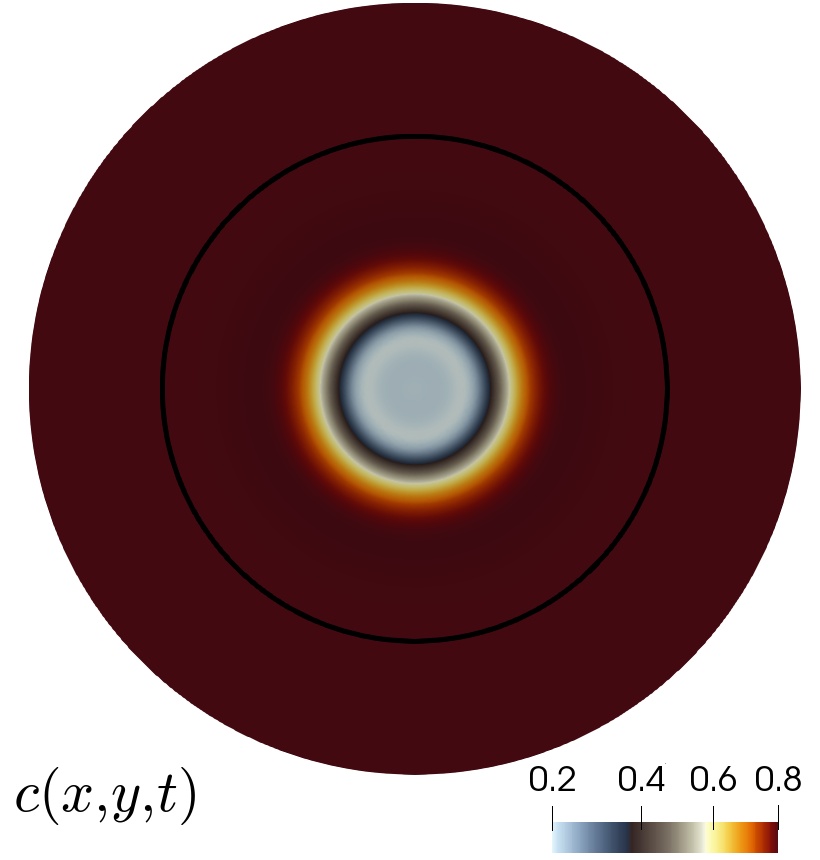}}
\subfigure[]{\includegraphics[width=0.24\textwidth]{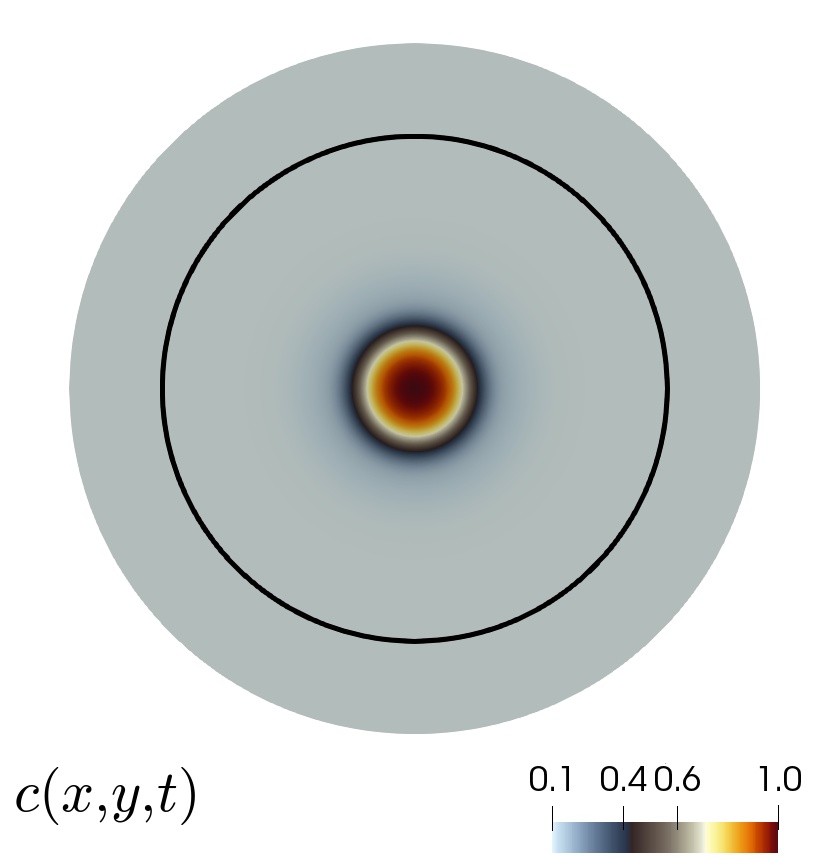}}
\subfigure[]{\includegraphics[width=0.24\textwidth]{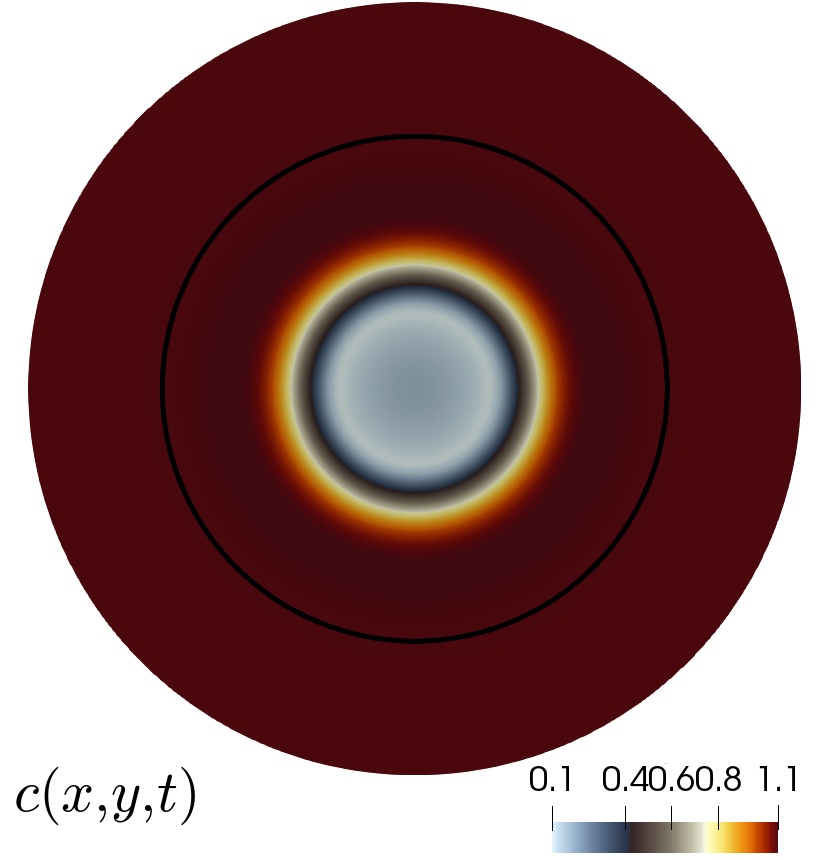}}\\
\subfigure[]{\includegraphics[width=0.24\textwidth]{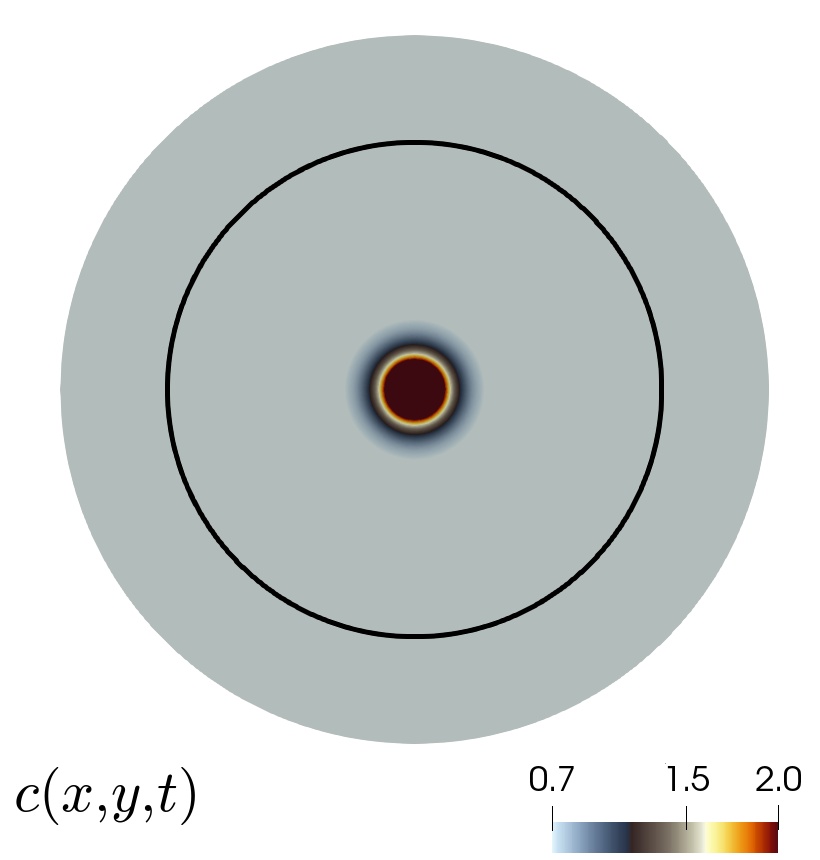}}
\subfigure[]{\includegraphics[width=0.24\textwidth]{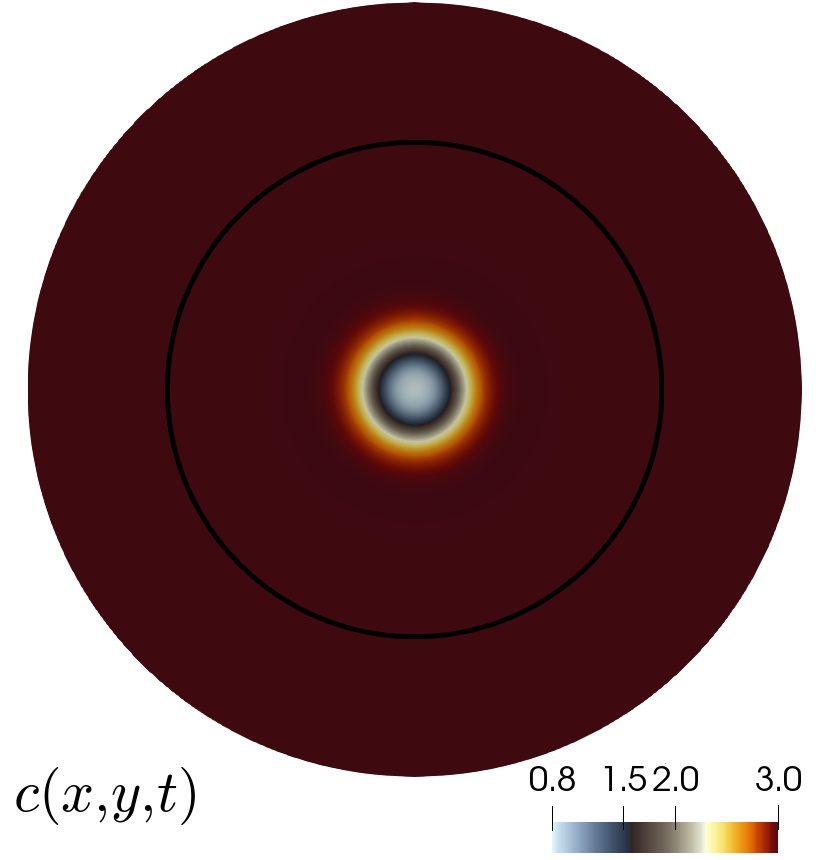}}
\subfigure[]{\includegraphics[width=0.24\textwidth]{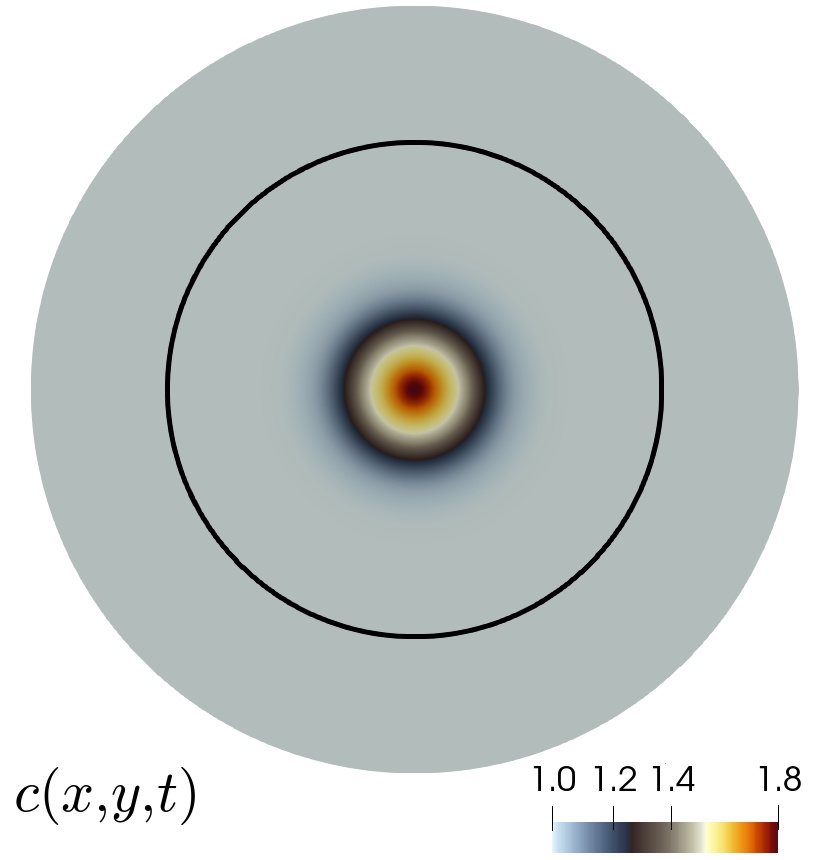}}
\subfigure[]{\includegraphics[width=0.24\textwidth]{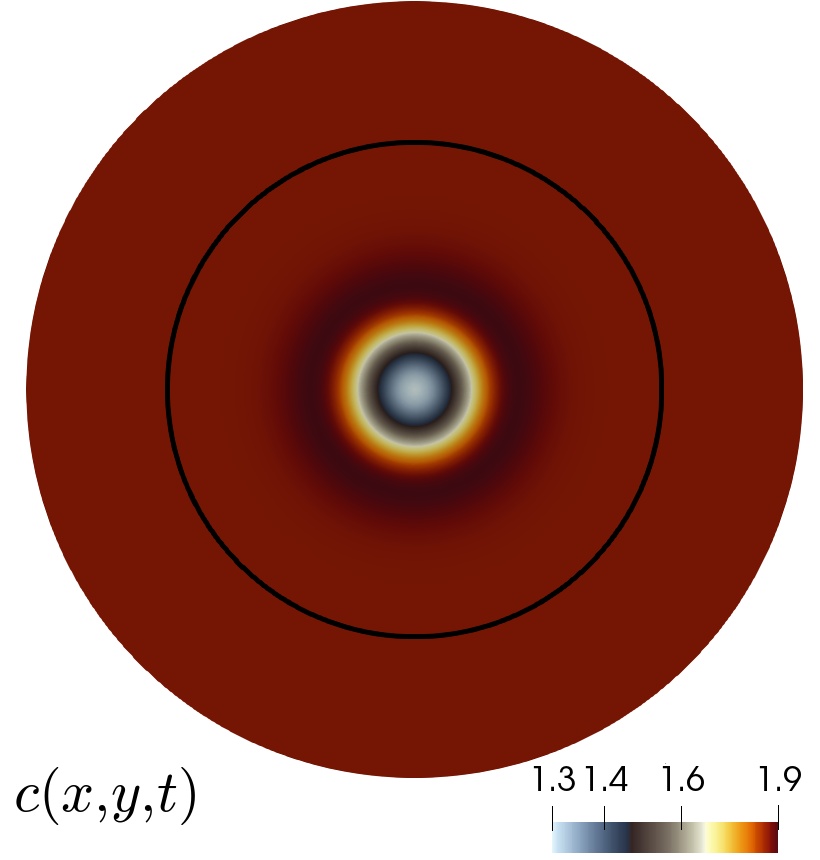}}\\
\subfigure[]{\includegraphics[width=0.24\textwidth]{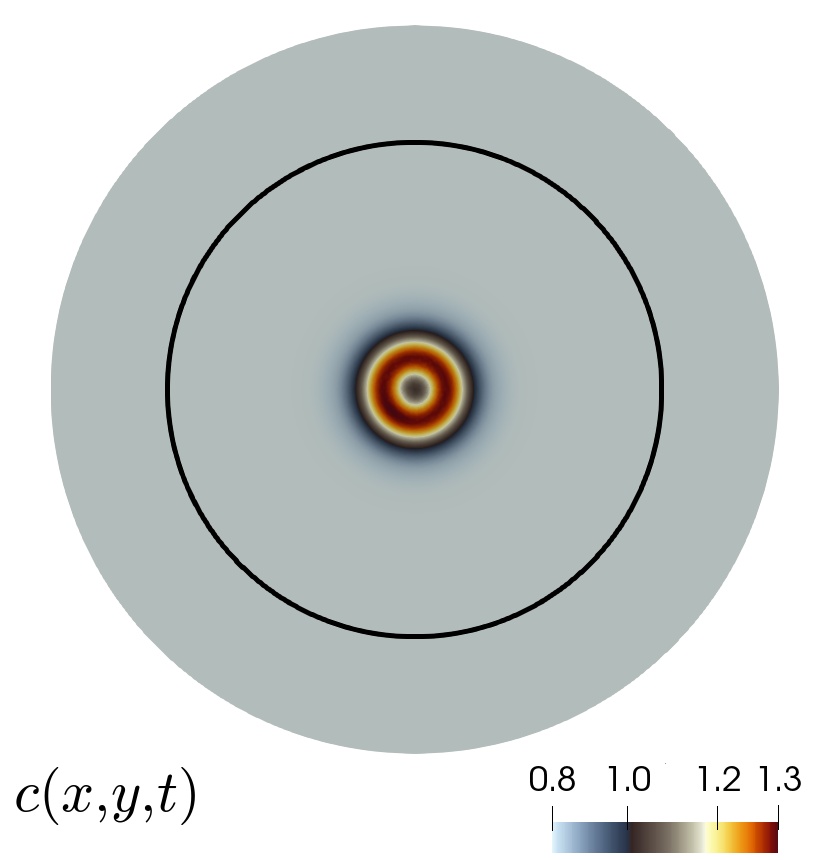}}
\subfigure[]{\includegraphics[width=0.24\textwidth]{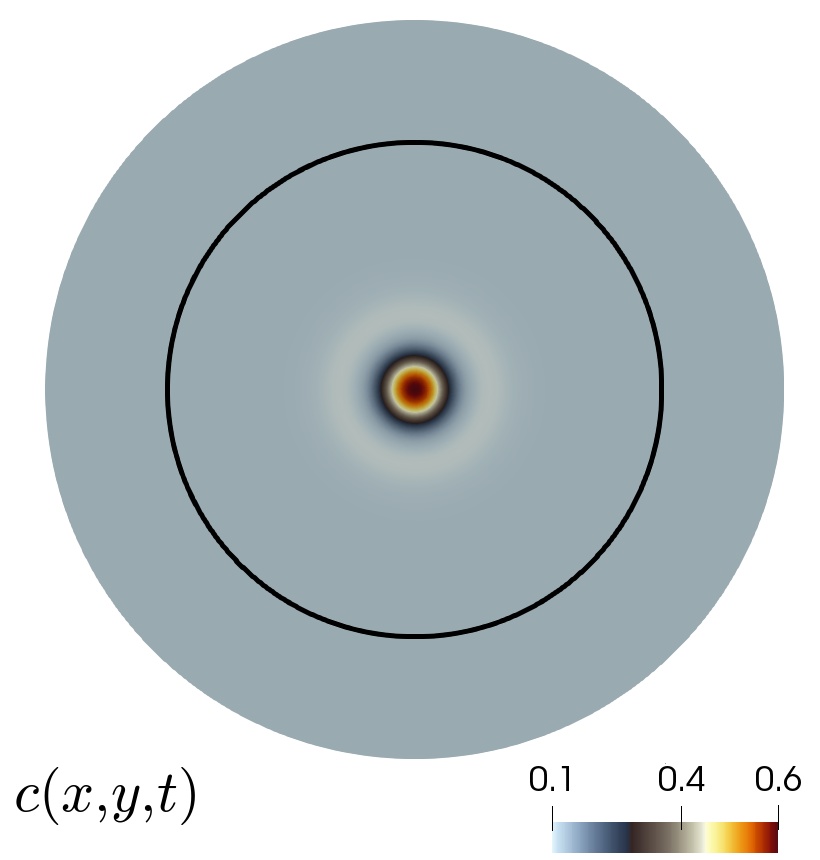}}
\subfigure[]{\includegraphics[width=0.24\textwidth]{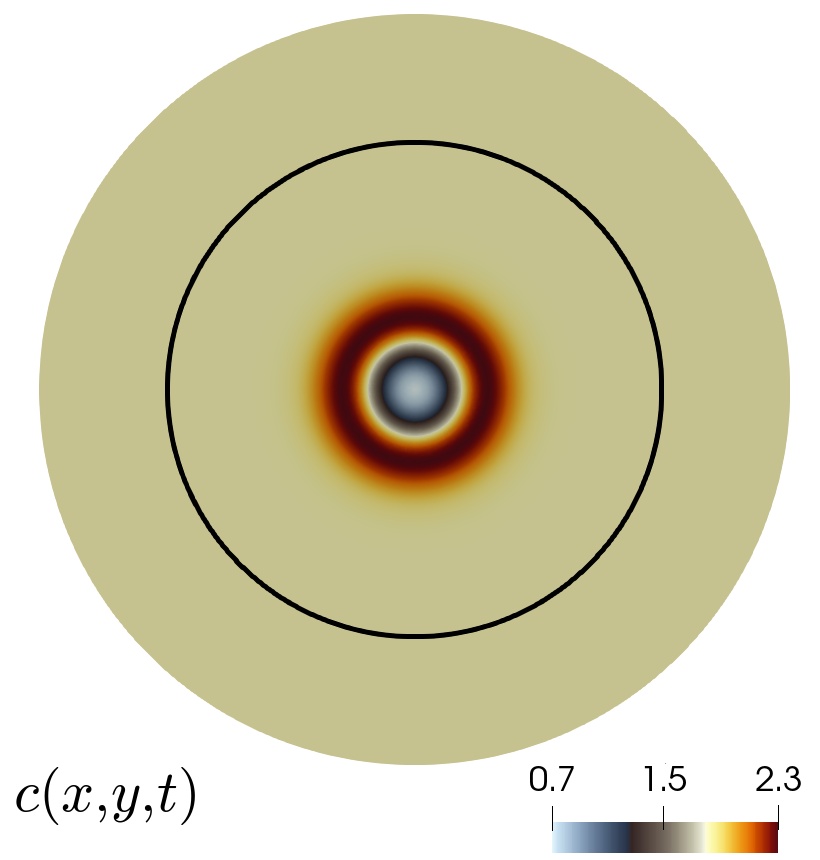}}
\subfigure[]{\includegraphics[width=0.24\textwidth]{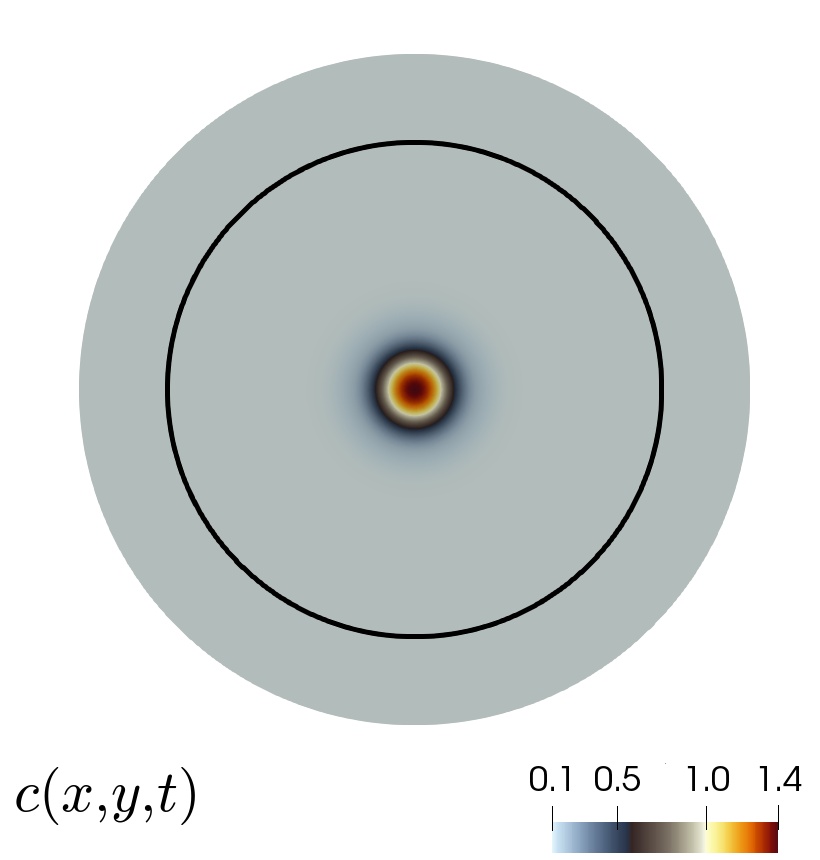}}\\
\subfigure[]{\includegraphics[width=0.24\textwidth]{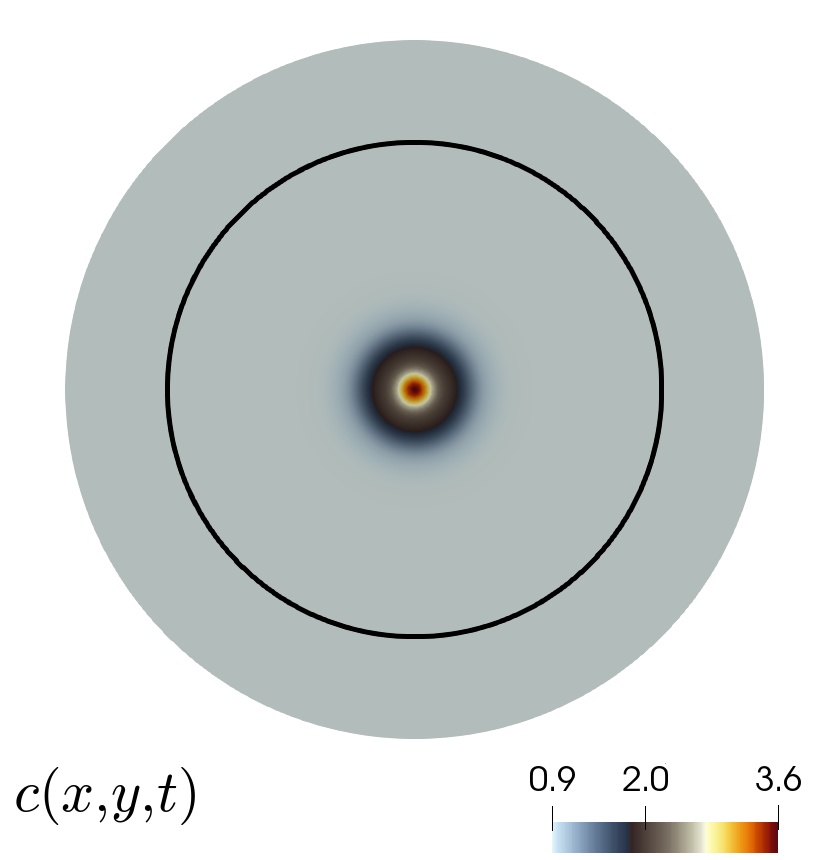}}
\subfigure[]{\includegraphics[width=0.24\textwidth]{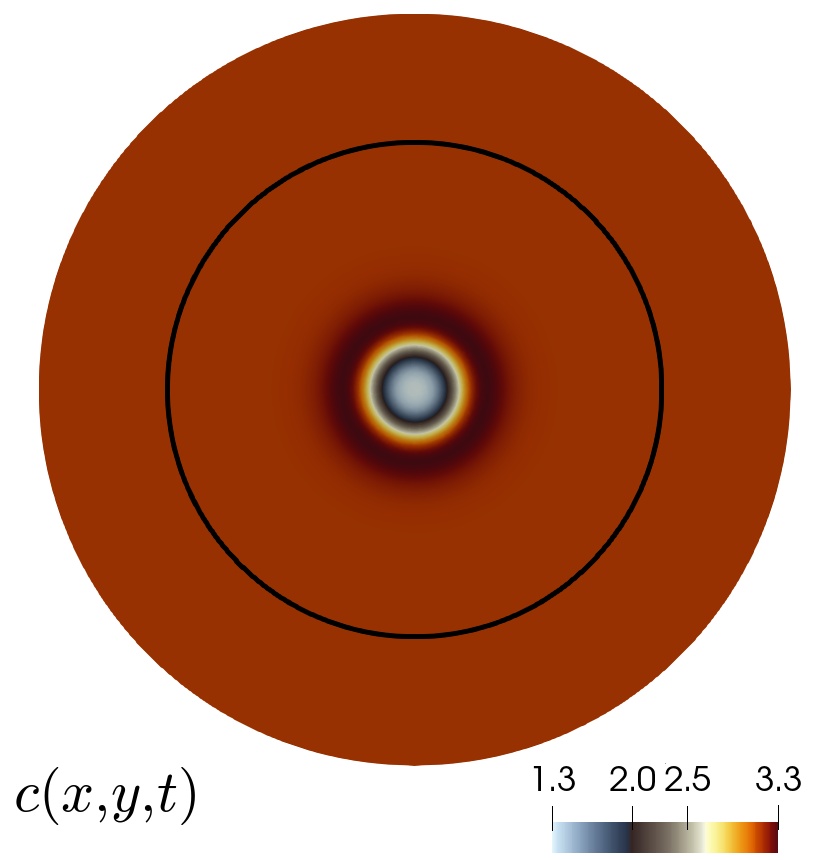}}
\subfigure[]{\includegraphics[width=0.24\textwidth]{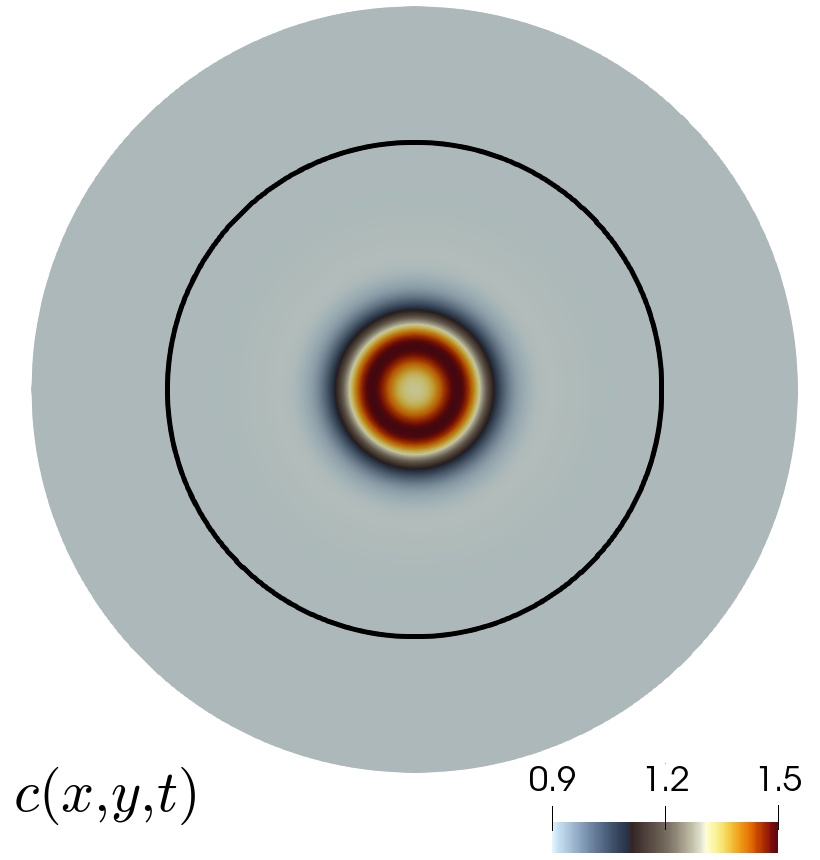}}
\subfigure[]{\includegraphics[width=0.24\textwidth]{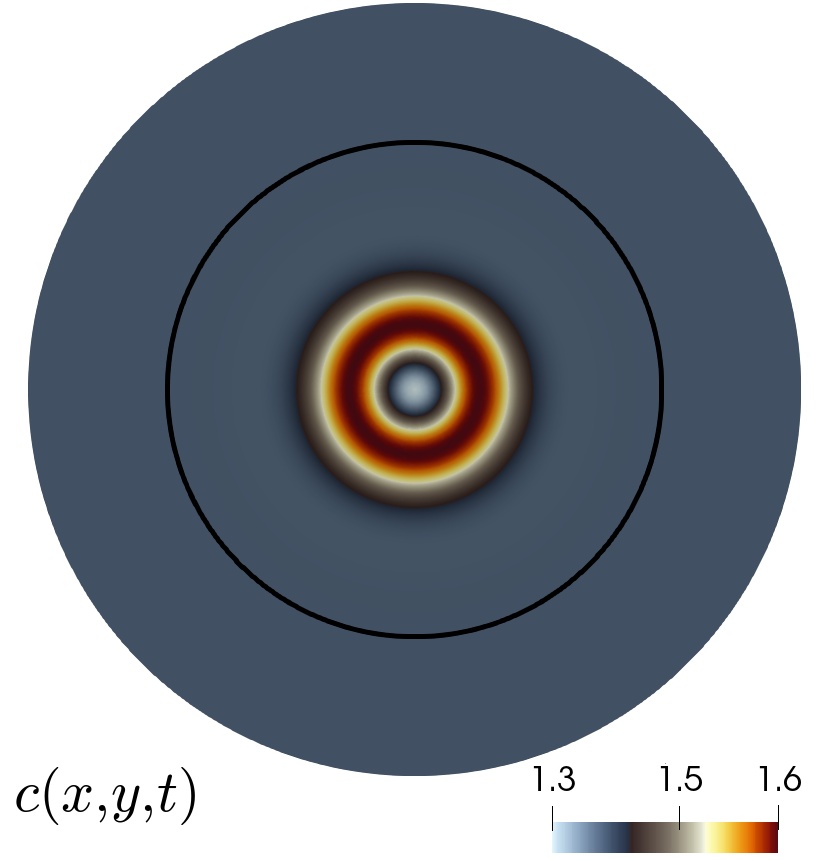}}
\end{center}
 
\caption{Test 2B. Transients of calcium concentration on the deformed domain. 
Case I (a,b,c,d), Case IV (e,f,g,h), Case VI (i,j,k,l), Case VIII (m,n,o,p). 
The black circle depicts the boundary of the undeformed domain.}\label{fig:test02b}
\end{figure}

We also select some other parameter regimes that, according to the 1D theory, would produce qualitatively interesting 
behaviour  such as oscillatory calcium transients. In connection with Test 1 above, note that in the uncoupled scenario of $\lambda = 0$, 
the cases $\mu = 0.288468$  and $\mu = 0.3$ generate a solitary wave and  a periodic wave train, respectively. On the other hand, if we fix $\mu = 0.288468$ and vary $\lambda$ we distinguish between 
the following cases: 
 Case I: $\lambda = 0.1$ and an equilibrium calcium state predicted by the 1D theory 
of $c_s = 0.4850$ produces a periodic wavetrain. 
 Case II: Setting $\lambda = 0.5$ and  $c_s = 0.6824$ gives also a periodic wavetrain with a slower period (of approximately 2.5 nondimensional time units). 
Case III: With 
$\lambda = 1.3$ and  $c_s = 0.9577$ we also find a periodic wavetrain that gradually decays, but where also the resting value elsewhere on the domain increases with time. 
 Case IV: With the largest value for the coupling parameter $\lambda =  2$ and $c_s = 1.19817$ we see that the periodic wavetrain has a similar period as before but it decays much faster. 

For these rounds it suffices to take a timestep of $\Delta t = 0.2$. If we set now a larger $\mu = 0.3$ we see that a smaller $\Delta t = 0.1$ is needed to 
capture the faster domain velocity. The values $\lambda = 0.1,c_s = 0.6034$ (Case V) also produce a periodic wavetrain, whereas increasing $\lambda$ further will still show wavetrains but decaying faster and faster (Cases VI, VII, VIII). For $\lambda=2,c_s = 1.2506$ (case VIII), we can also use a larger timestep $\Delta t = 0.2$. 
A few samples of these findings are shown in Figure~\ref{fig:test02b}.

\subsection{Test 3. Including calcium sparks} 
Assuming that calcium sparks appear on the continuum randomly, as experiments show, 
we can study their influence 
on the calcium propagation through the modification of \eqref{eq:c} as follows 
\begin{equation}\label{eq:c-modif}
  \partial_t c + \partial_t \bu \cdot \nabla c -  
	\vdiv\{ D^\star \nabla c \} 
= F(c,n,\bu) + I_f(\bx,t),
\end{equation}
where $F(\cdot,\cdot,\cdot)$ is a generic reaction term specified by the sought model (e.g. \cite{atri93}), and $I_f$ encodes 
a collection of flashes of different amplitudes, having a uniform random distribution. From a phenomenological 
perspective, if the impulses are applied sufficiently close to each other, a simple superposition phenomenon 
explains why the intensity of individual local pulses contributes to form a synchronous 
wave propagating front. These impulses are implemented as point sources (Dirac delta functions on a linear 
variational form) localised at randomly distributed points in $\Omega$, with amplitude equal to the constant 
calcium equilibrium found for each parameter configuration, and switched on at a given frequency. 
They are projected onto the appropriate finite element space (the one used for calcium approximation) and 
applied after assembling of the residual that constitutes the right-hand side of the tangent system \eqref{eq:jac}, 
at each Newton iteration.

\begin{figure}[!t]
\begin{center}
\subfigure[]{\includegraphics[width=0.24\textwidth]{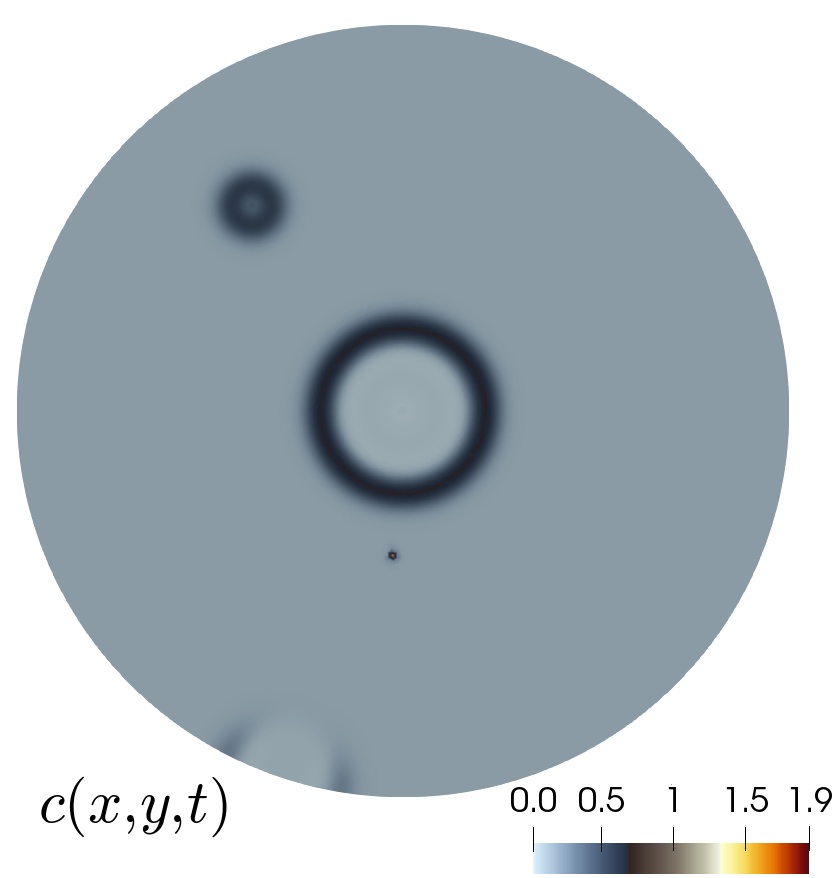}}
\subfigure[]{\includegraphics[width=0.24\textwidth]{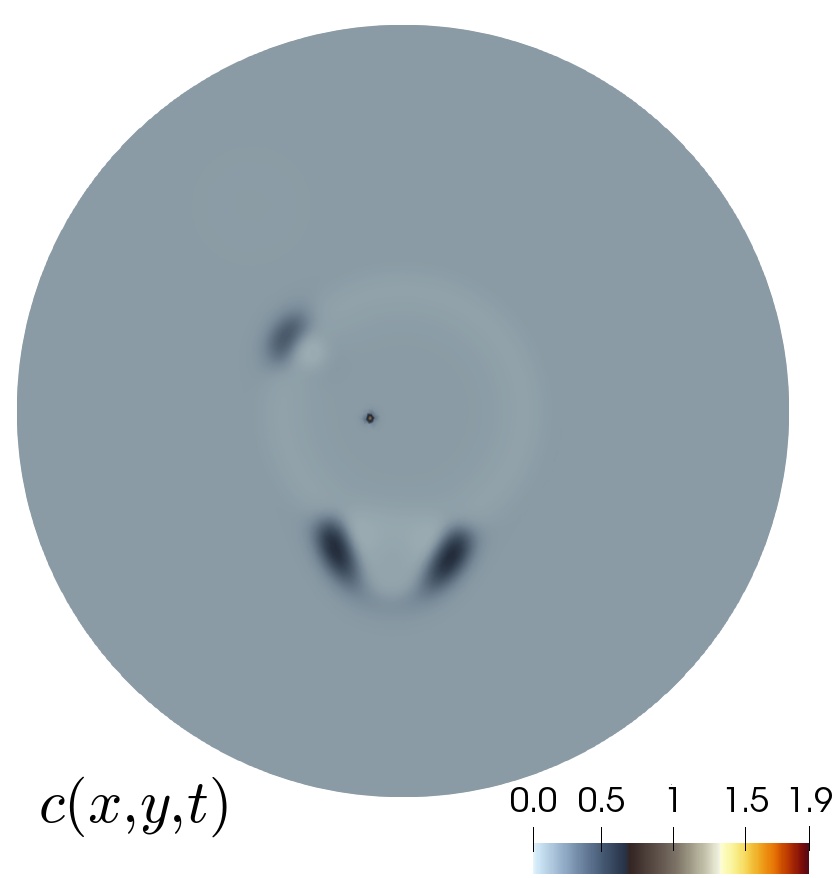}}
\subfigure[]{\includegraphics[width=0.24\textwidth]{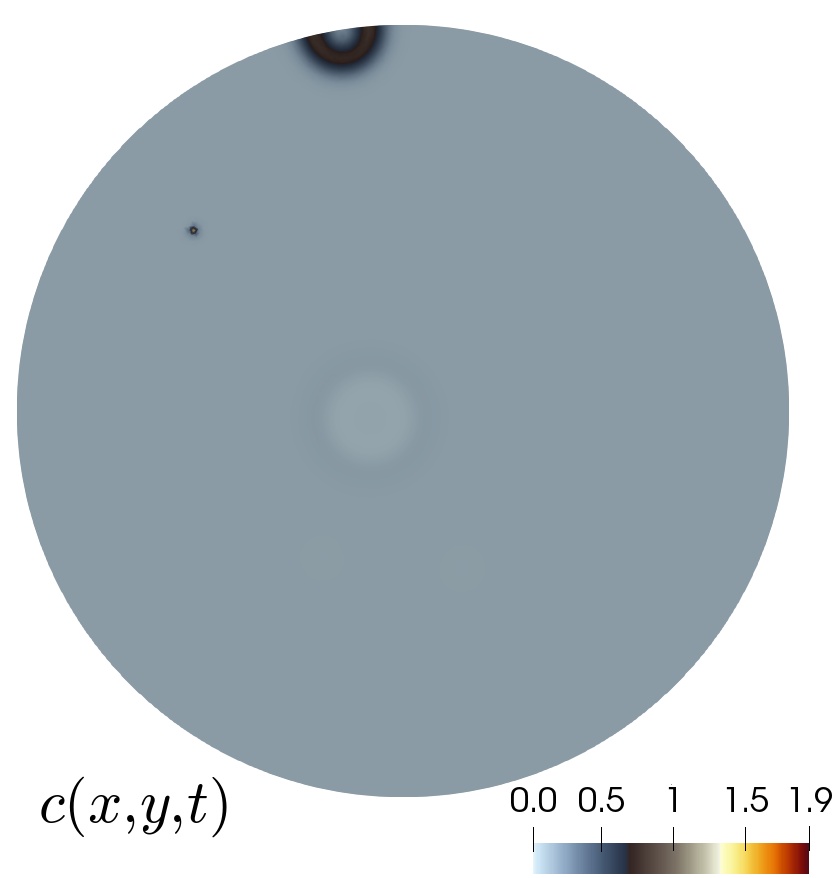}}
\subfigure[]{\includegraphics[width=0.24\textwidth]{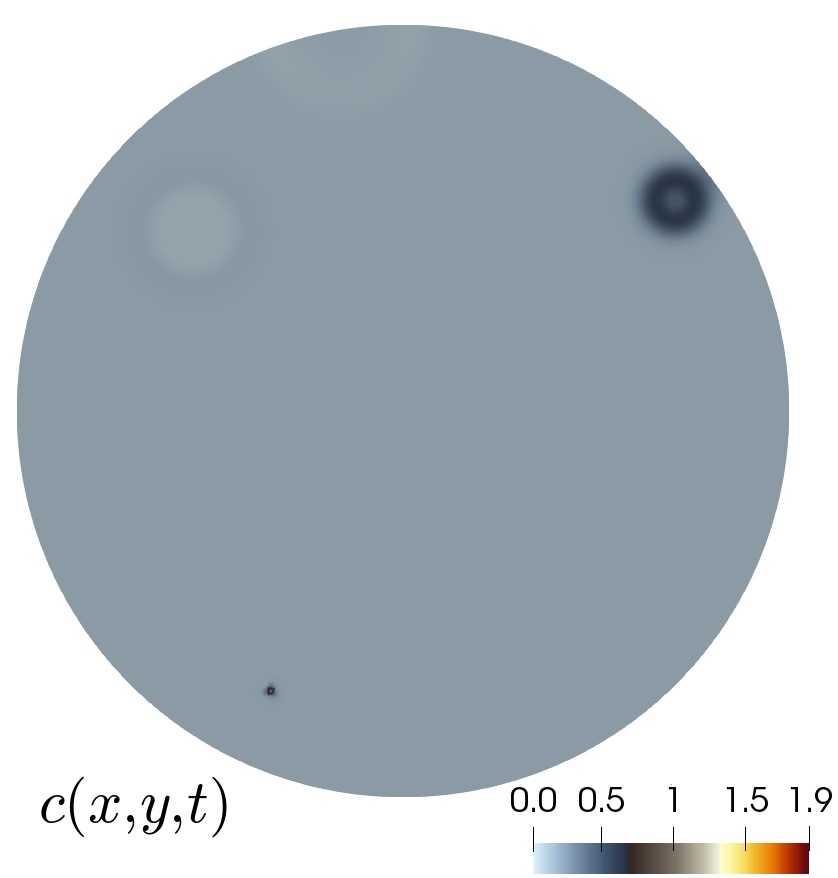}}\\
\subfigure[]{\includegraphics[width=0.24\textwidth]{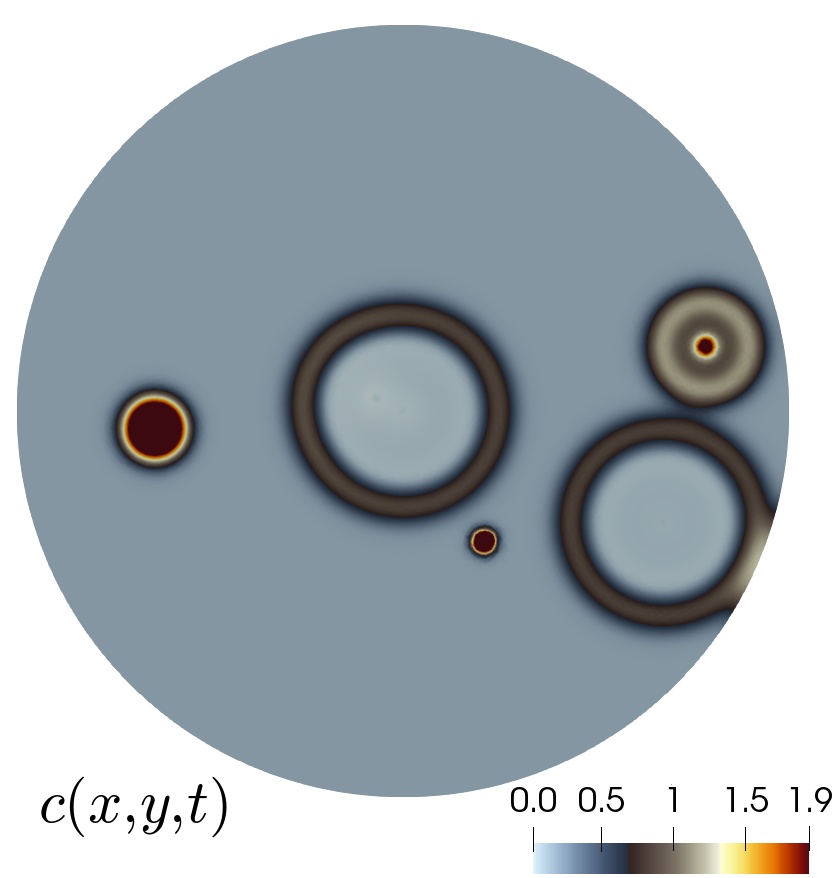}}
\subfigure[]{\includegraphics[width=0.24\textwidth]{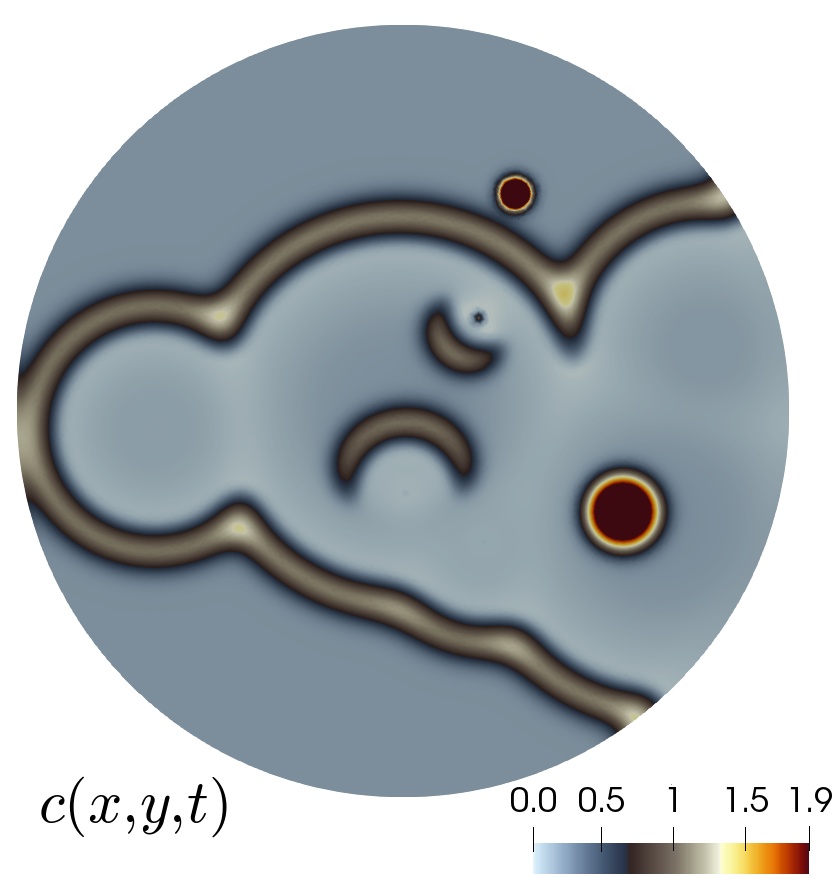}}
\subfigure[]{\includegraphics[width=0.24\textwidth]{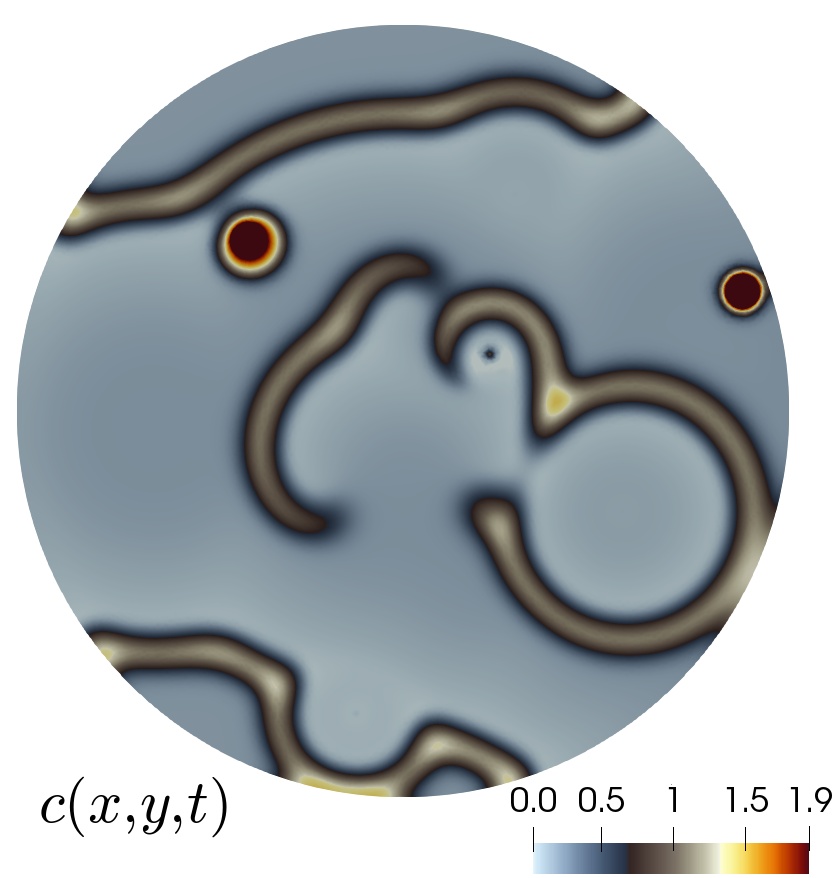}}
\subfigure[]{\includegraphics[width=0.24\textwidth]{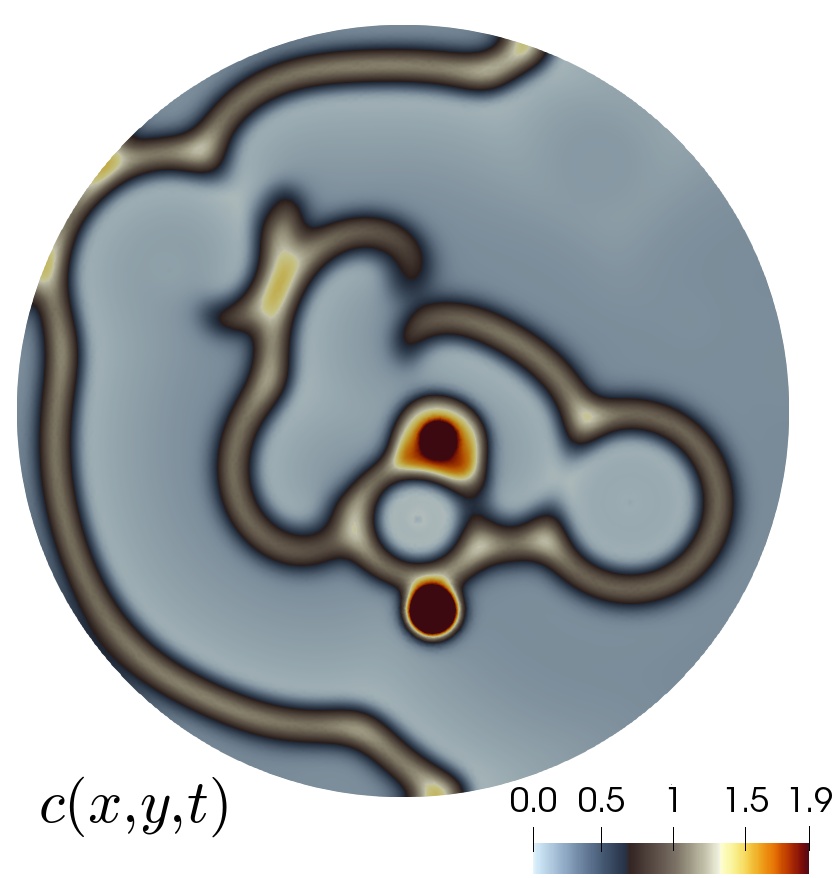}}
\end{center}
 
\caption{Test 3A. Transients of calcium concentration when sparks are initiated 
randomly in space and time with low (a,b,c) and high (d,e,f) frequency.}\label{fig:test03sparks}
\end{figure}

In the case of $\lambda = 0$, although the advection effect is still turned on, the calcium transients do not 
exhibit substantial differences with respect to their counterparts on a fixed domain (that is, when also the 
calcium-dependent contraction is turned off, with $ \beta_2 = 0$). Also, the structure of the waves (when regarded 
locally) does not change with respect to the observed behaviour in the case of a single calcium stimulus (meaning 
that the nature of 
solitary waves and 
periodic wavetrains also occurs in the 
case of spark phenomena). However, the nucleation does plays a role, as anticipated above. 
We show in Figure~\ref{fig:test03sparks} 
snapshots of the calcium concentration at four time instants and compare the effect of increasing the frequency of the 
applied stimuli, and see that isolated point sources initiate a propagating wave but eventually 
decay when the frequency is low enough (this is the expected trend observed also in the one-dimensional model from \cite{kaouri19}), while as soon as the frequency is increased from one each 20 time steps to one each 5 time steps, the solitary waves form synchronous wavefronts that are sustained in time.

\begin{figure}[!t]
\begin{center}
\subfigure[]{\includegraphics[width=0.24\textwidth]{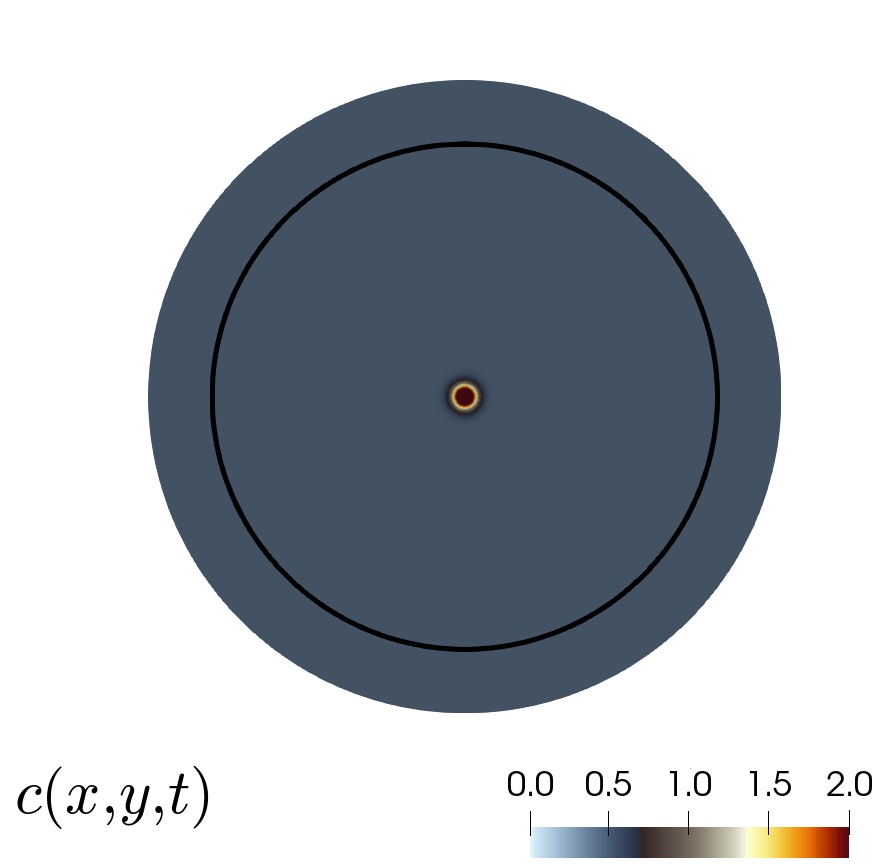}}
\subfigure[]{\includegraphics[width=0.24\textwidth]{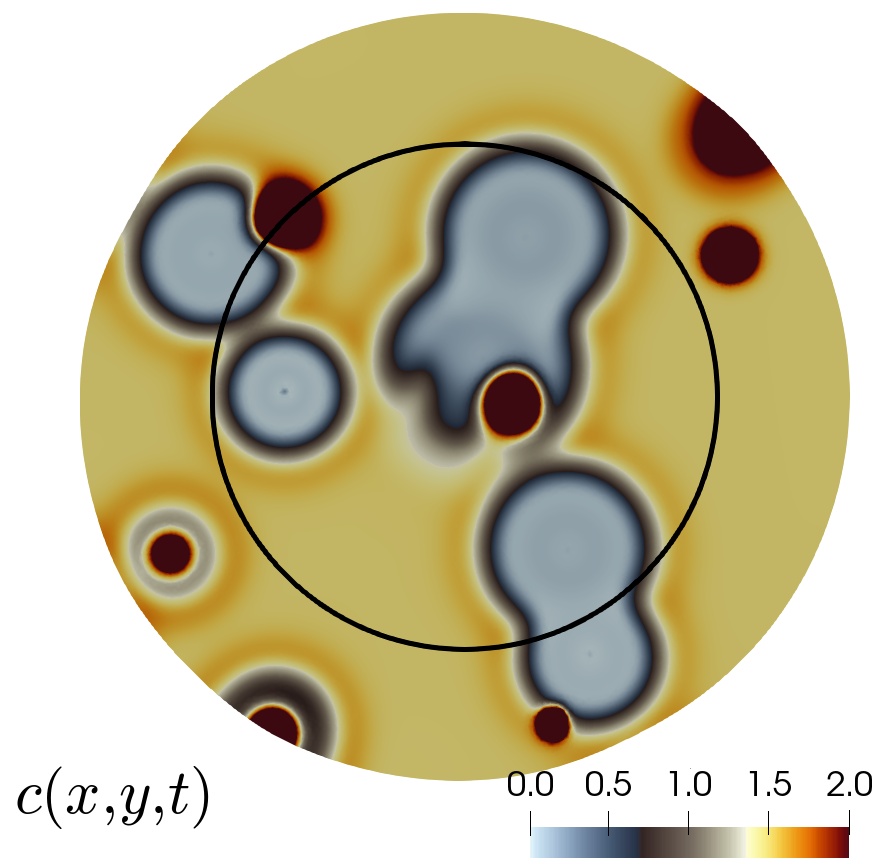}}
\subfigure[]{\includegraphics[width=0.24\textwidth]{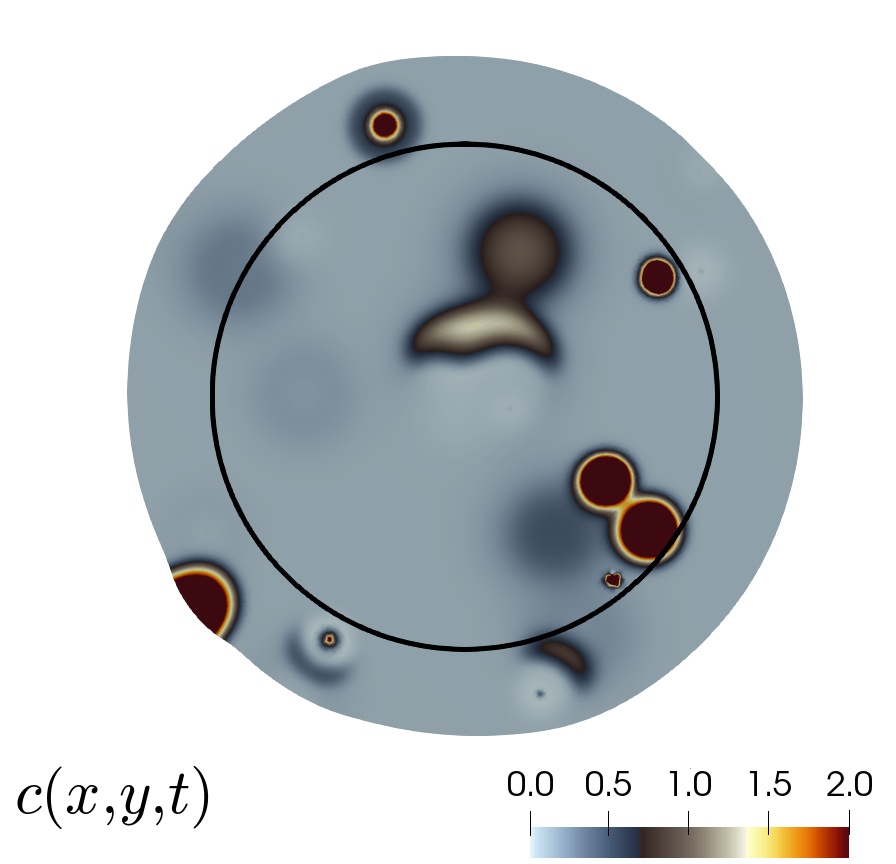}}
\subfigure[]{\includegraphics[width=0.24\textwidth]{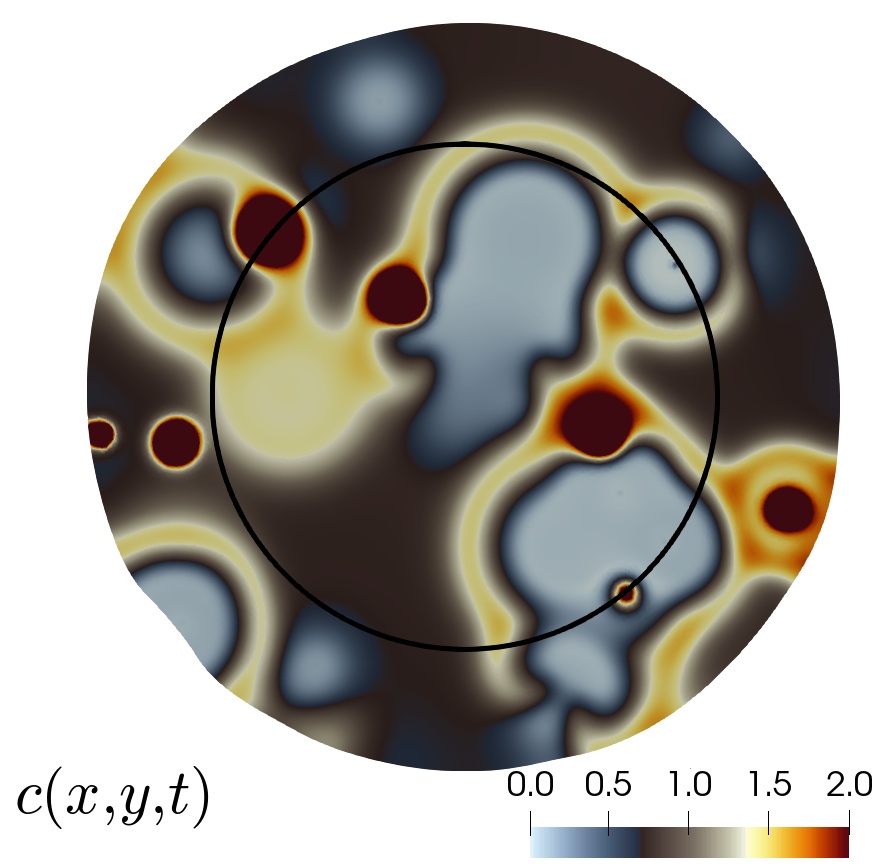}}\\
\subfigure[]{\includegraphics[width=0.24\textwidth]{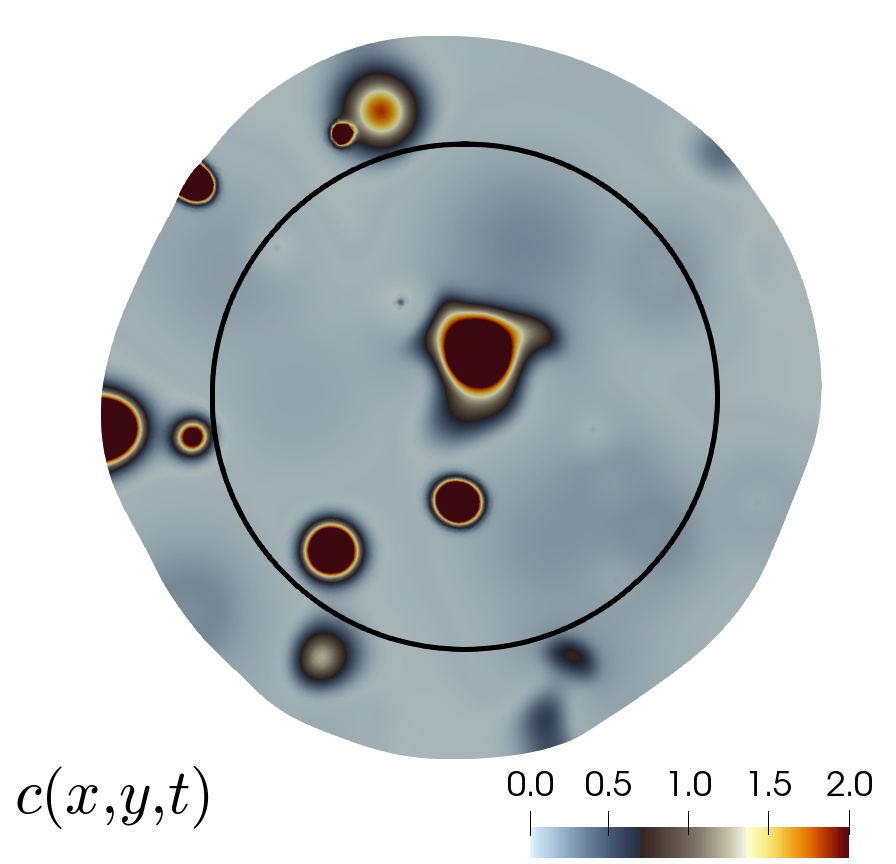}}
\subfigure[]{\includegraphics[width=0.24\textwidth]{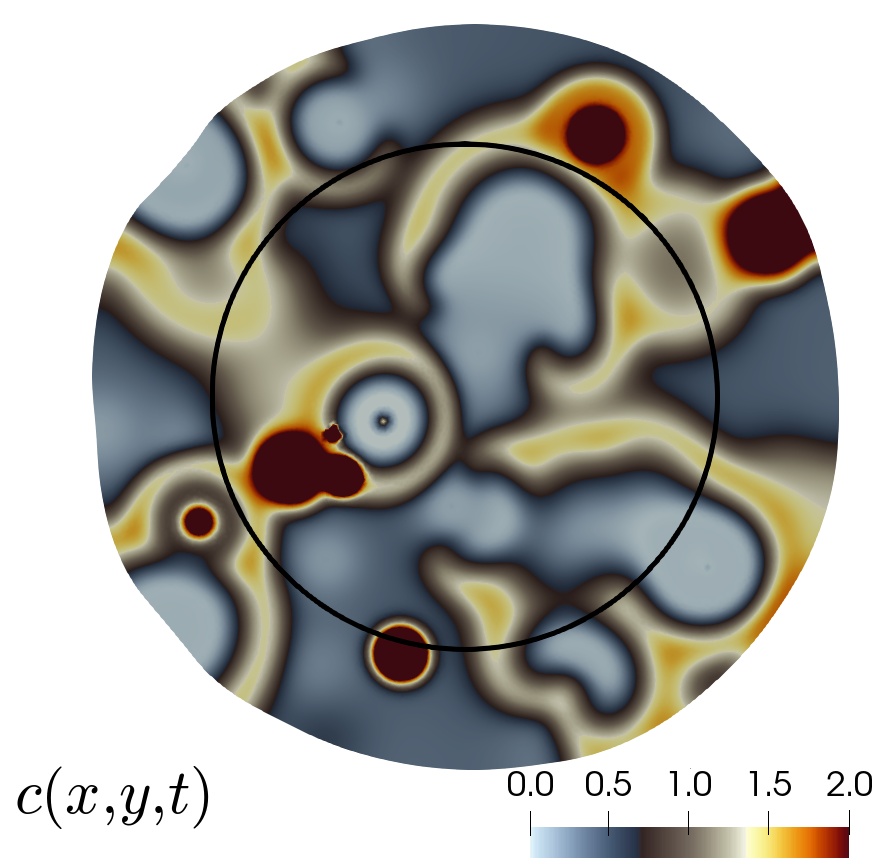}}
\subfigure[]{\includegraphics[width=0.24\textwidth]{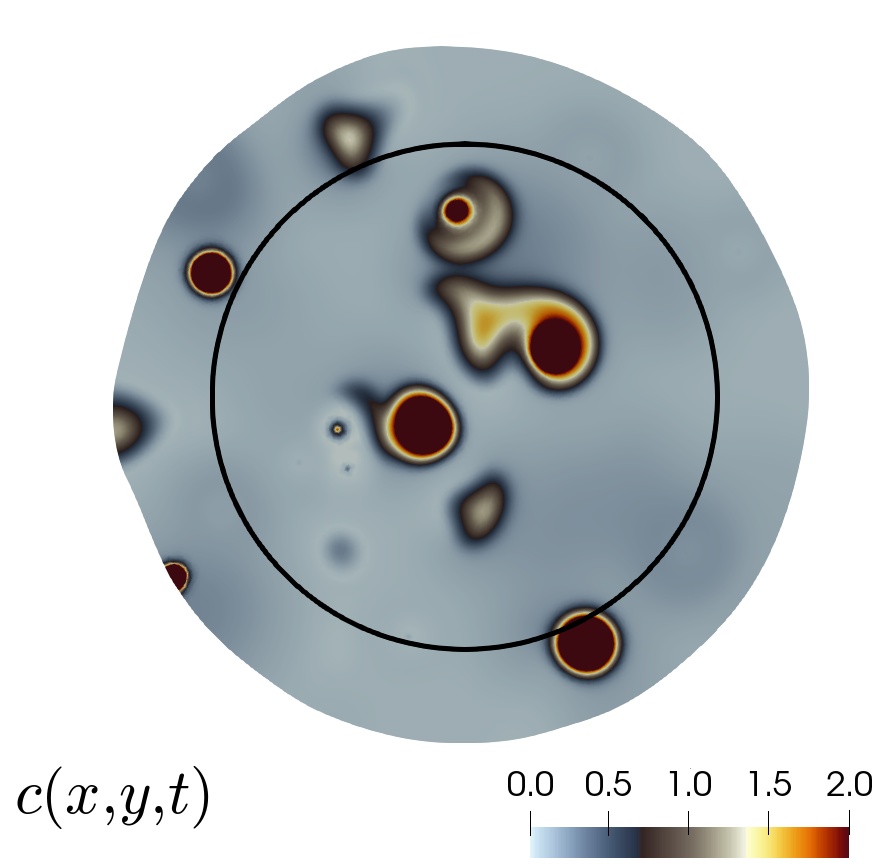}}
\subfigure[]{\includegraphics[width=0.24\textwidth]{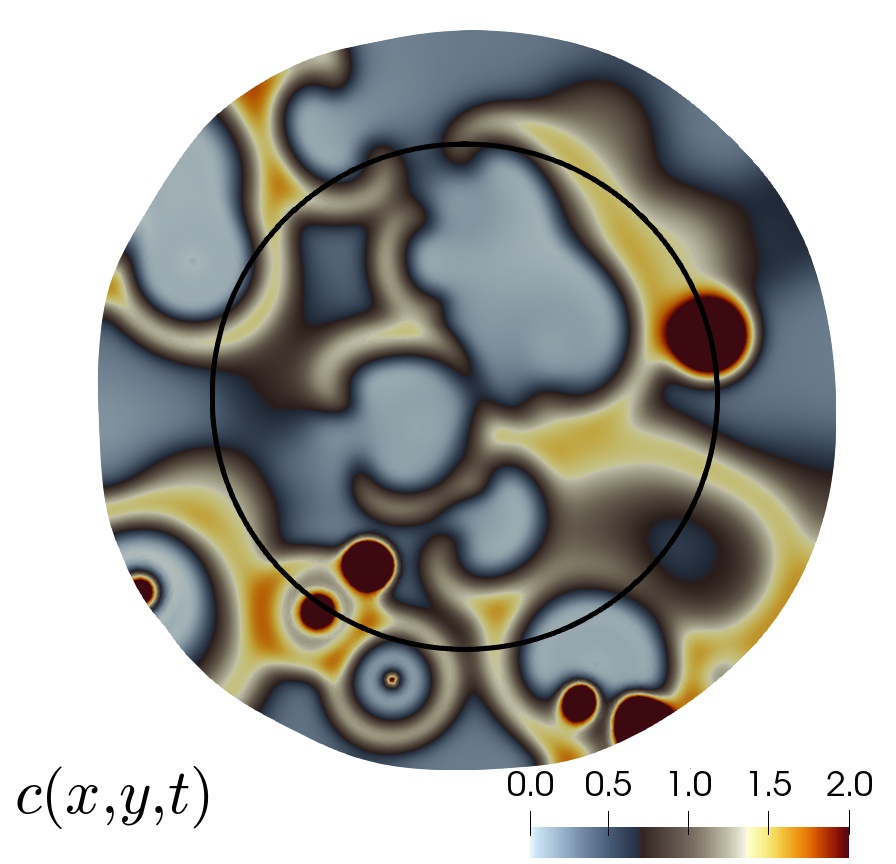}}
\end{center}

\caption{Test 3B. Transients of calcium concentration when sparks are initiated 
randomly in space and time, and being plotted on the deformed spatial domain, where the 
inner circle represents the boundary of the undeformed domain. Simulations obtained with 
$\mu=  0.288468$, $\lambda=0.1$, $c_s = 0.485044$.}\label{fig:test03Bdeform}
\end{figure}

Then we simulate for larger values of $\lambda$ and examine the differences in body deformation and calcium distribution. 
Samples of these computations are presented in Figure~\ref{fig:test03Bdeform}. 
We see that even with lower frequencies of calcium sparks (here taking one each 10 time steps) 
one is able to generate propagating fronts that in turn 
induce a periodic (but not uniform) dilation of the cell. In addition, if we regard the domain as a multi-cell aggregate instead as of 
a single cell, the short-time contractions and sparks could represent single-cell contractions whereas longer-term contractions and propagating fronts would account for synchronous cell movement and rearrangement and collective, larger inter-cellular calcium wavetrains. 
It is also clear that the deformations are no longer radial, which shows the benefit of incorporating 
multidimensional models.

\begin{figure}[!t]
\begin{center}
\subfigure[]{\includegraphics[width=0.4\textwidth]{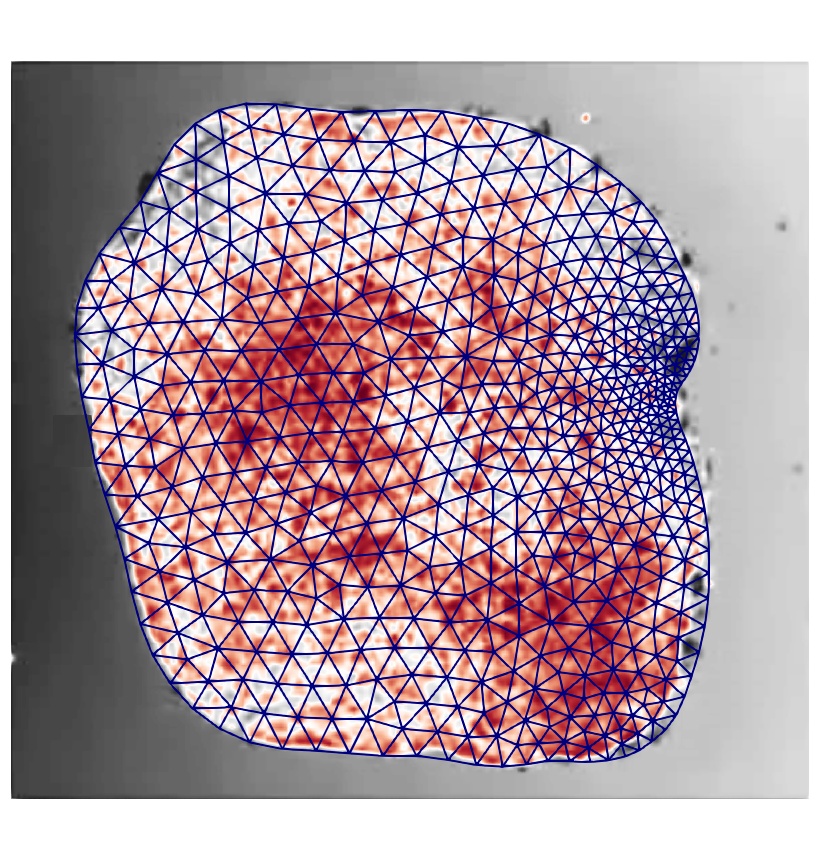}}\\
\subfigure[]{\includegraphics[width=0.32\textwidth]{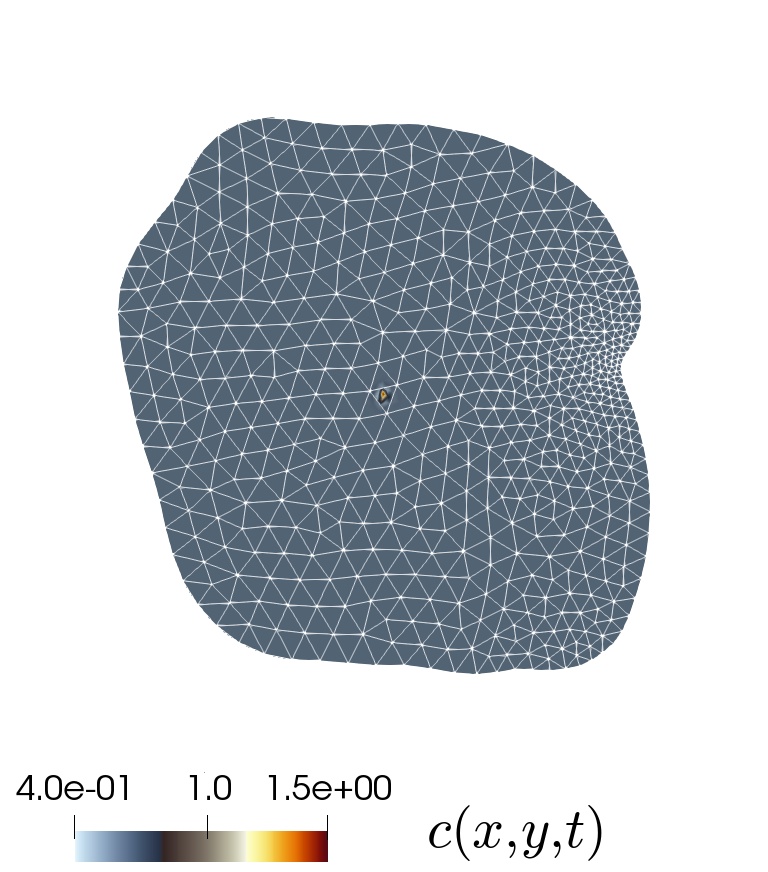}}
\subfigure[]{\includegraphics[width=0.32\textwidth]{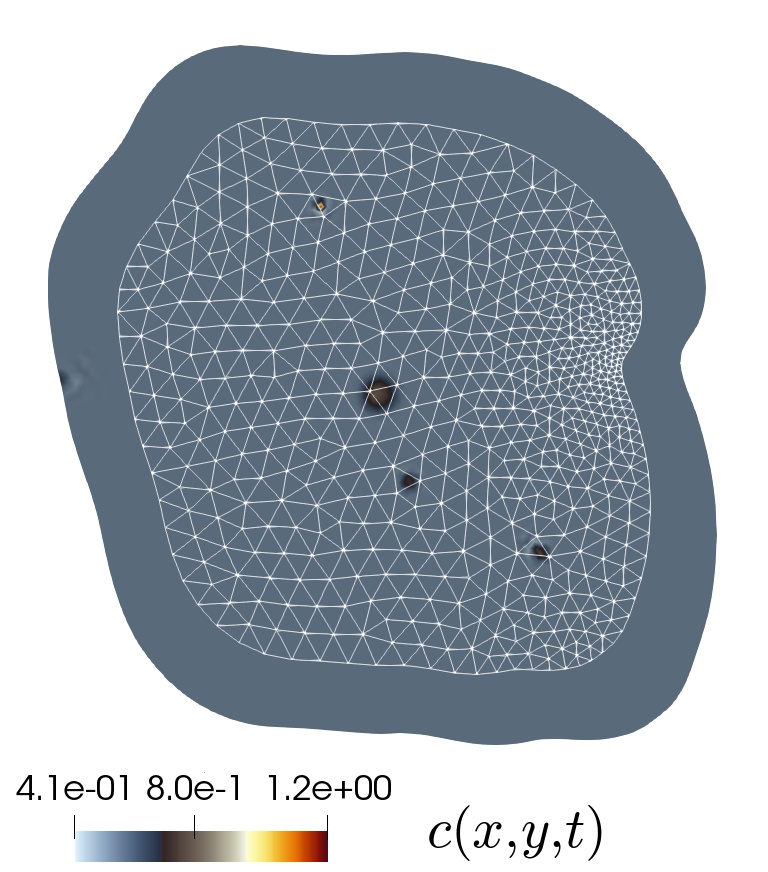}}
\subfigure[]{\includegraphics[width=0.32\textwidth]{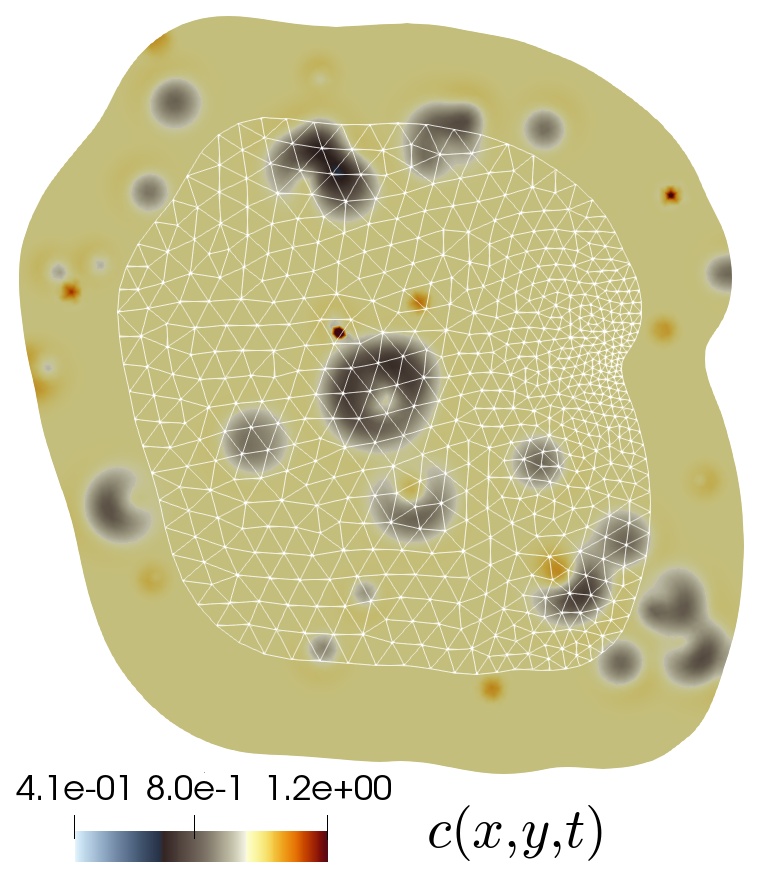}}\\
\subfigure[]{\includegraphics[width=0.32\textwidth]{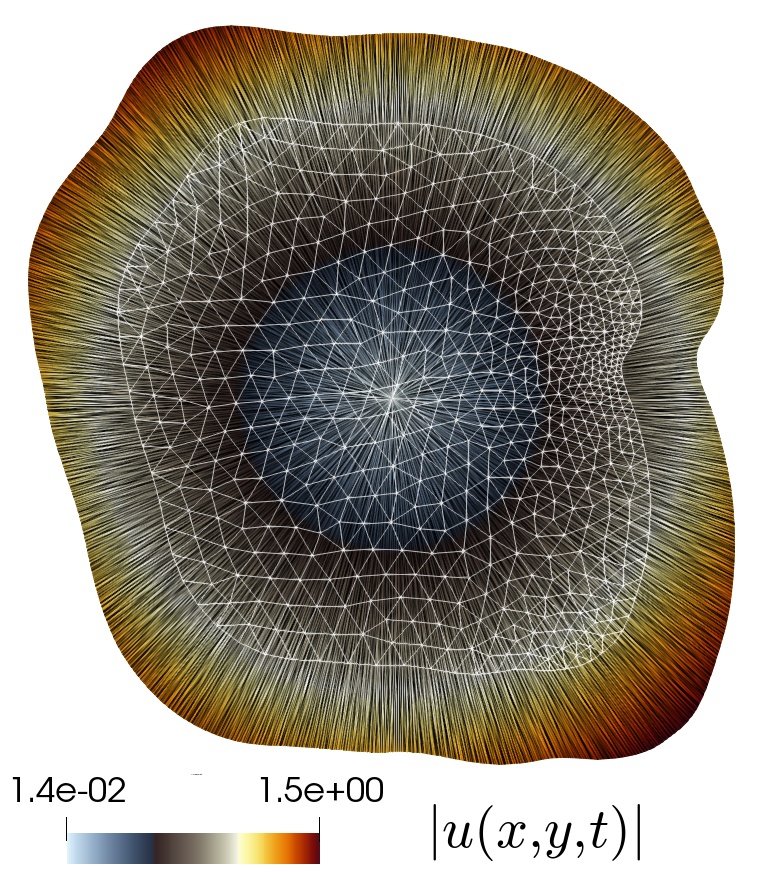}}
\subfigure[]{\includegraphics[width=0.32\textwidth]{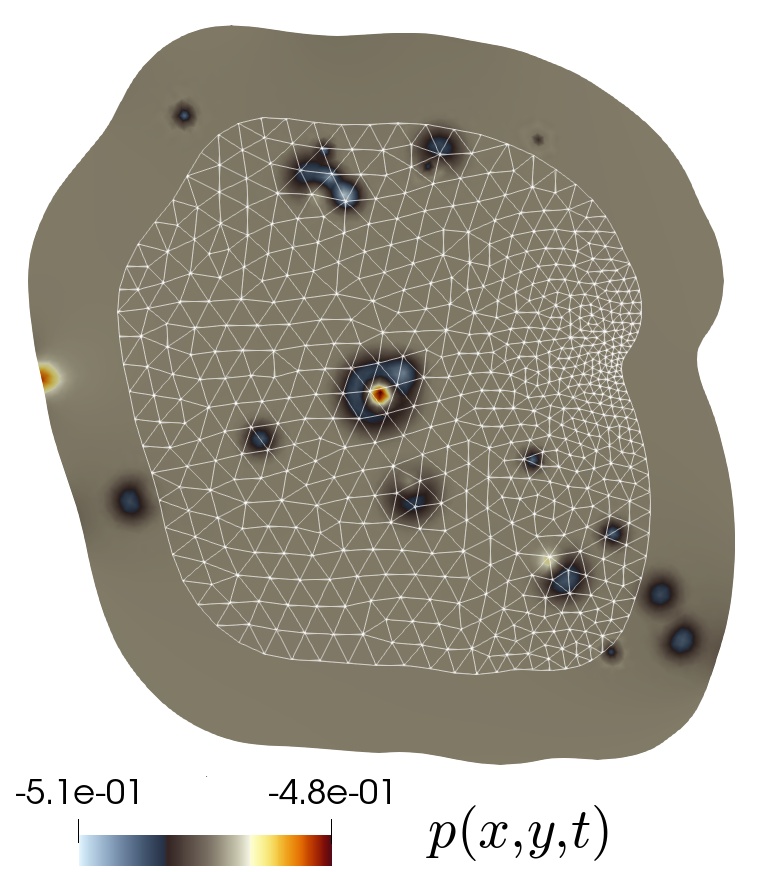}}
\subfigure[]{\includegraphics[width=0.32\textwidth]{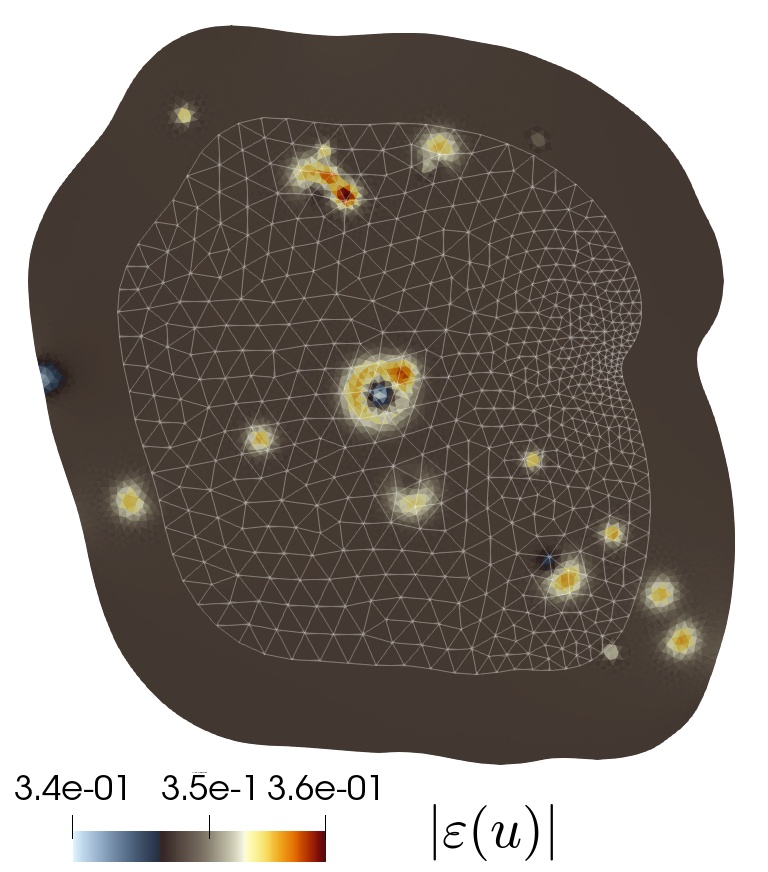}}
\end{center}

\caption{Test 3C. Sample of finite element mesh (a), and simulation of calcium sparks on a deforming sample of Xenopus 
embryonic tissue, shown at $t = 0.1, 3, 5$ (b,c,d). Panels (e,f,g) show displacement, pressure, and post-processed strain for $t = 10$. Results were generated using  
$\mu=  0.3$, $\lambda=0.1$, $c_s = 0.60386$.}\label{fig:test03Cembryo}
\end{figure}

We finish by replicating the behaviour of Test 3B in different geometries. A microscope image of a Xenopus embryonic tissue from \cite{kim14} is segmented and an unstructured triangular mesh is generated. The sparks are applied now with a lower intensity (75\% of $c_s$). 
In Figure~\ref{fig:test03Cembryo} we show a sample of coarse mesh and a few snapshots of the calcium distribution over the deformed domain. The tissue expands and contracts with an initial period of approximately 5 nondimensional units, and then the contraction increases in frequency. We also plot the displacement and pressure distribution, as well as the magnitude of the strain tensor, at $t=10$. 
We also consider a thin cylinder of radius 2.5 and height 0.4. The boundary is split into the bottom disk $\Gamma$ where we impose the slip condition \eqref{bc:Gamma} and on the remainder of the boundary we prescribe zero traction \eqref{bc:Sigma}. Orthogonality with respect to the space \eqref{RM} is therefore not required in this case. These mechanical boundary conditions represent a tissue that slips on a substrate, and in order to prevent it from slipping away we also prescribe zero displacement on the origin. 
Apart from a larger diffusion $D^\star= 0.02$ and a higher Poisson ratio $\nu = 0.45$, all other mechanochemical parameters are maintained as before. We compare qualitatively the deformation patterns with respect to imposing pure traction boundary 
conditions as before and show the outcome in Figure~\ref{fig:test03cylinder}. The former boundary treatment leads to mechanochemical patterns similar to those encountered in some of the 2D cases reported in Test 3B, whereas for the latter boundary conditions the deformations show a slightly more pronounced displacement in the $z$-direction.

\begin{figure}[!t]
\begin{center}
\subfigure[]{\includegraphics[width=0.24\textwidth]{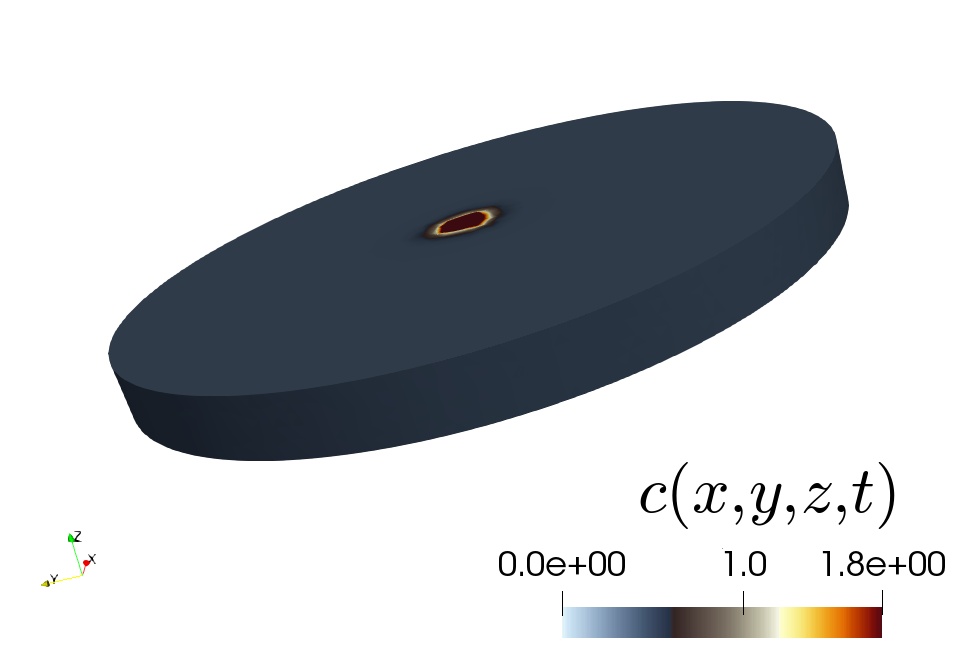}}
\subfigure[]{\includegraphics[width=0.24\textwidth]{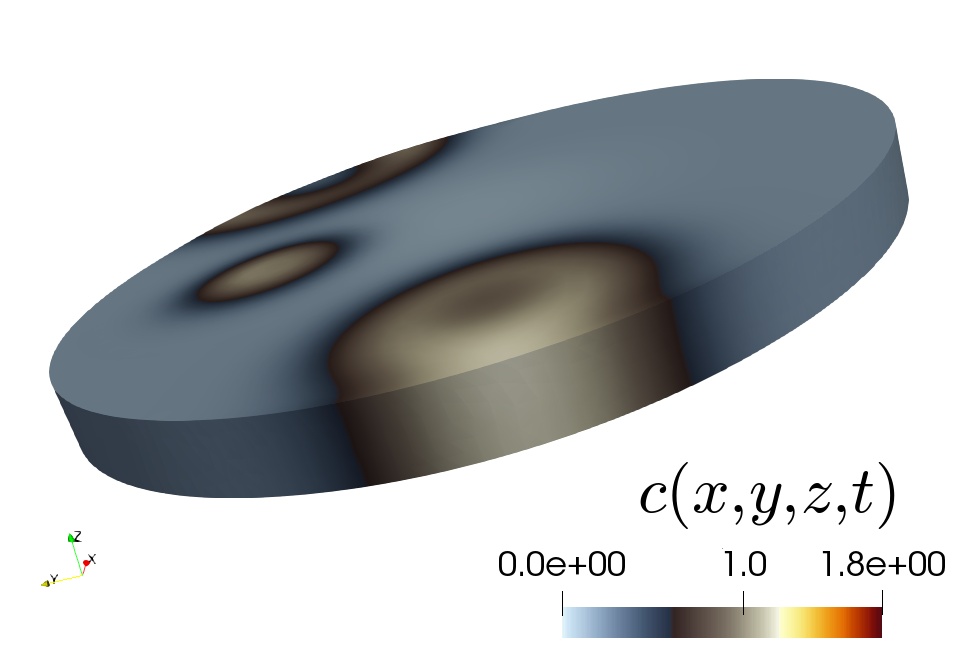}}
\subfigure[]{\includegraphics[width=0.24\textwidth]{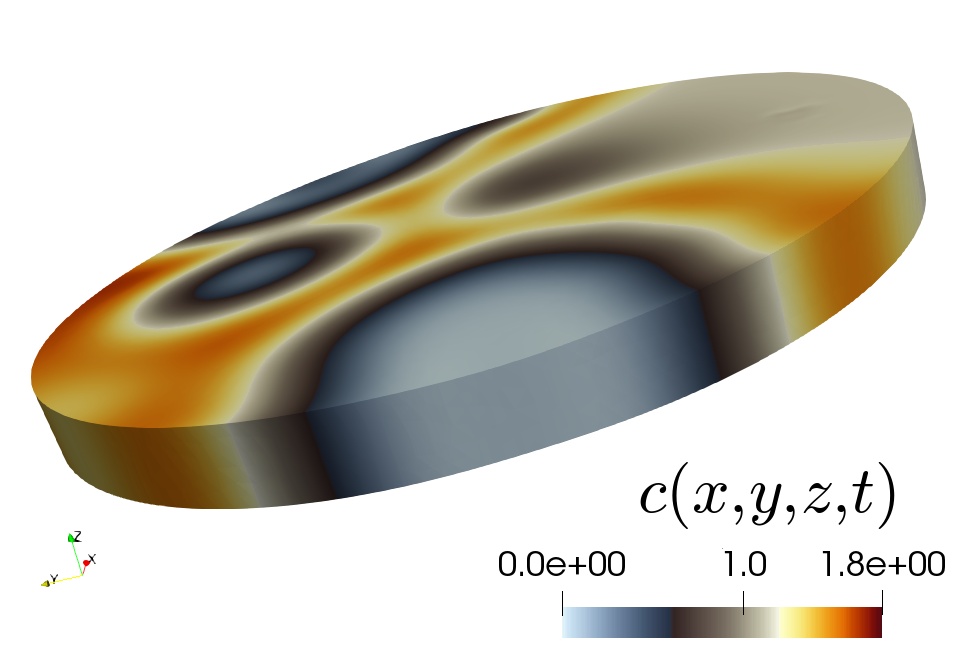}}
\subfigure[]{\includegraphics[width=0.24\textwidth]{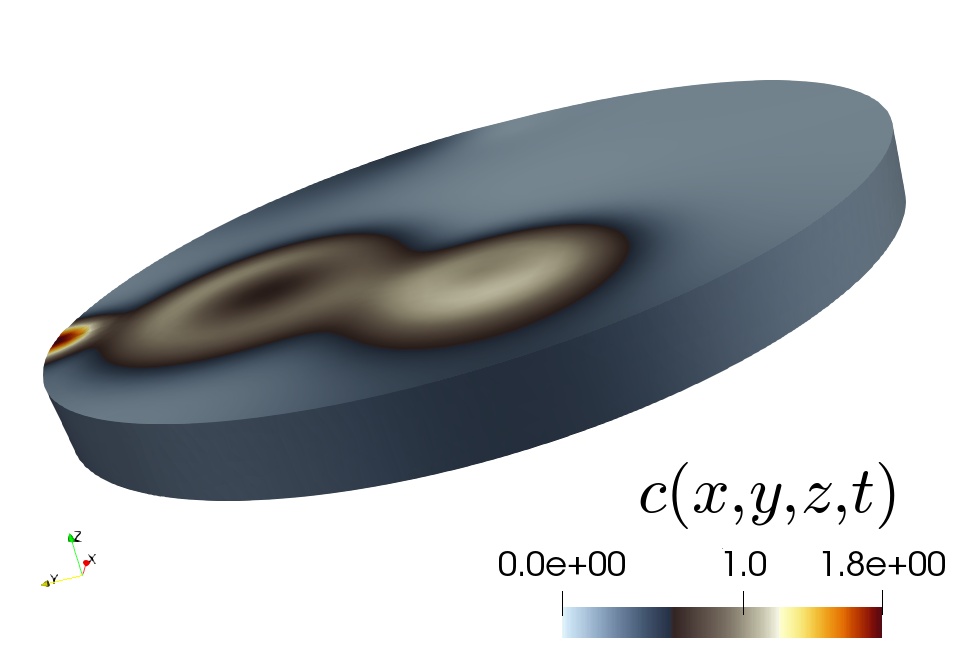}}\\
\subfigure[]{\includegraphics[width=0.24\textwidth]{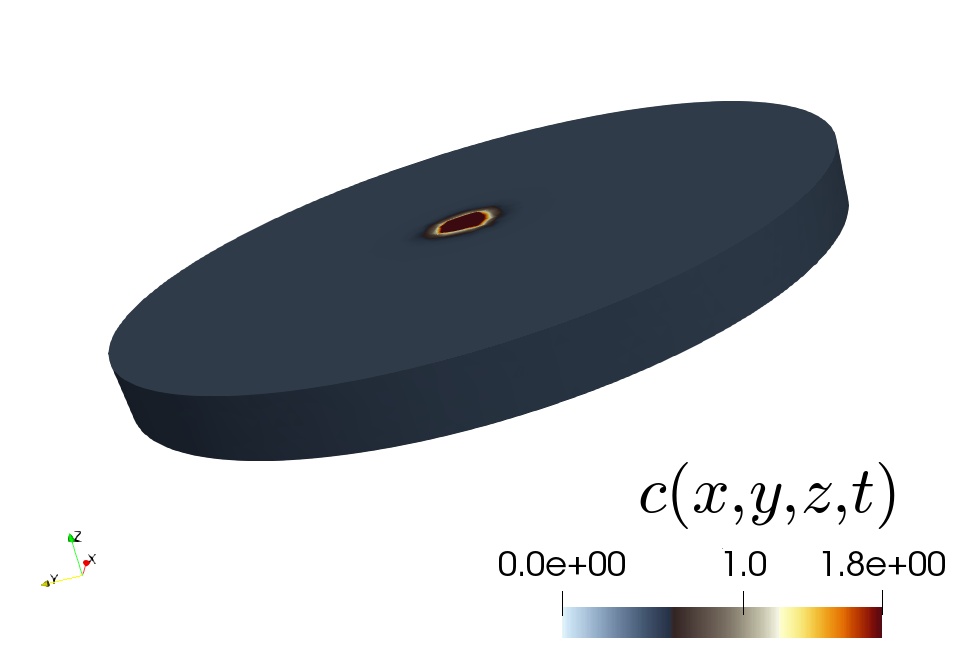}}
\subfigure[]{\includegraphics[width=0.24\textwidth]{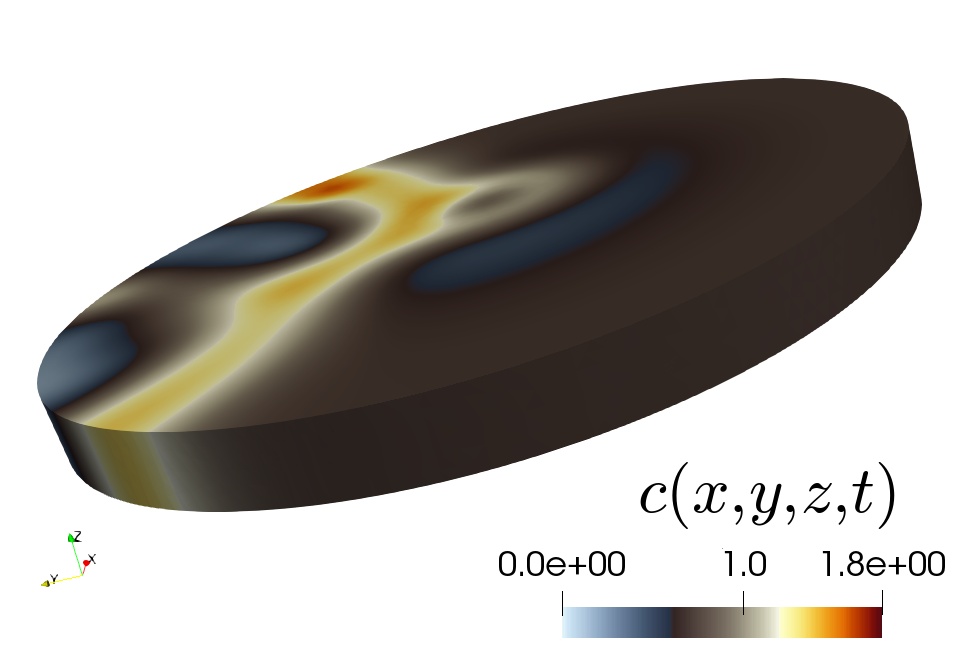}}
\subfigure[]{\includegraphics[width=0.24\textwidth]{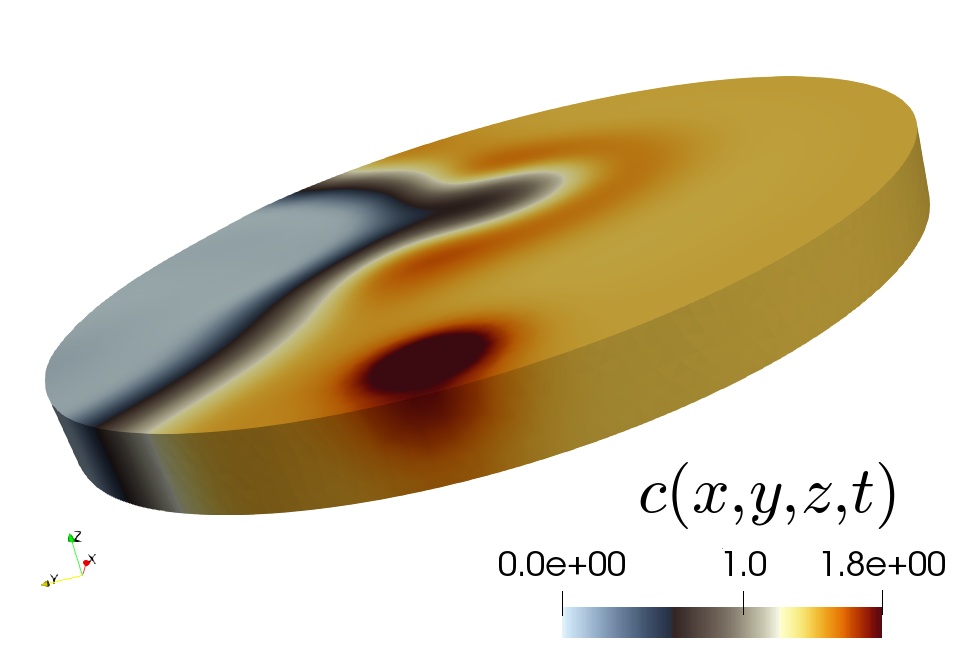}}
\subfigure[]{\includegraphics[width=0.24\textwidth]{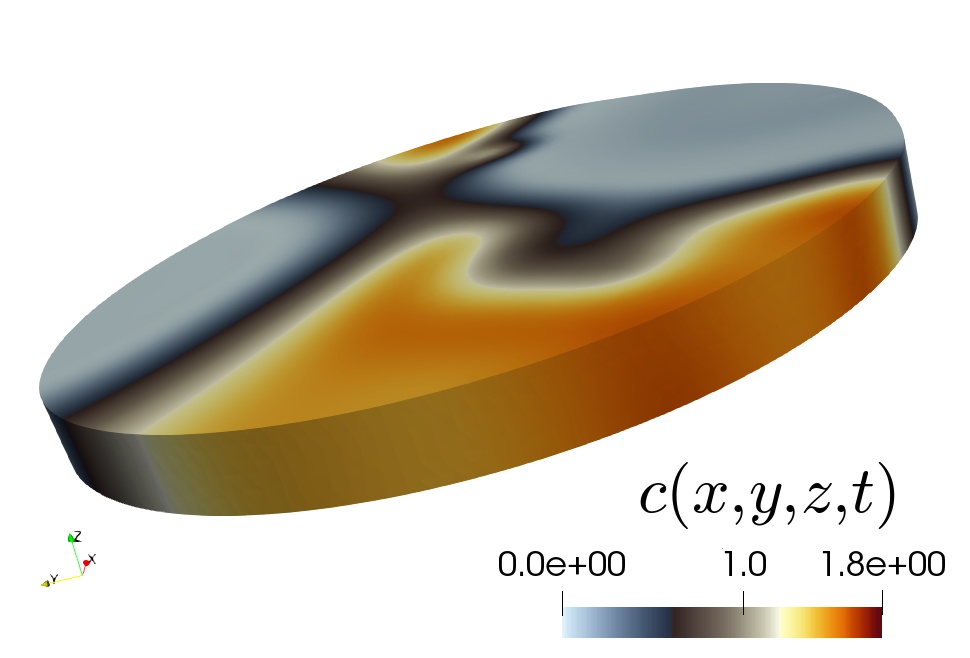}}
\end{center}

\caption{Test 3D. Computation of the mechanochemical coupling in a cylinder geometry using: 
mixed displacement-traction boundary conditions and representing a tissue on a substrate (a,b,c,d), 
and pure traction boundary conditions (e,f,g,h). Results were generated using $\nu =0.45$, $\mu=  0.3$, $\lambda=0.1$, $c_s = 0.60386$.}\label{fig:test03cylinder}
\end{figure}

%

\section{Concluding remarks}\label{sec:concl}
In this work we have extended to higher spatial dimensions the mechanochemical model in \cite{kaouri19} representing a two-way coupling between calcium signalling and mechanics in one dimension. We have
taken special care of traction boundary conditions for the motion of the cell. The governing equations couple 
 an advection-reaction-diffusion system for calcium signalling with a linear viscoelastic material. We include 
the two-way coupling mechanism through volume-dependent terms and an active contraction stress which is dependent on calcium dynamics as in \cite{murray84, kaouri19}. We have also established the existence of weak solutions, and we have proposed a finite element discretisation. 

Our model enables us to gather detailed insights in the mechanochemical 
processes during embryogenesis. More precisely, a quantitative evaluation of the effects of 
viscoelasticity and local dilation as mechanochemical processes underlying cytosolic calcium waves was undertaken. 

Regarding future directions, we note that in \cite{kim14} the authors suggest that  cells detect stress rather than stretch or strain because one does not 
observe transient stretch as contractions spread from cells exposed to ATP. This indicates that probably an 
appropriate formalism for further elucidating calcium signalling coupled to mechanics of cells and tissues is the framework of stress-assisted diffusion models, recently 
proposed in \cite{cherubini17,loppini18,propp20}. This relates also to the observation that in other types of cells that proliferate at a rate dependent on the substrate stiffness \cite{ghosh07}. 
Furthermore, at certain stages of morphogenesis, the assumption of small strains is no longer adequate and one needs to describe the motion through nonlinear elasticity. One could explore growth models using multiplicative decompositions of deformation gradients \cite{jones12,pandolfi17}, which will strongly depend on the type of phenomenon under consideration.  
Finally, it would be of interest to study the contraction and elongation of cell-cell gap junctions on the apical end of the tissue, since these regions contain a clustering of epithelia, and an important part of the overall morphogenesis occurs therein.   

The two-way coupling between mechanical forces and calcium signalling at both individual and collective cell levels 
is of course not unique to embryogenesis--it is a phenomenon shared by many other biological systems such as cancer cells, wound healing, keratinocytes, or white blood cells \cite{friedl,kobayashi14,yao16}. 
Modifications to the theoretical and computational tools presented here could be used to study these other processes too. 

Work is underway to extend the present model and method formulation to simulate the contact of the cell with a surrounding fluid (using an immersed boundary finite element method with distributed Lagrange multipliers) 
as well as the cell-to-cell interactions using a virtual element discretisation capturing more effectively the  
single cell geometries and the boundary contraction and elongation that together with junctional tension comprise the tissue-level deformation \cite{hara17}. Moreover, since calcium signalling evolves rapidly it would be natural to explore 
time-adaptive schemes based on local error estimators such as those advanced in \cite{deoliveira19} for a similar application.

\bigskip 
\noindent\textbf{Acknowledgement.} This work has been partially supported by the HPC-Europa3 Transnational Access Grant HPC175QA9K, and by the Monash Mathematics Research Fund S05802-3951284.

\small


\end{document}